\documentclass[a4paper,12pt,twoside]{article}

\usepackage{amsmath,amssymb,amsfonts,amsthm}
\usepackage{libertine}                      
\usepackage{setspace}                       
\usepackage{natbib}                        
\usepackage{refcount}                      
\usepackage{hyperref}                      
  \hypersetup{bookmarksopen, hidelinks}
\usepackage{graphicx}                     
\usepackage{color}                         
\IfFileExists{bbm.sty}{%
  \usepackage{bbm}%
}{}
\providecommand{\mathbbm}[1]{\mathbb{#1}}       
\usepackage{authblk}   
\usepackage{fancyhdr}  

\usepackage{inconsolata} 

\newtheorem{theorem}{Theorem}
\newtheorem{prop}[theorem]{Proposition}
\newtheorem{lemma}{Lemma}

\theoremstyle{definition}
\newtheorem{definition}{Definition}

\newtheorem{example}{Example}

\newtheorem{proofthmInternal}{\protect\customproofname}
\newcommand{\customproofname}{\normalfont\itshape Proof of Theorem}
\newcounter{customproofcounter} 


\newcommand{\proofof}[1]{Proof of #1}

\newcommand{\spacednabla}[2]{\nabla_{\!#1}\,{#2}}

\newcommand{\spacednablatranspose}[2]{\nabla_{\!#1}^T\,{#2}}
\newcommand{\grad}{\nabla}

\newcommand{\R}{\mathbb{R}}
\newcommand{\set}[1]{\left\{#1\right\}}
\newcommand{\indicatorofset}[1]{\mathbbm{1}_{#1}}
\newcommand{\indicator}[1]{\indicatorofset{\set{#1}}}

\newcommand{\secref}[1]{\hyperref[#1]{Section~\ref*{#1}}}
\newcommand{\figref}[1]{\hyperref[#1]{Fig~\ref*{#1}}}
\newcommand{\supsecref}[1]{\hyperref[#1]{Section~\ref*{#1} of the Supplementary Material}}

\newcommand{\smallmat}[1]{\left(\begin{smallmatrix}#1\end{smallmatrix}\right)}

\def\keywordfont{\fontsize{9}{10}\selectfont}
\newenvironment{keywords}
  {\par\addvspace{8pt}%
   \keywordfont\noindent{\itshape Key words}:\ \ignorespaces}
  {\par\addvspace{23pt}}

\newcommand{\figuresize}[1]{\def\figscale{#1}}
\newcommand{\figurebox}[3][]{%
  \includegraphics[scale=\figscale]{#1}%
}

\makeatletter
\newcommand{\customlabel}[2]{%
    \edef\@currentlabel{#2}%
    \label{#1}%
}
\makeatother
\newcommand{\customref}[1]{\ref{#1}}

\newcounter{lemmasubitem}[lemma] 
\renewcommand{\thelemmasubitem}{\alph{lemmasubitem}}
\newcommand{\lemmasubitem}[1]{%
  \refstepcounter{lemmasubitem}(\thelemmasubitem)\label{#1}%
}
\newcommand{\lemmaitemref}[2]{\hyperref[#2]{Lemma~\ref{#1}(\ref{#2})}}

\makeatletter
\renewcommand\@fnsymbol[1]{%
  \ifcase#1
  \or {\Large\textasteriskcentered} 
  \or \dagger
  \or \ddagger
  \else \@ctrerr
  \fi
}
\makeatother

\def\and{and\ }

\newcommand{\beginAppendixCounter}{
    \setcounter{lemma}{0}
    \renewcommand{\thelemma}{A\arabic{lemma}}
    \renewcommand{\theHlemma}{appendix.lemma.\arabic{lemma}}
    \renewcommand{\theHlemmasubitem}{appendix.lemmasubitem.\arabic{lemma}.\arabic{lemmasubitem}}
    \setcounter{equation}{0}
    \renewcommand{\theequation}{A\arabic{equation}}
    \renewcommand{\theHequation}{appendix.equation.\arabic{equation}}
    \setcounter{figure}{0}
    \renewcommand{\theHfigure}{appendix.figure.\arabic{figure}}
}
\newcommand{\beginSupplementCounter}{
    \setcounter{equation}{0}
    \renewcommand{\theequation}{S\arabic{equation}}
    \renewcommand{\theHequation}{supplement.equation.\arabic{equation}}
    \setcounter{figure}{0}
    \renewcommand{\thefigure}{S\arabic{figure}}
    \renewcommand{\theHfigure}{supplement.figure.\arabic{figure}}
    \setcounter{table}{0}
    \renewcommand{\thetable}{S\arabic{table}}
    \renewcommand{\theHtable}{supplement.table.\arabic{table}}
    \setcounter{section}{0}
    \renewcommand{\thesection}{\arabic{section}}
    \renewcommand{\theHsection}{supplement.section.\arabic{section}}
    \setcounter{lemma}{0}
    \renewcommand{\thelemma}{S\arabic{lemma}}
    \renewcommand{\theHlemma}{supplement.lemma.\arabic{lemma}}
    \renewcommand{\theHlemmasubitem}{supplement.lemmasubitem.\arabic{lemma}.\arabic{lemmasubitem}}
}

\begin{document}

\title{What is the price of approximation? The saddlepoint approximation to a likelihood function}

\author[1]{Godrick Oketch}
\author[2]{Rachel M. Fewster}
\author[2]{Jesse Goodman\thanks{Email: jesse.goodman@auckland.ac.nz
    \newline The authors gratefully acknowledge funding support from the Royal Society of New Zealand Marsden fund.}}
\affil[1]{School of Public Health, University of California, Berkeley, California, U.S.A.}
\affil[2]{Department of Statistics, University of Auckland, Private Bag 92019, Auckland, New Zealand}
\date{}

\maketitle

\thispagestyle{empty}
\begin{abstract}
The saddlepoint approximation to the likelihood, and its corresponding maximum likelihood estimate (MLE), offer an alternative estimation method when the true likelihood is intractable or computationally expensive. 
However, maximizing this approximate likelihood instead of the true likelihood inevitably comes at a price: a discrepancy between the MLE derived from the saddlepoint approximation and the true MLE. 
In previous studies, the size of this discrepancy has been investigated via simulation, or by engaging with the true likelihood despite its computational difficulties. 
Here, we introduce an explicit and computable approximation for the discrepancy, through which the adequacy of the saddlepoint-based MLE can be directly assessed even when the true MLE is unavailable.
We present examples demonstrating the approximation's accuracy in simple cases where the true likelihood can be calculated. 
Additionally, we present asymptotic results that capture the behaviour of the discrepancy in a suitable limiting framework.
\end{abstract}


\vspace*{-13mm}
\section{Introduction}

The saddlepoint approximation continues to attract research interest due to its impressive accuracy in approximating probability density and mass functions.
The approximation was first introduced to statistical applications by \cite{Daniels1954}. 
It is typically used in its first-order formulation, but the availability of higher-order formulations facilitates the study of approximation accuracy.
Some significant contributions to this development include works by \cite{LugannaniAndRice1980}, \cite{mccullagh1987tensor}, and \cite{Barndorff1989}.
Additionally, \cite{BarndorffNielsenCox1979} and \cite{Reid1988} offer detailed reviews of the historical evolution and broad applications of the saddlepoint approximation.

Here, we focus on the use of the saddlepoint approximation as a method for estimation.
When the true likelihood function is unknown or computationally intractable, several authors have used likelihoods based on the saddlepoint approximation to obtain maximum likelihood estimates (MLEs). We refer to such estimates as \emph{saddlepoint MLEs}.
\cite{Pedeli2015} used saddlepoint MLEs to overcome inefficiencies in existing methods for fitting integer-valued autoregressive models.
\cite{Zhang2019, Zhang2021, Zhang2022} used the approach for models in ecology and epidemiology involving linear transformations of multinomial random variables, and noted considerable gains in computation speed relative to alternative Bayesian approaches.
\cite{Davison2020} explored the use of 
classical and conditionally adjusted saddlepoint approximations for modelling linear birth-death branching processes, again highlighting the computational efficiencies achieved for large population sizes.

When using saddlepoint approximations for estimation, a key question is whether the resulting saddlepoint MLEs  are acceptably close to the \emph{true MLEs}—those that would be obtained from the true, but potentially uncomputable, likelihood.
We refer to the difference between the two sets of MLEs as the \emph{discrepancy}. 

In cases where the true likelihood can be computed, the discrepancy can be obtained by direct comparison.
For instance, \cite{Zhang2019} showed that the saddlepoint MLEs were virtually indistinguishable from the true MLEs for a misidentification model known as  $M_{t, \alpha}$, a comparison made possible due to an exact likelihood formulation derived for this model by \citet{Vale2014}.
\citet{Davison2020} also presented a direct comparison between the true and saddlepoint MLEs for their branching process model.

However, the saddlepoint approximation is often employed precisely because the true MLEs are uncomputable. 
In such cases, we require an assessment of the size of the discrepancy in order to justify the use of the saddlepoint MLE instead of the unknown true MLE.
Here, we develop a general approximation to the discrepancy that can be computed explicitly even when the true likelihood is unavailable.
We show that it achieves high accuracy using models for which the true likelihood and MLE can be computed directly,
and we refer the reader to \citeauthor{oketch2025} for a range of examples whose true likelihoods are unavailable, including the two-source model (\citeyear[section 5.4]{oketch2025}). 
We also provide theoretical justification via the limiting framework considered in \citet{Goodman2022}.
An R code package \citep{oketch2025} allows the user to compute the discrepancy approximation automatically along with the saddlepoint MLE.
The discrepancy approximation therefore permits a straightforward assessment of the appropriateness of saddlepoint MLE techniques, without requiring a model-specific analysis.

\vspace*{-6mm}
\section{The discrepancy approximation}\label{secDisc:SetupNotationDefinitions}
\subsection{Notation for generating functions and derivatives}
\label{secDisc:Notations}

Let $X$ be an observable $d$-dimensional random variable.
We write $X = X_{\theta}$ to indicate the dependence of $X$ on a parameter vector $\theta$.  
We treat both $X_{\theta}$ and $\theta$ as column vectors (i.e., $X_{\theta} \in \R^{d \times 1}$ and $\theta \in \R^{p \times 1}$).
Define the moment generating function (MGF) and the cumulant generating function (CGF) of $X$ by
$M_X(t;\theta) = E(e^{tX_{\theta}})$ and $K_X(t;\theta) = \log M_X(t;\theta)$, respectively. 
Here, $t \in \R^{1 \times d}$ is a row vector of the same dimension as $X$, and we 
assume that $M_X(t;\theta)$ is defined in some open neighbourhood of $\R^{1\times d}\times\R^{p\times 1}$ containing some point $(t,\theta)=(0,\theta_0)$, and is twice continuously differentiable in both $t$ and $\theta$ in this neighbourhood.

To specify derivatives, we adopt the following conventions.
Gradient vectors with respect to $\theta$ are denoted $\nabla_{\theta}$ and treated as vectors with the shape of $\theta^T$.
Thus, for scalar-valued $f(\theta,t)$, we interpret $\nabla_\theta f(\theta,t) \in \mathbb{R}^{1\times p}$ as a row vector and its transpose
$\nabla_\theta^T f(\theta,t) \in \mathbb{R}^{p\times 1}$ as a column vector. 
Likewise, for vector-valued $g(\theta, t) \in \mathbb{R}^{d\times 1}$, we define
$\spacednabla{\theta}{g(\theta, t)} \in \mathbb{R}^{d\times p}$.
For scalar-valued $f(\theta,t)$, we write the Hessian with respect to $\theta$ as $\nabla_{\theta}^T \nabla_{\theta} f(\theta,t)$, with $i,j$ entry 
$\partial^2 f/ \partial \theta_i \partial \theta_j$.

The gradient and Hessian of the CGF with respect to $t$, denoted $K_X'(t;\theta) \in \mathbb{R}^{d \times 1}$ and $K_X''(t;\theta) \in \mathbb{R}^{d \times d}$ respectively, are directly related to the first and second moments of the random vector $X_{\theta}$ via $ E(X_\theta) = K_X'(0;\theta)$ and  $\mathrm{cov}(X_\theta,X_\theta) = K_X''(0;\theta).$
We assume throughout that $\mathrm{cov}(X_\theta,X_\theta)$ is non-singular for all $\theta$. 
Since $K_X''(0;\theta)$ is then positive definite for each $\theta$, and $K_X''(t;\theta)$ depends continuously on $(t,\theta)$,
it follows by restricting to a suitably small open set around $(0,\theta)$ that
\vspace*{-2mm}
\begin{equation}
\label{eqDisc:NonSingularCovAndMean}
K_X''(t;\theta)~\text{is positive definite for every $(t,\theta)$ in the specified open neighbourhood.}
\end{equation}

\subsection{True discrepancy and discrepancy approximation}

Suppose $x$ is an observation of the random variable $X=X_\theta \in \mathbb{R}^{d \times 1}$, and denote the true likelihood function by
$
    L(\theta; x).
$
The saddlepoint approximation to this likelihood is defined
as
\begin{equation}\label{eqDisc:SaddlepointLikelihood_Defn}
    \hat{L}(\theta; x) = 
    \exp(K_X(\hat{t};\theta)-\hat{t}x)
    \{\det(2\pi K_X''(\hat{t};\theta) )\}^{-1/2},
\end{equation}
where $\hat{t} = \hat{t}(\theta;x) \in \mathbb{R}^{1 \times d}$ is the solution to the \textit{saddlepoint equation}: 
\begin{equation}\label{eqDisc:saddlepoint_eqn}
    K_X'(\hat{t};\theta) = x.
\end{equation}
Since the CGF is twice continuously differentiable and convex, the saddlepoint equation has a solution  
$\hat{t}(\theta;x)$ for values of $(\theta, x)$ within an open neighbourhood about $(\theta_0, E(X_{\theta_0}))$.
We refer to \eqref{eqDisc:SaddlepointLikelihood_Defn} as the \textit{saddlepoint likelihood}.

We can equivalently write \eqref{eqDisc:SaddlepointLikelihood_Defn} as the saddlepoint log-likelihood function
\begin{align}
\log \hat{L}(\theta; x) &= K_{X}(\hat{t};\theta) - \hat{t}x - \tfrac{d}{2} \log(2\pi) - \tfrac{1}{2} \log \det { K_{X}''(\hat{t};\theta) }, \label{eqDisc:SaddlepointLogLikelihood} 
\end{align}
where $\hat{t} = \hat{t}(\theta;x)$  is the solution to the saddlepoint equation.
\begin{definition}[True discrepancy]
\label{labelDisc:discrepancy}
Let 
$\hat{\theta}_{\mathrm{true}} = {\mathrm{argmax}}_{\theta}\, \log L(\theta;x)$
 be the true MLE, and $\hat{\theta}_{\mathrm{spa}} = {\mathrm{argmax}}_{\theta}\, \log \hat{L}(\theta;x)$ the saddlepoint MLE, if they each exist.
The true discrepancy is
\begin{equation*}\label{eqDisc:TrueDiscrepancy}
\delta = \hat{\theta}_{\mathrm{true}} - \hat{\theta}_{\mathrm{spa}}.
\end{equation*}
\end{definition}
The discrepancy expresses the approximation error in the scale of model parameter values.
That is, the discrepancy is the (unknown) difference in $\theta$ resulting from maximizing $\log \hat{L}(\theta;x)$ instead of $\log L(\theta;x)$.
To obtain a computable approximation to the true discrepancy, we consider the second-order saddlepoint log-likelihood,
\begin{equation}\label{eqDisc:SecondOrderSaddlepointLogLikelihood}
\log \hat{L}(\theta;x)_2 = \log \hat{L}(\theta;x) + T_X(\theta,x),
\end{equation}
where the correction term $T_X(\theta,x)$ is given by
\begin{equation}
\label{eqDisc:funcT}
\begin{aligned}
T_X(\theta, x)
&=
\frac{1}{8}
\sum_{j_1,j_2,j_3,j_4}^d
\frac{\partial^4 K_X(\hat t; \theta)}
{\partial t_{j_1}\partial t_{j_2}\partial t_{j_3}\partial t_{j_4}}
\hat Q_{j_1 j_2}\hat Q_{j_3 j_4} \\
&\quad
-
\frac{1}{8}
\sum_{i_1,i_2,i_3}^d
\sum_{j_1,j_2,j_3}^d
\frac{\partial^3 K_X(\hat t; \theta)}
{\partial t_{i_1}\partial t_{i_2}\partial t_{i_3}}
\frac{\partial^3 K_X(\hat t; \theta)}
{\partial t_{j_1}\partial t_{j_2}\partial t_{j_3}}
\hat Q_{i_1 i_2}\hat Q_{i_3 j_1}\hat Q_{j_2 j_3} \\
&\quad
-
\frac{1}{12}
\sum_{i_1,i_2,i_3}^d
\sum_{j_1,j_2,j_3}^d
\frac{\partial^3 K_X(\hat t; \theta)}
{\partial t_{i_1}\partial t_{i_2}\partial t_{i_3}}
\frac{\partial^3 K_X(\hat t; \theta)}
{\partial t_{j_1}\partial t_{j_2}\partial t_{j_3}}
\hat Q_{i_1 j_1}\hat Q_{i_2 j_2}\hat Q_{i_3 j_3}
.
\end{aligned}
\end{equation}
Here,
$\hat Q_{ij}$ denotes the $i,j$ entry of the inverse of the matrix $K''_X (\hat{t};\theta)$, i.e., $\hat Q_{ij} = (K''_X (\hat{t};\theta)^{-1})_{ij}$ and $\hat{t} = \hat{t}(\theta;x)$ is the solution of the saddlepoint equation. 
Unless otherwise indicated, all indices
$i_{*}$ and $j_{*}$ range over $1,\dots,d$.
For the derivation of \eqref{eqDisc:SecondOrderSaddlepointLogLikelihood} and \eqref{eqDisc:funcT}, and further extensions to higher-order terms, see \citet[Chapter~6]{mccullagh1987tensor}.

\begin{definition}[Discrepancy approximation]
\label{labelDisc:approxdiscrepancy} 
The discrepancy approximation is
\begin{equation}\label{eqDisc:DiscrepancyFormula}
\hat{\delta} = 
-\{
\nabla_{\theta}^T \nabla_{\theta} \log \hat{L}(\hat{\theta}_{\mathrm{spa}};x)
\}^{-1}
\spacednablatranspose{\theta}{T_X(\hat{\theta}_{\mathrm{spa}}, x)}
.
\end{equation}
\end{definition} 
This expression for $\hat{\delta}$ provides an approximation to the unknown true discrepancy $\delta = \hat{\theta}_{\mathrm{true}} - \hat{\theta}_{\mathrm{spa}}$, which is computable solely in terms of the saddlepoint MLE $\hat{\theta}_{\mathrm{spa}}$.
The intuition behind the formula \eqref{eqDisc:DiscrepancyFormula} is to approximate the true discrepancy by a single Newton step for adjusting $\hat{\theta}_{\mathrm{spa}}$ closer to the root of $\nabla_\theta \log  \hat{L}(\theta; x)_2=0$.

\subsection{Considerations for exponential families}

If the parametric model reduces to an exponential family for which $X_\theta$ is the sufficient statistic, then the saddlepoint likelihood is exact up to a constant factor \cite[Section 2.5.6]{Goodman2022}.
Thus the saddlepoint MLE will coincide with the true MLE, and the true discrepancy $\delta$ is identically zero.
We show in \supsecref{sec:supp:exactness-expfam} that the same is true for $\hat{\delta}$.
That is, when the underlying model is an exponential family with $X_\theta$ as its sufficient statistic, the discrepancy approximation correctly signals that the saddlepoint MLE is exact by returning the value $\hat{\delta}=0$ identically.

In practice, however, the presence of exponential family components in a model offers little assistance when assessing the discrepancy.
Commonly, the observed data vector $X_\theta$ arises from unobserved variables drawn from exponential families, but does not itself form an exponential family.
This is because saddlepoint methods are chiefly used in settings where operations on random variables preserve tractable CGFs but render the likelihood intractable \citep{oketch2025};
and such operations do not typically preserve the exponential-family property.
Most applications from the literature follow this pattern \citep[e.g.,][]{Pedeli2015,Zhang2019,Zhang2021,Zhang2022,Davison2020}, as do all the examples in \secref{secDisc:Examples}.
In such cases, the discrepancy $\delta$ will not vanish identically and its size must be assessed.

The asymptotics we prove in Theorems~\ref{theoremDisc:Theorem1}--\ref{theoremDisc:Theorem3} apply to general families of distributions in a suitable limit, regardless of whether these distributions arise from exponential families.
The examples in \secref{secDisc:TheoreticalAnalysis} confirm that the asymptotic behaviour is unchanged in cases where the distributions are derived from exponential families.

\vspace*{-6mm}
\section{Explicitly computable examples}
\label{secDisc:Examples}
\subsection{Validating approximation accuracy}

In this section we illustrate the use and accuracy of the discrepancy approximation $\hat{\delta}$.
The examples, including two taken from the literature, are selected because the true likelihood and true MLE can be obtained with sufficient computational effort, so the discrepancy approximation $\hat{\delta}$ can be compared directly with the true discrepancy $\delta$.
Although saddlepoint methods are not needed to fit these models, they provide controlled settings in which the discrepancy approximation can be tested against the exact discrepancy.

All saddlepoint-based computations below are performed using the framework introduced by \cite{oketch2025} and implemented in the R package \texttt{saddlepoint} \citep{saddlepoint_pkg}.
In each example, a CGF code object corresponding to the modelled random variable is created using the operations described in \cite{oketch2025}; each code object implements the mathematical CGF, as well as numerous specialized functions for efficiently computing the discrepancy approximation.
The saddlepoint MLE is computed using \texttt{find.saddlepoint.MLE(...)}, and the user can call for the discrepancy approximation by setting the argument \texttt{discrepancy=TRUE}.

\subsection{Gamma distribution with a fixed rate} \label{secDisc:GammaFixedRate}

Consider a single observation $x$ of a gamma-distributed random variable $X$ with an unknown shape parameter $\alpha$ and a fixed rate parameter $r = 1$.
For this simple model, the true log-likelihood is 
$
\log L(\alpha; x) = (\alpha-1) \log x - x - \log \Gamma(\alpha).
$
The CGF of $X$ is 
$ 
    K_X(t;\alpha) = -\alpha \log (1 - t)~\text{for}~\alpha > 0~\text{and}~t < 1,
$
and the saddlepoint log-likelihood function \eqref{eqDisc:SaddlepointLogLikelihood} can be found explicitly as
\[
    \log \hat{L}(\alpha; x) = 
      (1 - \alpha) \log(1-\hat{t}) - \hat{t}x - 
      \frac{1}{2} \log \left( 2 \pi \alpha \right),
\]
where $\hat{t} = \hat{t}(\alpha;x) = 1 - \alpha/x$.

We use the \verb|saddlepoint| R package to compute the saddlepoint MLE $\hat{\alpha}_{\mathrm{spa}}$ by maximizing $\log \hat{L}(\alpha; x)$, and the discrepancy approximation $\hat{\delta}$. 
Additionally, we obtain the true MLE, $\hat{\alpha}_{\mathrm{true}}$, by maximizing $\log L(\alpha; x)$, and compute the true discrepancy: $\delta = \hat{\alpha}_{\mathrm{true}} - \hat{\alpha}_{\mathrm{spa}}$. 
Figure~\ref{figDisc:GammaFixedRateComparison} demonstrates a close correspondence between the saddlepoint and true MLEs from $100$ simulated observations using a true shape parameter of $\alpha = 5$, displaying an evident visual alignment of estimates in the left panel.
The saddlepoint MLE is consistently slightly 
smaller than the true MLE, i.e., $\delta > 0$.
The right panel illustrates that the discrepancy approximation $\hat{\delta}$ 
accurately captures the sign and magnitude of the true discrepancy $\delta$.  
Among the 100 simulated observations, the largest discrepancy occurs 
at the observed value $x=1.58177$, where the true MLE is
$2.0564$ and the saddlepoint MLE is $2.0248$. Even in this instance, 
the approximate discrepancy $\hat{\delta}=0.033$ closely aligns with the true discrepancy $\delta=0.032$.

\begin{figure}[htbp]
\centering
\figuresize{.42}
\figurebox[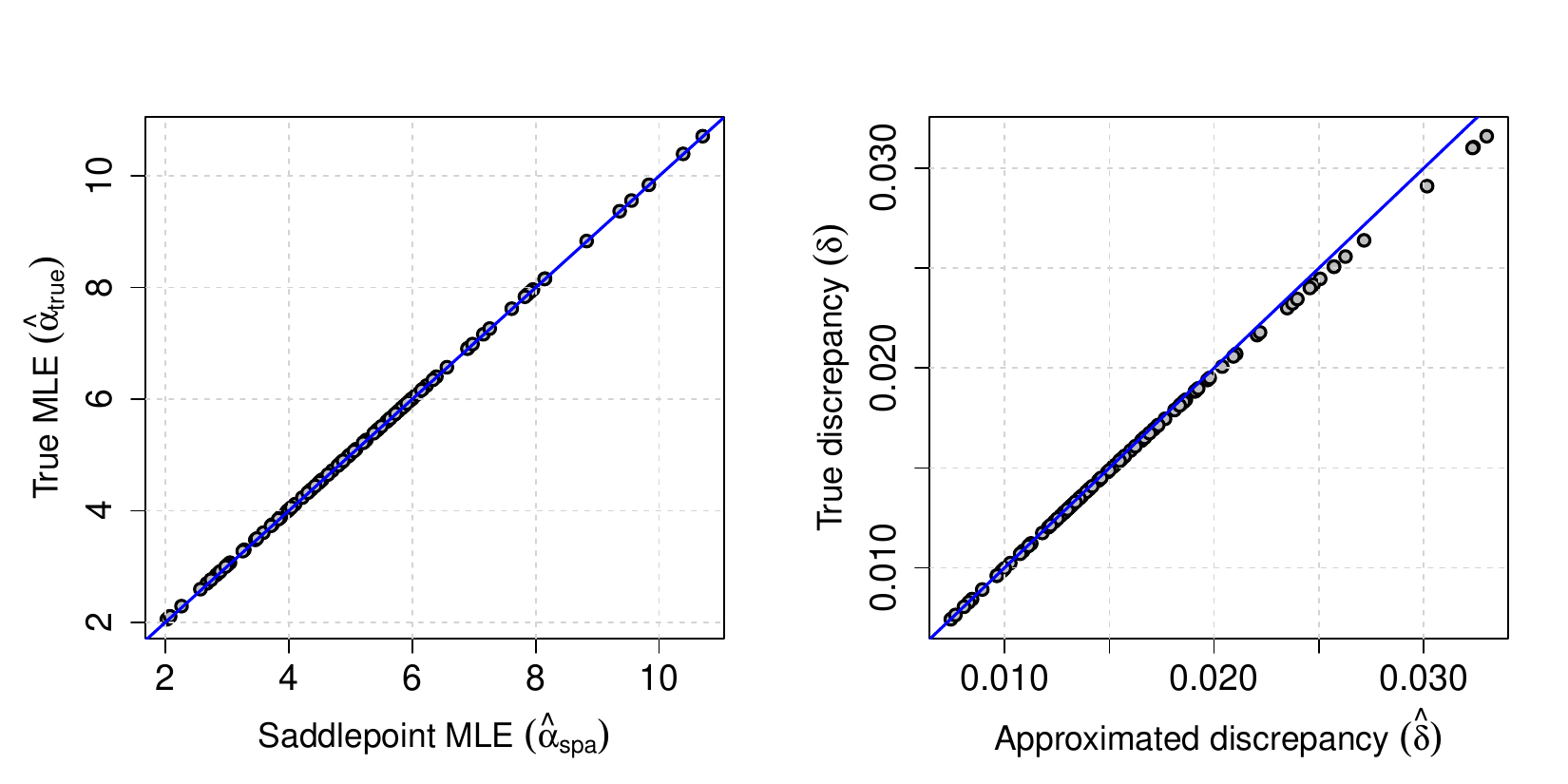]{5pc}{5pc}
\caption{
Comparison of MLEs and discrepancies for a gamma distribution with $100$ simulated observations and fixed rate $r = 1$; the true shape parameter is $\alpha = 5$.
The left panel shows the true MLE vs.\ saddlepoint MLE for shape parameter estimation; the right panel shows the true discrepancy vs.\ discrepancy approximation. 
The blue line in both panels represents the ideal scenario where the two values exactly match.
}
\label{figDisc:GammaFixedRateComparison}
\end{figure}

\subsection{Capture-recapture model with misidentification}\label{ssDisc:CaptureRecapture}

Model $M_{t, \alpha}$ is a capture-recapture model for estimating the size $N$ of an animal population in scenarios where captured animals may be misidentified \citep{Link2010, Vale2014, Zhang2019}. On each of several capture occasions, 
each of the $N$ animals in the population may be 
captured and correctly identified, captured but misidentified, or not captured. These responses for a single animal across all capture occasions are called the latent history for that animal,
and the vector $Y$ of latent history counts is multinomial with index $N$. However, misidentified animals are observed as single-capture `ghost' individuals, so instead of $Y$ we observe the vector $X=AY$, where $A$ is a known deterministic matrix
which splits misidentifications from correct captures for each latent history count.

The model parameters are the population size $N$, the probability $\alpha$ of correct identification, and the vector $\vec{p}$ of capture probabilities for each capture occasion. The vector $\pi$ of multinomial cell probabilities for $Y$ is a function of the model parameters.
The CGF of the observed counts $X$ is
$ 
    K_X(t;N, \alpha, \vec{p}) = K_{Y}(tA; N, \pi),
$
where $t$ is a row vector having the same dimension as $X$.

The linear mapping operation of the {\tt saddlepoint} R package constructs the CGF code object for $X$, and the {\tt find.saddlepoint.MLE(...)} function computes the saddlepoint MLEs and discrepancy approximations. 
The true MLEs for this model can be computed using the true likelihood function detailed in \cite{Vale2014}.

\begin{figure}[!htbp]
\centering
\figuresize{.42}
\figurebox[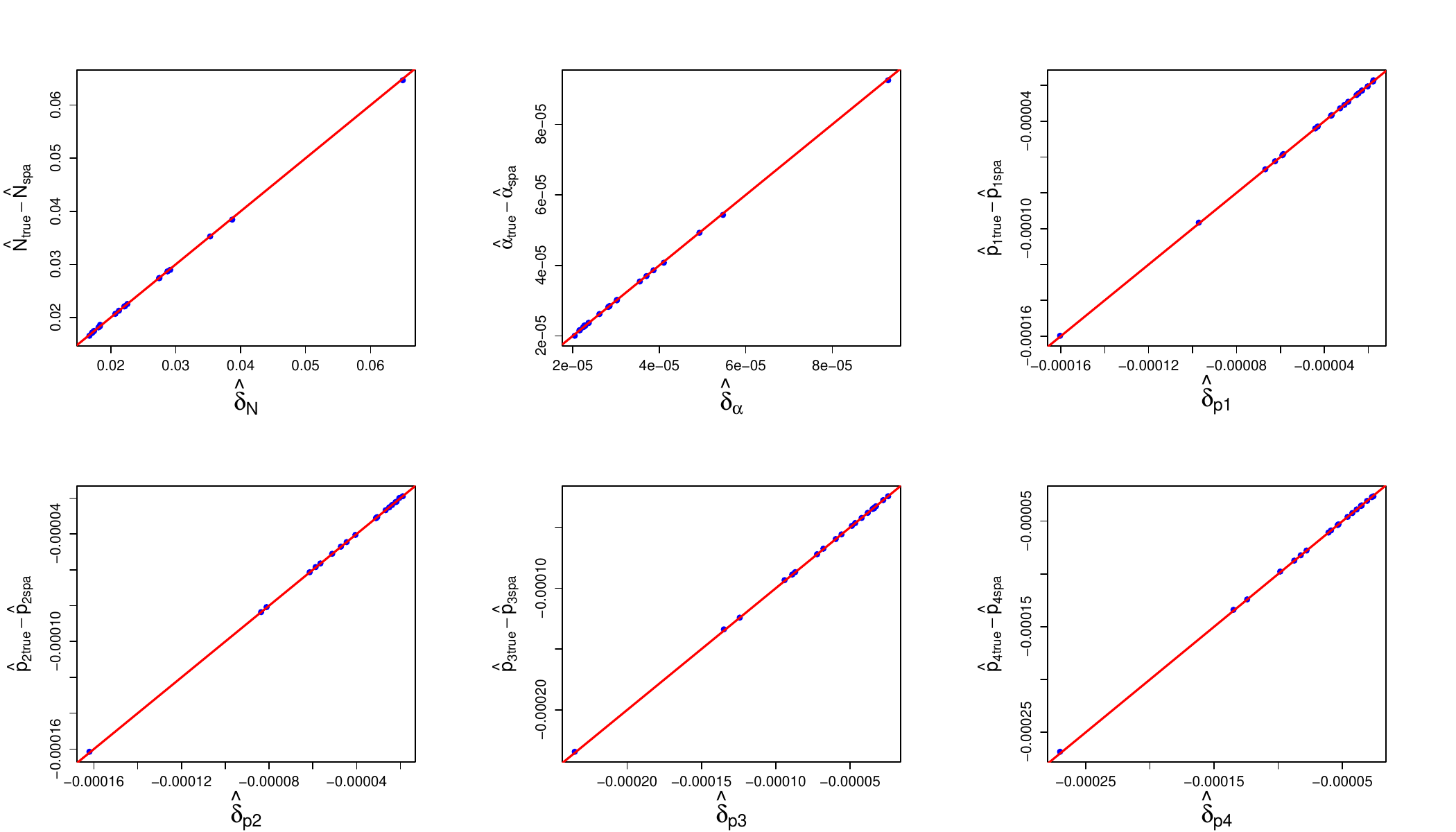]{10pc}{10pc}
\caption{
Comparison of true and approximated discrepancies between the saddlepoint and true MLEs from 20 simulations under model $M_{t,\alpha}$, with settings $N = 250$, $\alpha = 0.8$, and $\vec{p} = (0.4, 0.4, 0.6, 0.6)$. 
Each plot compares a component of the true discrepancy, such as $\delta_N=\hat{N}_{\text{true}}-\hat{N}_{\text{spa}}$, against the corresponding component of the discrepancy approximation vector $\hat{\delta}$.
Points on the identity line indicate close agreement between the true and approximated discrepancies for each parameter.    
}
\label{figDisc:DiscrepancyComparisonAllEstimatesMtalpha}

\figuresize{.5}
\figurebox[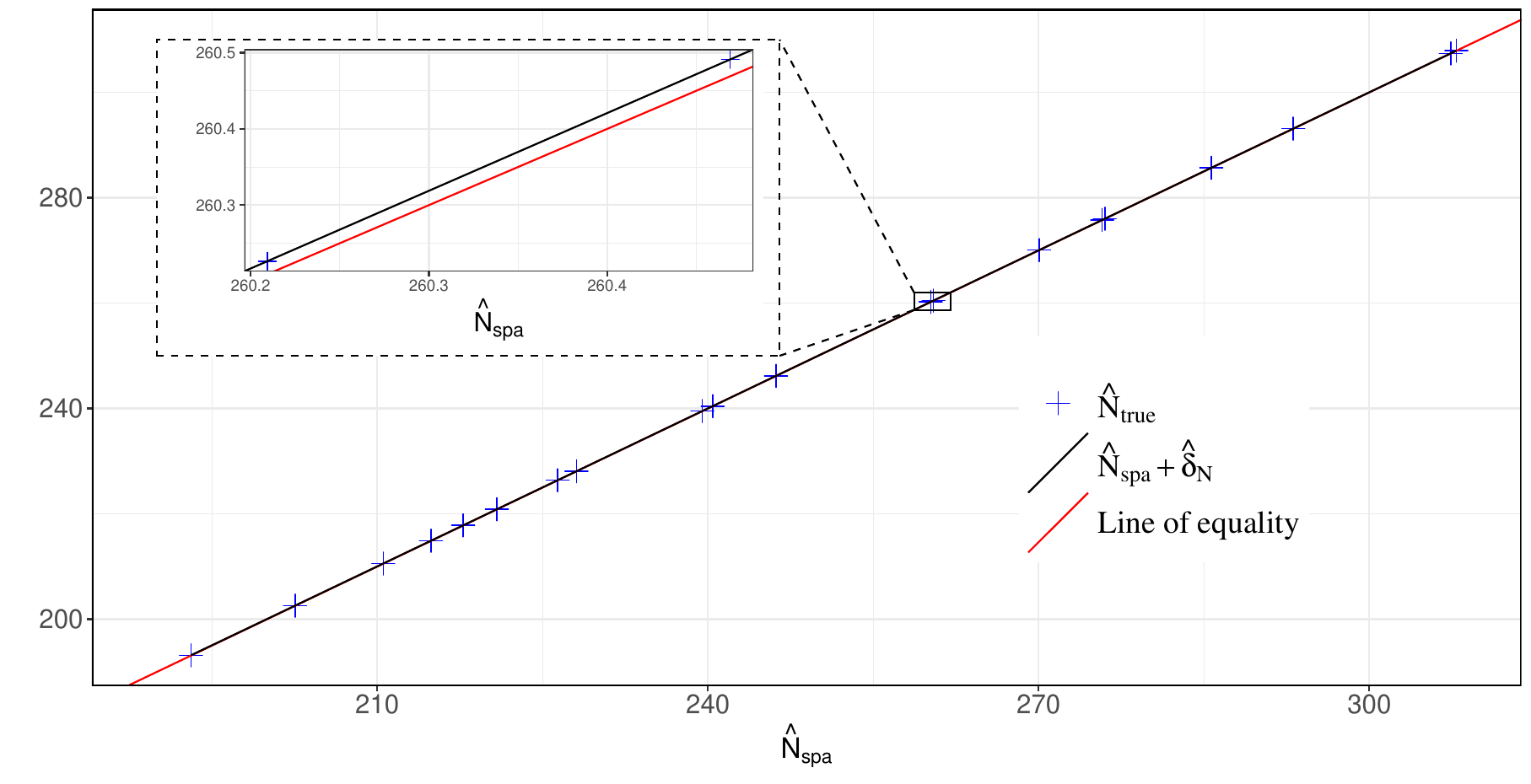]{10pc}{10pc}
\caption{
Comparison of the true MLEs $(\hat{N}_{\mathrm{true}})$  and saddlepoint MLEs $(\hat{N}_{\mathrm{spa}})$ for the interest parameter $N$ across 20 simulated datasets under model $M_{t,\alpha}$ with $N=250$, $\alpha=0.8$ and $\vec{p}= (0.4, 0.4, 0.6, 0.6)$.
The horizontal axis represents $\hat{N}_{\mathrm{spa}}$; the vertical axis plots both $\hat{N}_{\mathrm{true}}$ (blue crosses) and $\hat{N}_{\mathrm{spa}} + \hat{\delta}_N$ (black line), where $\hat{\delta}_N$ is the discrepancy approximation component computed from \eqref{eqDisc:DiscrepancyFormula} corresponding to estimates of $N$. 
The close alignment of the blue crosses to the line of equality highlights the similarity between the saddlepoint and true MLEs for $N$.  
The nearly perfect overlap of the black line with the blue crosses in the zoomed-in section highlights the relationship
$\hat{N}_{\mathrm{true}} \approx \hat{N}_{\mathrm{spa}} + \hat{\delta}_{N}  $.
}
\label{figDisc:ZoomedMtalphaPlot}
\end{figure}

We compare the true and approximated discrepancies through simulation results shown in Figures~\ref{figDisc:DiscrepancyComparisonAllEstimatesMtalpha} and \ref{figDisc:ZoomedMtalphaPlot}. 
Figure~\ref{figDisc:DiscrepancyComparisonAllEstimatesMtalpha} shows consistent agreement between true and approximated discrepancies across all parameters. 
In Figure~\ref{figDisc:ZoomedMtalphaPlot}, the alignment of the blue crosses with both the red line of equality and the black line (adjusted saddlepoint MLEs) demonstrates both the efficacy of the saddlepoint likelihood and the accuracy of our discrepancy approximation for the parameter of interest, $N$. 
The main panel shows that the saddlepoint MLEs are indistinguishable from the true MLEs on the scale of inference, while the insert panel shows that adjusting the saddlepoint MLE by the discrepancy approximation yields results that are indistinguishable from the true MLE even on the tightly zoomed scale.

In Table~\ref{tableDisc:MtalphaDiscrepancyComparisonTable}, we present results from a single dataset.
Here, both the true and discrepancy approximation come to $0.14$ 
to two decimal places, indicating
that the saddlepoint population size estimate $\hat{N}_{\mathrm{spa}}$ should have been larger by $0.14$ animals in order to match the true MLE.
Considering that the MLE for $N$ exhibits substantial sampling variability, with a standard error of $579.5$ animals, an adjustment of merely $0.14$ underscores the efficacy of the saddlepoint likelihood for the $M_{t, \alpha}$ model.

In contexts where the true likelihood is unknown, the discrepancy approximation provides a basis for assessing the utility of the saddlepoint likelihood as an estimation approach,
as well as providing an option to add the corresponding correction term to the saddlepoint MLE as shown by the black line in Figure~\ref{figDisc:ZoomedMtalphaPlot}.
An example of such a context is the two-source capture-recapture model discussed in \citet{McClintock2013} and \citet{Bonner_Holmberg2013}, where direct likelihood computation is intractable.

\begin{table}[htbp]
\centering
	\caption{Estimates and corresponding discrepancies under model $M_{t,\alpha}$ for a single dataset generated from the setting $N=900$, $\alpha=0.8$, and $\vec{p}= (0.7, 0.3, 0.1)$. }
{	\begin{tabular}{rrrcrr}
		\hline\hline
          & \multicolumn{2}{c}{Estimates (Standard Errors)} && \multicolumn{2}{c}{Discrepancies} \\
          \cline{2-3} \cline{5-6}
		   & True MLEs & Saddlepoint MLEs && True & Approximated \\ 
	    \hline
		\( \hat{N}\) & 1039.9 (579.5) & 1039.7 (579.6) && 1.4e-01 & 1.4e-01 \\ 
		\( \hat{\alpha}\) & 0.852 (0.230) & 0.851 (0.230) && 5.8e-05 & 5.7e-05\\ 
		\( \hat{p}_{1}\) & 0.612 (0.342) & 0.612 (0.342) && -8.4e-05 & -8.5e-05 \\ 
		\( \hat{p}_{2}\) & 0.250 (0.140) & 0.250 (0.140) && -3.4e-05 & -3.5e-05\\ 
        \( \hat{p}_{3}\) & 0.088 (0.050) &  0.088 (0.050) && -1.2e-05 & -1.2e-05 \\ 
		\hline
	\end{tabular}} 
	\label{tableDisc:MtalphaDiscrepancyComparisonTable}
\end{table}

\subsection{Linear birth-death population model}\label{ssDisc:LinearBirthDeath}

\cite{Davison2020} used a continuous-time birth-death model to infer demographic parameters of a whooping crane population based on records of annual counts of female whooping cranes arriving in Texas from 1938 to 2007.
They modelled per-individual birth and death rates from observations 
$z_0,\dots, z_k$ of the population sizes $Z_0,\dots, Z_k$ at 
years $j = 0, \dots, k$ using
a branching process. By the Markov property it is sufficient to consider the conditional distribution of $Z_j$ given $Z_{j-1}=z_{j-1}$ for $j=1,\dots,k$, which is
\begin{equation}\label{eqDisc:SumOfIID_LBD_BranchingProcess}
    Z_j = \sum_{i=1}^{z_{j-1}} U^{(i)}.
\end{equation}
Here, the $U^{(i)}$ are i.i.d.\ copies of the offspring distribution $U$, parametrized by 
birth and death rates $\lambda$ and $\mu$, which are the parameters to be inferred from the population trajectory.

The true likelihood function for this model is available, but
\cite{Davison2020} highlighted the computational efficiency of the saddlepoint likelihood, especially for large population sizes. In their work, estimates from true and saddlepoint likelihoods using simulations revealed moderate discrepancies. Here, we investigate the performance of the discrepancy approximation.

Since each $Z_j\,|\,Z_{j-1} = z_{j-1}$ is the sum of $z_{j-1}$ i.i.d\ copies of $U$, and denoting the CGF of $U$ by $K_U$, we have for each step $j = 1, \ldots, k$ that
 $   K_{Z_j \,|\, Z_{j-1}=z_{j-1}} (t_j \,;\, \lambda, \mu) 
    = 
    z_{j-1} K_U(t_j \,;\,\lambda, \mu).$
For a given sample path $z_0, z_1,\ldots,z_k$, we can define a vector of independent random variables 
$\tilde{Z}_1, \ldots, \tilde{Z}_k$ such that $\tilde{Z}_j \sim (Z_j\,|\,Z_{j-1} = z_{j-1})$. 
Then by the Markov property, 
the model likelihood for $(Z_1 = z_1, \ldots, Z_k = z_k|Z_0 = z_0)$ is equivalent to that for 
$(\tilde{Z}_1 = z_1, \ldots, \tilde{Z}_k = z_k)$. 
We leverage this equivalence in the \texttt{saddlepoint} R package by crafting the CGFs of the independent variables $\tilde{Z}_j$ rather than formulating a multivariate CGF for the serially dependent variables $Z_j$.

Using the whooping crane data, the saddlepoint MLE slightly overestimates both $\lambda$ and $\mu$ by approximately $0.0024$, compared to the true MLEs (0.1952 vs.\ 0.1928 for $\lambda$, and 0.1515 vs.\ 0.1492 for $\mu$). 
These true discrepancies of $-0.0024$ for both rates align with results reported by \cite{Davison2020}.
The discrepancy approximation gives $-0.0026$ for both rates, again demonstrating a close match to the true discrepancy in both sign and magnitude.
These discrepancies correspond to approximately $0.08$ standard errors for both estimates, so the approximation error introduced by using the saddlepoint likelihood is minor compared to the overall inferential uncertainty of the estimation process.
Moreover, even without access to the true MLE, the discrepancy approximation correctly signals the smallness of the approximation error.

\vspace*{-6mm}
\section{Computational details}
\label{secDisc:ComputationDetails}

The computation of the discrepancy approximation begins by obtaining the saddlepoint MLE
$\hat{\theta}_{\mathrm{spa}}$, defined as the maximizer of the saddlepoint log-likelihood
$\log \hat{L}(\theta;x)$.
The discrepancy approximation \(\hat{\delta}\) in \eqref{eqDisc:DiscrepancyFormula}
comprises two factors.
The first factor,
$(\nabla_\theta^T\nabla_\theta \log \hat L(\hat\theta_{\mathrm{spa}};x))^{-1}$,
is the inverse Hessian of the objective function $\log \hat{L}(\theta;x)$, available from standard optimization output. 
For the second factor, once a routine for evaluating $T_X(\theta, x)$ is in place the gradient $\nabla_\theta T_X(\theta, x)$ is obtained by automatic differentiation using the tools CppAD and TMBad \citep{cppad, TMB2016}. 
Since $T_X$ is a scalar-valued function, the Cheap Gradient Principle asserts that computing its gradient has approximately the same cost as evaluating the function itself.
The main challenge is evaluating $T_X(\theta,x)$, which does not benefit from the Cheap Gradient Principle even though it is built from derivatives of $K_X$.
Here we explain how the calculation of $T_X(\theta,x)$ can be organized for computational efficiency, using the structure of the underlying CGFs.

A straightforward evaluation of $T_X(\theta,x)$ in \eqref{eqDisc:funcT} requires six nested summations, which quickly becomes computationally expensive for moderate or high-dimensional $X$.
Even computing and storing the full set of third- and fourth-order partial derivatives of $K_X$ can become expensive.
Moreover, for typical applications, the mixed partial derivatives do not vanish identically, so sparsity does not offer much relief.

For a general symmetric positive semidefinite matrix $Q$, we define 
\begin{equation}
\label{eqDisc:TXagain}
\begin{aligned}
g_{X}(t,\theta,Q)
&= \frac{1}{8} \sum_{j_1,j_2,j_3,j_4}
    \frac{\partial^4 K_X(t;\theta)}{\partial t_{j_1}\partial t_{j_2}\partial t_{j_3}\partial t_{j_4}}
    Q_{j_1 j_2} Q_{j_3 j_4} \\
&\quad - \frac{1}{8} \sum_{i_1,i_2,i_3}\sum_{j_1,j_2,j_3}
    \frac{\partial^3 K_X(t;\theta)}{\partial t_{i_1}\partial t_{i_2}\partial t_{i_3}}
    \frac{\partial^3 K_X(t;\theta)}{\partial t_{j_1}\partial t_{j_2}\partial t_{j_3}}
    Q_{i_1 i_2} Q_{i_3 j_1} Q_{j_2 j_3} \\
&\quad - \frac{1}{12} \sum_{i_1,i_2,i_3}\sum_{j_1,j_2,j_3}
    \frac{\partial^3 K_X(t;\theta)}{\partial t_{i_1}\partial t_{i_2}\partial t_{i_3}}
    \frac{\partial^3 K_X(t;\theta)}{\partial t_{j_1}\partial t_{j_2}\partial t_{j_3}}
    Q_{i_1 j_1} Q_{i_2 j_2} Q_{i_3 j_3}
    ,
\end{aligned}
\end{equation}
such that $T_X$ is the result of $g_{X}$ evaluated at $t = \hat{t}(\theta,x)$ and $Q = \hat Q = K''_X(\hat{t};\theta)^{-1}$.
If $X$ is constructed from a simpler random vector $Y$, there may be useful relations between $g_X$ and $g_Y$ even when no corresponding relation exists between $T_X$ and $T_Y$. 
In particular, if $X = A Y$ for a fixed matrix $A$, then $K_X(t;\theta)=K_Y(tA;\theta)$ and
\begin{equation}\label{eqDisc:g_X-g_Y-linear_map}
    g_X(t,\theta,Q)=g_Y(tA,\theta,A^T QA)
    .
\end{equation}
No corresponding relation holds for $T_X$ and $T_Y$.

If $X$ has independent entries, \eqref{eqDisc:TXagain} reduces to
\begin{equation*}
    \begin{aligned}
        g_X(t,\theta,Q)&=
        \frac{1}{8}\sum_{j=1}^d K^{(4)}_{X_j}(t_j;\theta)\,Q_{jj}^2
        - \frac{1}{8} \sum_{i=1}^d \sum_{j=1}^d  K^{(3)}_{X_i}(t_i;\theta) K^{(3)}_{X_j}(t_j;\theta) Q_{ii} Q_{ij} Q_{jj} 
        \\ &\quad
        -\frac{1}{12} \sum_{i=1}^d \sum_{j=1}^d K^{(3)}_{X_i}(t_i;\theta) K^{(3)}_{X_j}(t_j;\theta) Q_{ij}^3
        ,
    \end{aligned}
\end{equation*}
which requires only two vectors of scalar derivatives
$K^{(r)}_{X_j}(t_j;\theta) = \frac{\partial^r}{\partial t_j^r}\,K_{X_j}(t_j;\theta)$ 
for $r = 3,4$.
However, $g_X$ must still accept a general positive semidefinite matrix $Q$, even when $K_X''(t;\theta)^{-1}$ is diagonal, because transformations such as \eqref{eqDisc:g_X-g_Y-linear_map} can yield non-diagonal
matrices.

Even when the derivatives of $K_X$ do not exhibit sparsity, the function $g_X$ may simplify considerably.  
For example, when $X$ is multinomial, the function $g_X$ can be evaluated in terms of a small number of relatively simple quantities that can be computed from $Q$ in $O(d^2)$ steps; see \supsecref{supsec:multinomialAABBCC}.

Alternatively, we may use special structure of the positive semidefinite matrix $Q$.
Suppose $Q$ admits a factorization of the form
\begin{equation}
\label{eqDisc:Q_decomp}
Q = BDB^T
  = \sum_{k=1}^r D_{kk}\, b_k b_k^T,
\end{equation}
where $D$ is a diagonal $r \times r$ matrix, $B \in \mathbb{R}^{d \times r}$, and $b_k \in \mathbb{R}^d$ denotes the $k$th column of $B$.
Factorizations of this kind include the standard Cholesky, LDL$^T$, and eigenvalue decompositions of $Q$, all of which have $r=d$.
In terms of this factorization, we define
\begin{equation}
\label{eqDisc:gXfactored}
    g_{X,\mathrm{factored}}(t,\theta,B, D) = g_X(t,\theta,B D B^T)
    .
\end{equation}

The function $g_{X,\mathrm{factored}}$ has the advantage that it can be evaluated in terms of simple scalar-valued functions based on higher-order derivatives, with a triple sum instead of the six-fold sum in \eqref{eqDisc:TXagain} without explicitly storing all third- and fourth-order mixed partial derivatives; the corresponding
formulas are given in \supsecref{labelDisc:K3K4operators}.

The function $g_{X,\text{factored}}$ is particularly useful in the important special case in which $X=AY$, where $A$ is a non-square matrix of size $d\times\ell$, typically with $d\ll \ell$.
In \eqref{eqDisc:g_X-g_Y-linear_map}, computing the matrix $A^T Q A$ required $\ell^2$ entries.
By contrast, the relation
\begin{equation}\label{eqDisc:gX-gY-factored}
g_{X,\mathrm{factored}}(t, \theta, B, D)
=
g_{Y,\mathrm{factored}}(tA, \theta, A^T B, D),
\end{equation}
only requires an $\ell \times r$ matrix $A^T B$, with a computational cost linear in $\ell$.
Intuitively, the computational gain arises because computing $A^T B$ in \eqref{eqDisc:gX-gY-factored} preserves the rank $r$ of the decomposition \eqref{eqDisc:Q_decomp}, whereas computing $A^T Q A$ does not.

Overall, the preceding discussion sets out two equivalent pathways for computing $T_X(\theta,x)$:
we may evaluate $g_X(t,\theta,Q)$ at $t = \hat{t}(\theta,x)$ and $Q = \hat Q = K''_X(\hat{t};\theta)^{-1}$; or,
with a convenient factorization $\hat Q=BDB^T$, evaluate
$g_{X,\mathrm{factored}}(t,\theta,B,D)$.
In our implementation, we structure the computation in a modular way, with each CGF code object using whichever pathway best exploits the available structure of the CGF and of $\hat Q$.

The computation of the discrepancy, including the implicit dependence of $\hat{t}$ and $\hat{Q}$ on $\theta$, the evaluation of $T_X$, and the computation of its gradient, can be invoked transparently through the \texttt{saddlepoint} R package \citep{saddlepoint_pkg} by the option \texttt{discrepancy = TRUE}.

\vspace*{-6mm}
\section{Theoretical analysis and asymptotics}
\label{secDisc:TheoreticalAnalysis}

\subsection{A limiting framework for saddlepoint methods}

Here we provide a theoretical basis for the discrepancy approximation $\hat{\delta}$ by showing that, in suitable limiting regimes, it captures the leading behaviour of the true discrepancy.
We present asymptotic results for the sizes of discrepancies in models with two distinct classes of identifiability conditions.

We suppose that the observed data vector $X_n \in \mathbb{R}^{d\times 1}$ is obtained as a sum
\begin{equation}\label{eqDisc:SAR}
    X_n = \sum_{i = 1}^{n} U^{\left(i\right)}
    ,
\end{equation}
where the $U^{(i)}$ are unobservable i.i.d.\ copies of a random vector $U_{\theta}$. 
Thus the observed data consist of totals aggregated across a population of size $n$.
The vector $U_\theta$ represents the response of a single individual, but the individual responses are not observed separately; only their aggregate contribution $X_n$ is observed.
Thus the model is built by specifying the individual response $U_\theta$, which we can typically model but not observe, whereas the population totals $X_n$ are observed but typically have an intractable likelihood.

Almost all applications of the saddlepoint likelihood in the literature concern models that match the sum-of-i.i.d.\ structure in \eqref{eqDisc:SAR}.
For example, the branching process model from \cite{Davison2020} represents each observed population size as an i.i.d.\ sum of offspring, see \eqref{eqDisc:SumOfIID_LBD_BranchingProcess}.
Similar remarks apply to the models in \cite{Pedeli2015} and \cite{Zhang2019,Zhang2021,Zhang2022}.
The same structure underlies the classical asymptotic theory of the saddlepoint approximation, including the error estimate of \cite{Daniels1954} and Edgeworth expansions, where the sum-of-i.i.d.\ form yields error bounds with simple and interpretable expressions in terms of $n$.

Our analysis of the discrepancy is formulated in a small-sample asymptotics viewpoint: the number of summands $n$ is taken to be large, but not so large that we can afford to ignore error terms.
As $n\to\infty$, the distribution of $X_n$ becomes asymptotically normal, so a likelihood based on the normal approximation to $X_n$ provides a natural alternative to the saddlepoint likelihood. 
The discrepancy due to the normal approximation is asymptotically small, but substantially larger than that of the saddlepoint approximation: $O(1/n)$ instead of $O(1/n^2)$ under suitable hypotheses \cite[Theorem~7; supplement, Example~26]{Goodman2022}. 
\citet{Davison2020} make this comparison directly and find that the normal approximation, which in their model is equivalent to a moment-based approach, is insufficiently accurate for practical use.

The MGF and CGF of $X_n$ are $M_X(t;\theta) = M_U(t;\theta)^n$ and $K_X(t;\theta) = nK_U(t;\theta)$.
We write $X_n = X_{\theta,n}$ when necessary to emphasize $\theta$-dependence, and omit the $n$-dependence when we write $M_{X}$ and $K_{X}$. 
Given a sequence of observed values $x_n$ for $X_n$, we define the corresponding sample means $u_n=x_n/n$. In accordance with the law of large numbers, we consider a limiting scenario where $u_n$ approaches a limiting value, $u_0$, while $x_n$ scales with $n$:
\begin{equation*}
  u_n = x_n/n,\;\; x_n=n u_n, ~\text{where}~ u_n\to u_0\text{ as }n\to\infty
.
\end{equation*}
Under this scaling, we can explicitly portray the $n$-dependence of the first- and second-order saddlepoint log-likelihoods, see for instance \citet[Chapter~6]{mccullagh1987tensor}:
\begin{align}
\log \hat{L}_n(\theta;x_n) 
&= n\log \hat{L}_{U}^* (\theta, u_n) 
+
c_n 
+
\log \hat{P}_U(\theta, u_n),
\label{eqDisc:AsymptoticsFirstOrderLogLikelihood} 
\\
\log \hat{L}_n(\theta;x_n)_2 
&= n\log \hat{L}_{U}^* (\theta, u_n) 
+
c_n 
+
\log \hat{P}_U(\theta, u_n)
+
T_U(\theta, u_n)/n,
\label{eqDisc:AsymptoticsSecondOrderLogLikelihood}
\end{align}
where
\begin{align*}
\log \hat{L}_{U}^* (\theta, u)  
&= 
K_{U}(\hat{t}_U(\theta; u);\theta) - \hat{t}_U(\theta; u) u, 
\\
c_n 
&= 
- \tfrac{1}{2} d\log(2\pi n),
\\
\log \hat{P}_U(\theta, u) 
&= 
- \tfrac{1}{2} \log \det ( K''_{U}(\hat{t}_U(\theta; u);\theta) ). 
\end{align*}
The final term of \eqref{eqDisc:AsymptoticsSecondOrderLogLikelihood} follows because
the matrix $Q$ from \eqref{eqDisc:funcT} is 
$Q = K''_X(\hat{t};\theta)^{-1} = n^{-1} K''_U(\hat{t};\theta)^{-1}$,
so from \eqref{eqDisc:SecondOrderSaddlepointLogLikelihood} and \eqref{eqDisc:funcT} we obtain the relation
\begin{equation}
\label{eqDisc:TX_TU}
    T_X(\theta, x_n) = T_U(\theta, u_n)/n.
\end{equation}

In this limiting framework, the scaling of the true log-likelihood is commensurate with that of the approximations in \eqref{eqDisc:AsymptoticsFirstOrderLogLikelihood} and \eqref{eqDisc:AsymptoticsSecondOrderLogLikelihood}. 
See, for instance, \citet[Chapter~6]{mccullagh1987tensor}, where:
\begin{align}
    \log L_n(\theta; x_n) - \log \hat{L}_n(\theta;x_n) 
    &=
    O(1/n),
    \\
    \log L_n(\theta; x_n) - \log \hat{L}_n(\theta;x_n)_2 
    &=
    O(1/n^2).
    \label{eqDisc:AsymptoticsTrueBigOnotation}
\end{align}
To make this comparison more explicit, we can define a function $P_n(\theta,u)$ such that
\begin{equation}
    \label{eqDisc:AsyptoticTrueLoglik}
    \log L_n(\theta; x_n) 
    =
    n\log \hat{L}_{U}^* (\theta, u_n) 
    +
    c_n 
    +
    \log P_n(\theta, u_n),
\end{equation}
so \eqref{eqDisc:AsymptoticsTrueBigOnotation} is the assertion that
$\log P_n(\theta, u_n) = \log \hat{P}_U(\theta, u_n) + n^{-1}T_U(\theta, u_n) + O(1/n^2)$.

The MLEs and discrepancies depend on $n$, which we indicate explicitly by writing 
\begin{equation}
\begin{aligned}
    \label{eqDisc:theta_n_notation}
    &\hat{\theta}_{\mathrm{true},n} = \underset{\theta}{\mathrm{argmax}}\, \log L_n(\theta;x_n),
    \quad 
    \hat{\theta}_{\mathrm{spa},n} = \underset{\theta}{\mathrm{argmax}}\, \log \hat{L}_n(\theta;x_n), 
    \\
    &\delta_n = \hat{\theta}_{\mathrm{true},n} - \hat{\theta}_{\mathrm{spa},n},
    \quad
    \hat{\delta}_n = 
    -\{\nabla_{\theta}^T \nabla_{\theta} \log \hat{L}_n(\hat{\theta}_{\mathrm{spa},n};x_n)\}^{-1}
    \spacednablatranspose{\theta}{T_X(\hat{\theta}_{\mathrm{spa},n}, x_n)}.
\end{aligned}
\end{equation}
Although $\hat{\theta}_{\mathrm{true},n}$, $\hat{\theta}_{\mathrm{spa},n}$, $\delta_n$, and $\hat{\delta}_n$ also implicitly depend on $x_n$, or equivalently on $u_n$, we suppress this in our notation. 

\subsection{Regularity conditions}

Some regularity assumptions are needed to ensure that, for instance, the random variable $X_n$ has a well-behaved density for large $n$.
Here we use a sufficient condition expressed in terms of the MGF.
We will require that, uniformly over $(t,\theta)$ in a neighbourhood to be specified, the MGF $M_U(t,\theta)$ is well-defined and satisfies
\begin{equation}\label{eqDisc:RegularityCondition1}
    \left|\frac{M_U(t + \mathrm{i}\omega; \theta)}{M_U(t; \theta)}\right| \leq (1+\eta|\omega|^2)^{-\eta}~\text{for all}~\omega \in \mathbb{R}^{1 \times d}~\text{and a constant} ~\eta > 0
    ,
\end{equation}
\begin{equation}\label{eqDisc:RegularityCondition2}
\begin{aligned}
    &\left|\frac{\partial^{k+l} M_U (t + \mathrm{i}\omega; \theta) }{\partial \theta_{i_1}\dots \partial \theta_{i_k}\partial t_{j_1}\dots\partial t_{j_l}}\right| \leq \gamma(1+|\omega|)^{\gamma}~\text{for some}~\gamma \in \mathbb{R}^{+},\\
    &\hspace{7cm}  0\leq k \leq 4~\text{and}~1\leq k+l \leq 9
    .
\end{aligned}
\end{equation}
That is, the MGF must decay at least polynomially along the imaginary direction of the complex plane, with derivatives that grow at most polynomially.
These conditions are very mild: for instance, \citet[Example 32]{Goodman2022} shows that they can hold even for ill-behaved distributions.

\subsection{Asymptotics when the model is fully identifiable at the level of the sample mean}

Here we address the asymptotic behaviour of the true discrepancy $\delta_n$ and its approximation $\hat{\delta}_n$ 
for models in our first identifiability class. In general, if we observe
$x_n \sim X_{\theta_0,n}$, then by the law of large numbers, the sample mean $u_n = x_n/n$ converges to a limiting value $u_0 = E(U_{\theta_0})$ that is uniquely determined by the parameter vector $\theta_0$.  
Our first identifiability class describes the situation where, conversely, the parameter vector $\theta_0$ is uniquely determined by $u_0$.
We begin with the well-specified case.
The following theorem is proved in \customref{appendix1}.

\begin{theorem}
\label{theoremDisc:Theorem1}
    Consider a well-specified model with observation $x_n$ drawn from the distribution $X_{\theta_0, n}$, 
    where 
    $\theta_0 \in \mathbb{R}^{p \times 1}$, 
    and let $u_0 = E(U_{\theta_0})$.  Suppose the parametric model is such that 
    \begin{equation}
        \nabla_{\theta} K'_U(0; \theta_0)~\text{has rank $p$},
        \label{eqDisc:Theorem1NonSingularCondition}
    \end{equation}
    and \eqref{eqDisc:RegularityCondition1}--\eqref{eqDisc:RegularityCondition2} hold for $(t,\theta)$ in a neighbourhood of $(0, \theta_0)$.
    Then, with $\theta$ restricted to a suitable neighbourhood of $\theta_0$, the MLEs 
    $\hat{\theta}_{\mathrm{true},n}$, $\hat{\theta}_{\mathrm{spa},n}$
    exist uniquely with high probability, i.e., with probability that tends to 1 as $n\to\infty$, and
    \begin{equation}
    \label{eqDisc:theta_true_spa_convergence_theta0}
        \hat{\theta}_{\mathrm{true},n} = \theta_0 + O_P(n^{-1/2}), \quad
        \hat{\theta}_{\mathrm{spa},n} = \theta_0 + O_P(n^{-1/2})~\text{as $n\to\infty$}.
    \end{equation}
    Moreover there exists a constant 
    \begin{equation}
    \label{LimitingRescaledDiscrepancy}
        k_0 = -\{\nabla_{\theta}^T \nabla_{\theta} 
        \log \hat{L}_{U}^* (\theta_0, u_0)
        \}^{-1} 
        \spacednablatranspose{\theta}{T_U(\theta_0, u_0)}
        \in \mathbb{R}^{p \times 1}
        ,
    \end{equation}
    such that 
    \begin{align}
        &\delta_n =
        n^{-2} (k_0 + O_P(n^{-1/2})),
        \label{Theorem1_results1}
        \\
        &\hat{\delta}_n =
        n^{-2} (k_0 + O_P(n^{-1/2})),~\text{and}~
        \label{Theorem1_results2}
        \\
        &\delta_n = \hat{\delta}_n + O_P(n^{-3})~~~\text{as $n \to \infty $}.
        \label{Theorem1_results3}
    \end{align}
\end{theorem}

In Theorem~\ref{theoremDisc:Theorem1}, the assumption \eqref{eqDisc:Theorem1NonSingularCondition} implies that the function $\theta\mapsto E(U_\theta)$ is invertible in a neighbourhood of $\theta_0$.
Thus, observing the sample mean $u_n$ is asymptotically sufficient to identify the parameter $\theta_0$, as confirmed by \eqref{eqDisc:theta_true_spa_convergence_theta0}. 
In the terminology of \cite{Goodman2022}, the model is fully identifiable at the level of the sample mean. 
Indeed, as shown in \citet[Theorem 5]{Goodman2022}, $n^{1/2}(\hat{\theta}_{\mathrm{true},n} - \theta_0)$ and $n^{1/2}(\hat{\theta}_{\mathrm{spa},n} - \theta_0)$ both converge in distribution to a normal random variable with a known covariance matrix.

The next theorem, also proved in \customref{appendix1}, extends Theorem~\ref{theoremDisc:Theorem1} to the case where the model is misspecified. Suppose
the observed $x_n$ is drawn according to $x_n = \sum_{i=1}^n V_i$, where the $V_i$ are i.i.d.\ random variables whose distribution need not coincide with any $U_\theta$. By the law of large numbers,
the sample mean still satisfies
$u_n\to u_0=E(V_i)$, but now $u_0$ need not coincide with any $E(U_{\theta_0})$.
With an appropriate generalization of the identifiability condition \eqref{eqDisc:Theorem1NonSingularCondition}, the same conclusions apply.

\begin{theorem}
\label{theoremDisc:Theorem2}
Fix $u_0 \in \mathbb{R}^{d\times 1}$ and suppose the sample means $u_n=x_n/n$ satisfy $u_n = u_0 + O_P(n^{-1/2})$.
Suppose there exist $t_0\in \mathbb{R}^{1\times d}$ and $\theta_0 \in \mathbb{R}^{p\times 1}$ such that
\begin{align}
    &K'_{U}(t_0;\theta_0) = u_0
    ,
    \label{eqDisc:Theorem2SaddlepointEqn}
    \\
    &\nabla_{\theta} \log \hat{L}_{U}^* (\theta_0, u_0) = 0
    ,
    \label{eqDisc:gradientLeadingOrder}
    \\
    &\nabla_{\theta}^T \nabla_{\theta} \log \hat{L}_{U}^* (\theta_0, u_0)~\text{is negative definite}
    ,
    \label{eqDisc:HessianLeadingOrder}
\end{align}
and \eqref{eqDisc:RegularityCondition1}--\eqref{eqDisc:RegularityCondition2} hold for $(t,\theta)$ in a neighbourhood of $(t_0,\theta_0)$.
Then, with $\theta$ restricted to a suitable neighbourhood of $\theta_0$, 
the MLEs $\hat{\theta}_{\mathrm{true},n}$, $\hat{\theta}_{\mathrm{spa},n}$ exist uniquely
with high probability, 
and \eqref{eqDisc:theta_true_spa_convergence_theta0}, \eqref{Theorem1_results1}--\eqref{Theorem1_results3} still hold with the same constant $k_0$ from \eqref{LimitingRescaledDiscrepancy}.

If instead the sample means satisfy $u_n=u_0+o_P(1)$, then \eqref{Theorem1_results3} holds without change, and \eqref{eqDisc:theta_true_spa_convergence_theta0}, \eqref{Theorem1_results1}--\eqref{Theorem1_results2} hold with $O_P(n^{-1/2})$ replaced by $o_P(1)$.
\end{theorem}

Here, assumption \eqref{eqDisc:Theorem2SaddlepointEqn} ensures that the saddlepoint equation has a solution. Assumptions \eqref{eqDisc:gradientLeadingOrder}--\eqref{eqDisc:HessianLeadingOrder} are the conditions for the parameter $\theta_0$ to be the non-degenerate local maximum of $\log \hat{L}_U^*(\theta, u_0)$. 
The negative definiteness of the Hessian in \eqref{eqDisc:HessianLeadingOrder} generalizes \eqref{eqDisc:Theorem1NonSingularCondition} from Theorem~\ref{theoremDisc:Theorem1}, as shown in \supsecref{labelDisc:ConditionNegDefin}.
Importantly, Theorem~\ref{theoremDisc:Theorem2} shows that, even under potential model misspecification,  the discrepancy results remain the same, showcasing the robustness of these findings.

Overall, for models whose parameters are identifiable from their sample mean, Theorems~\ref{theoremDisc:Theorem1} and~\ref{theoremDisc:Theorem2} indicate that
the approximate discrepancy  $\hat{\delta}_n$ closely matches the true discrepancy $\delta_n$. Further, these discrepancies become negligible as $n \to \infty$.
To illustrate these results, we use the model in \secref{secDisc:GammaFixedRate}.

\begin{figure}[!b] 
\centering
\figuresize{.7}
\figurebox[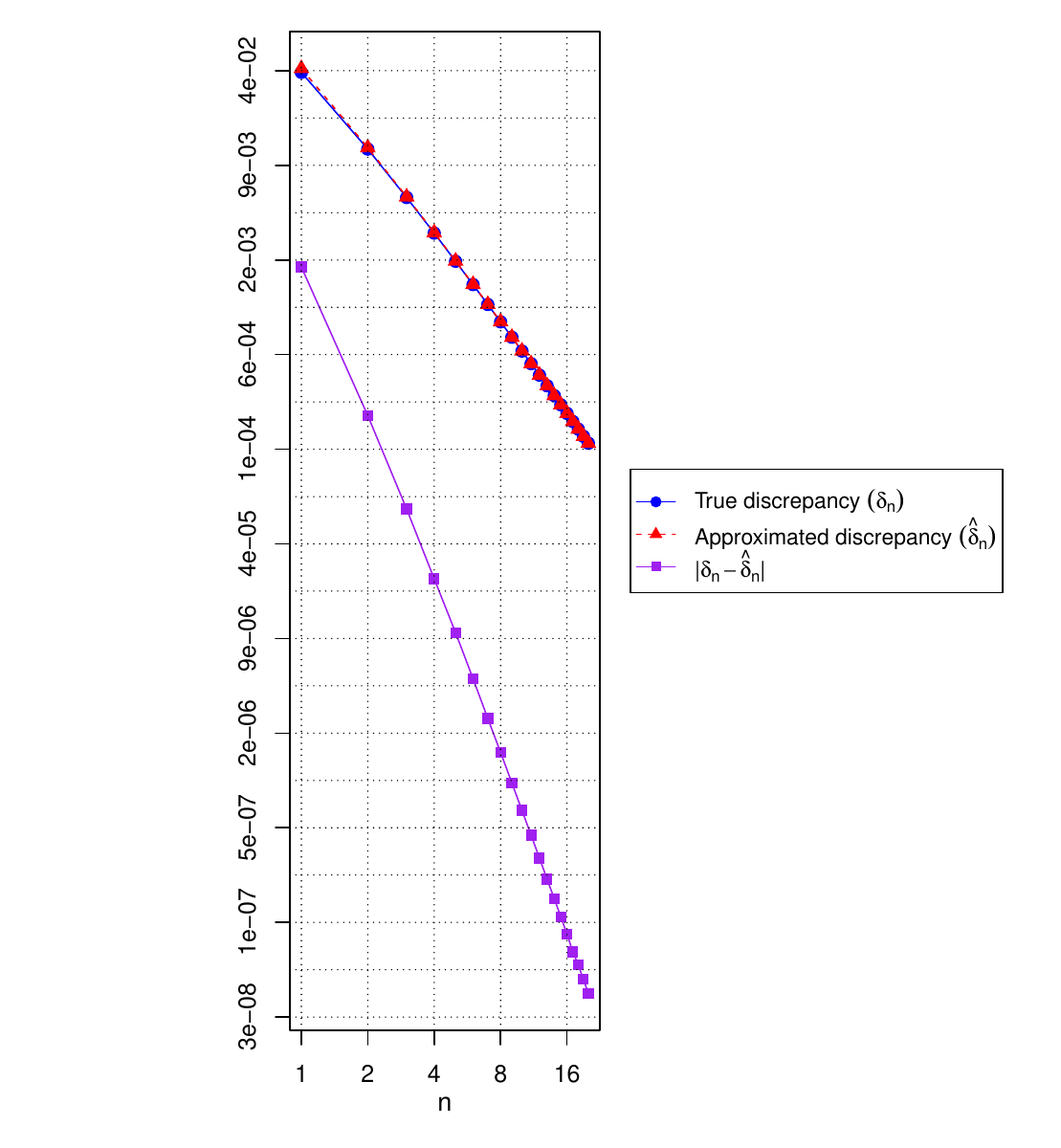]{20pc}{25pc}
\caption{
Log-log plot of true (blue) and approximated (red) discrepancies for the shape parameter estimates of a gamma distribution with a fixed rate. 
The underlying estimates are  derived from observations where $x_n = n u_0$ with $u_0 = 1.3045$.
Grid lines are plotted at equivalent spacing on horizontal and vertical axes.
}
\label{figDisc:GammaFixedRateDiscrepancyAsymtotocs}
\end{figure}

\begin{example}
Suppose we observe an instance of a random variable $X_{\alpha,n}$, assumed to be the sum of $n$ i.i.d. copies of $U_{\alpha} \sim \mathrm{Gamma}(\alpha, 1)$.
It follows that $X_{\alpha,n} \sim \mathrm{Gamma}(n\alpha, 1)$, and the corresponding CGF is
\[ 
    K_X(t;\alpha) = n K_U(t;\alpha) = -n\alpha \log (1 - t)~\text{for}~\alpha > 0~\text{and}~t < 1.
\]
With $K_U'(0;\alpha) = E(U_\alpha) = \alpha$, condition \eqref{eqDisc:Theorem1NonSingularCondition} is satisfied.
For a specific instance $x_n = n u_n$ of $X_{\alpha,n}$, the true and saddlepoint log-likelihood are respectively
\begin{equation*}
    \begin{aligned}
     \log L_n(\alpha; x_n) &= (n\alpha - 1)\log x_n - x_n - \log \Gamma(n\alpha), 
     \\ 
     \log \hat{L}_n(\alpha; x_n) &=  
    (n\alpha - 1) \log x_n - x_n + n \alpha +
    \left(\tfrac{1}{2}-n\alpha\right) \log (n\alpha)
    - \tfrac{1}{2}\log (2\pi),
    \end{aligned}
\end{equation*}
and we define 
$\hat{\alpha}_{\mathrm{true},n} = \mathrm{argmax }_{\alpha}~ \log L_n(\alpha; x_n)$ and $\hat{\alpha}_{\mathrm{spa},n} = \mathrm{argmax}_{\alpha}~ \hat{L}_n(\alpha; x_n)$.

Figure~\ref{figDisc:GammaFixedRateDiscrepancyAsymtotocs} illustrates the theoretical findings \eqref{Theorem1_results1}--\eqref{Theorem1_results3}, showing the true discrepancies $\delta_n = \hat{\alpha}_{\mathrm{true},n} - \hat{\alpha}_{\mathrm{spa},n}$ and their approximations $\hat{\delta}_n$ obtained using \eqref{eqDisc:theta_n_notation}. 
The plot is based on artificial data in which $u_n$ is held fixed at $u_n=u_0=1.3045$ for all $n$, so $x_n=n u_0$.
The plot demonstrates how the discrepancies and their differences
$|\delta_n - \hat{\delta}_n|$ change across a range of values for $n$, in accordance with the predictions of Theorem~\ref{theoremDisc:Theorem1}.
The true and approximated lines are visually indistinguishable and 
become closer to zero with increasing $n$. Both discrepancies exhibit approximately $n^{-2}$ decline, reflected by the slope of approximately $-2$ in the log-log plot. The difference between discrepancies, represented by the purple line, shows a steeper $n^{-3}$ decline.
\end{example}

\subsection{Asymptotics when the model is partially identifiable at the level of the sample mean}
We now consider models in our second identifiability class, where only some of the parameters in $\theta$ are uniquely identified by $E(U_{\theta})$.
The following theorem is proved in \supsecref{ProofTheorem3}.

\begin{theorem}
\label{theoremDisc:Theorem3}
Suppose the model parameter $\theta \in \mathbb{R}^{p \times 1}$ can be decomposed as
\begin{equation}
    \theta = 
    \begin{pmatrix}
        \omega \\
        \tau
    \end{pmatrix},
    ~\text{where}~ \omega \in \mathbb{R}^{p_{1} \times 1}~\text{and}~\tau \in \mathbb{R}^{p_2 \times 1},
\end{equation}
in such a way that
\begin{equation}
\label{eqDisc:PartiallyIdentifiableAssumption}
    K_U'(0; \theta) = E(U_\theta)~\text{depends only on}~\omega.
\end{equation} 
Fix $
\theta_0 = 
\smallmat{\omega_0 \\ \tau_0}
$ and set $u_0 =  K_U'(0; \theta_0) = E(U_{\theta_0})$.
Suppose that \eqref{eqDisc:RegularityCondition1}--\eqref{eqDisc:RegularityCondition2} hold for $(t,\theta)$ in a neighbourhood of $(0, \theta_0)$, and that the $d \times p_1$ matrix
$
B = \nabla_{\omega}  K'_U \left( 0; \theta_0 \right)
$
has rank $p_1$.

Fix $z_0 \in \mathbb{R}^{d \times 1}$ and consider the deterministic sequence of observations $x_n$ for $X_n$ given by 
\begin{equation}
\label{eqDisc:scalingOnObservationsTheorem3}
    x_n = n u_0 + n^{1/2}z_0. 
\end{equation}
To express the asymptotic identifiability condition, introduce the parametric model $Z_{v,\tau} = Bv + e$, where $v \in \mathbb{R}^{p_1 \times 1}$ and $e \in \mathbb{R}^{d\times 1}$ is normally distributed with a zero mean vector and a covariance matrix
$
K_U''
(
0;
\smallmat{\omega_0 \\ \tau}
)
$.
We suppose that this model has non-degenerate local MLEs for $v$ and $\tau$ when $Z_{v,\tau} = z_{0}$.

Under these assumptions:

(a) \label{theoremItem:item_a_theorem3}
A constant 
$
\smallmat{
    k_{0,v} \\
    k_{0,\tau}
}
\in 
\mathbb{R}^{p \times 1}
$
exists such that 

\begin{align}
    &\begin{pmatrix}
        \delta_{n,\omega} \\
        \delta_{n,\tau}
    \end{pmatrix}
    = 
    \begin{pmatrix}
        n^{-3/2} (k_{0,v} + o(1)) \\
        n^{-1} (k_{0,\tau} + o(1))
    \end{pmatrix},
    \label{eqDisc:Theorem3DiscTrue}
    \\[4pt]
    &\begin{pmatrix}
        \hat{\delta}_{n,\omega} \\
        \hat{\delta}_{n,\tau}
    \end{pmatrix}
    = 
    \begin{pmatrix}
        n^{-3/2} (k_{0,v} + o(1)) \\
        n^{-1} (k_{0,\tau} + o(1))
    \end{pmatrix}
    ~\text{and}~
    \label{eqDisc:Theorem3DiscApprt}
    \\[4pt]
    &\begin{pmatrix}
        \delta_{n,\omega} \\
        \delta_{n,\tau}
    \end{pmatrix}
    =
    \begin{pmatrix}
        \hat{\delta}_{n,\omega} \\
        \hat{\delta}_{n,\tau}
    \end{pmatrix}
    +
    \begin{pmatrix}
        O(n^{-5/2})\\
        O(n^{-2})
    \end{pmatrix}
    ~\text{as}~n \to \infty.
    \label{eqDisc:Theorem3DiscDiff}
\end{align}

(b) \label{theoremItem:item_b_zero_matrix_condition} 
In addition
abbreviate 
$
\Tilde{Q} = K''_U
(0; \theta )^{-1}
$,
$
\Tilde{B} =
\nabla_{\omega} 
K'_U
( 0; \theta )
$,
and $\Tilde{J} = \Tilde{Q} - \Tilde{Q}\Tilde{B}(\Tilde{B}^T \Tilde{Q} \Tilde{B})^{-1}\Tilde{B}^T \Tilde{Q}$.
Suppose that for all $i = 1,\dots, p_2$, the $p_1 \times p_2$ matrix
\begin{equation}
    \label{eqDisc:ZeroMatrixConditionTheorem3}
    \Tilde{B}^T \Tilde{Q}
    \frac{\partial K''_U}{\partial \tau_i}
    \left(
    0; 
    \theta
    \right)
    \Tilde{J}
    ~\text{is the zero matrix for all $\theta$}.
\end{equation}
Then for an explicit constant 
\begin{equation}
\label{eqDisc:ConstantNuDiagonalMatrixTheorem3}
    k_{0,\omega}^* 
    = 
    \left(
    B^{T}
    K''_U
    (0; 
        \theta_0
    )^{-1}
    B
    \right)^{-1}
    \spacednablatranspose{\theta}{
    T_U(
        \theta_0
        ,
        u_0
        )}
        \in \mathbb{R}^{p_1 \times 1}
        ,
\end{equation}
\begin{align}
    \delta_{n,\omega} &= n^{-2}(k_{0, \omega}^* + o(1))
    \label{eqDisc:TrueOmegaDiscSpecialCase}
    \\
    \hat{\delta}_{n,\omega} &= n^{-2}(k_{0, \omega}^* + o(1))
    ~\text{and}~ 
    \label{eqDisc:ApprxOmegaDiscSpecialCase}
    \\
    \delta_{n,\omega}
    &=
    \hat{\delta}_{n,\omega} + 
    O(n^{-3})~\text{as}~n \to \infty.
    \label{eqDisc:OmegaDiscDiffSpecialCase}
\end{align}

\end{theorem}

Theorem~\ref{theoremDisc:Theorem3} addresses scenarios where only a subset of parameters, namely $\omega$,
determines the mean, while the remaining parameters in $\tau$ affect the variance but have
no influence on the mean. Thus the exact expected value $u_0=E(U_{\theta_0})$ is not
sufficient to uniquely recover the full parameter vector $\theta_0$. When an observation
$x_n$ is given, only $\omega$ is asymptotically identifiable from the sample mean $x_n/n$,
while the remaining parameter $\tau$ is not.

The scaling \eqref{eqDisc:scalingOnObservationsTheorem3} of the observed values $x_n$ can be interpreted as holding the z-score
vector fixed as $n$ increases. Fixing the vector $z_0$ creates a limiting setting in which
we can simultaneously examine the differing approximation performance of the $\omega$ and
$\tau$ parameters. This framework for studying $n$-dependence in the approximation
discrepancy differs from customary limiting settings for studying inferential performance,
in which $z_0$ is typically subject to sampling variability.

Under the scaling \eqref{eqDisc:scalingOnObservationsTheorem3}, estimates of $\omega$ always converge to $\omega_0$ regardless of the value of $z_0$, whereas estimates of $\tau$ may depend upon $z_0$; details are given
in the proof. 
To express this, we introduce the reparametrization $\omega = \omega_0 + v n^{-1/2}$ such that $v$ captures the small deviations of $\omega$ from $\omega_0$ on the $n^{-1/2}$ scale.

The model $Z_{v,\tau}$ is introduced as a way of succinctly expressing the conditions for Theorem~\ref{theoremDisc:Theorem3}.
Under the reparametrization $\omega=\omega_0 + v n^{-1/2}$, the
quantities $Bv$ and $\mathrm{cov}(e)$ emerge as the limiting mean and covariance of
$n^{1/2} (X_n /n - u_0)$.  
The vector expression $Bv$ is the change in $E(U_{\theta})$ when $\omega$ is given a small perturbation away from $\omega_0$ in the direction $v$,
whereas the covariance matrix $\mathrm{cov}(e)$ corresponds to $\mathrm{cov}(U_\theta)$ at $\theta=\left(\genfrac{}{}{0pt}{}{\omega_0}{\tau}\right)$.  Introducing the model $Z_{v, \tau}$ separates the roles of $v$ and $\tau$ in the sense that $v$ has no influence on the covariance; this need not be true of $\omega$ in the original model.  
Thus, in the model $Z_{v, \tau}$, the parameter $v$ affects the mean only, in a linear way; whereas the parameter $\tau$ affects the covariance only, potentially non-linearly.  
We may regard $Z_{v, \tau}$ as a type of regression model, where $B$ represents the design matrix and $v$ is the vector of regression coefficients.

The identifiability condition in Theorem~\ref{theoremDisc:Theorem3} specifies that the observation
$Z_{v,\tau}=z_0$ must admit a non-degenerate local MLE. 
This will ensure the original model likelihood for $X_n = x_n$ also has a unique local maximizer when $n\to\infty$, so we can express the limiting behaviour for $X_n$ in terms of the approximating model $Z_{v,\tau}$ along with appropriate correction terms.  
The interpretation of $Z_{v,\tau}$ as a regression model makes it clear that the existence of an MLE for $v$ holds generally, because $B$ is of full rank and $\mathrm{cov}(e)$ is non-singular.
The condition for $\tau_0$ can alternatively be expressed as the explicit conditions \eqref{eqDisc:H3_appendixH_Goodman2022}--\eqref{eqDisc:H4_appendixH_Goodman2022} in the supplementary material, but it is difficult to provide a general and succinct interpretation of these expressions.

Theorem~\ref{theoremDisc:Theorem3}(b) shows that under an additional condition, the scaling of the discrepancy for $\omega$ reverts from $n^{-3/2}$ to $n^{-2}$ as in Theorems~\ref{theoremDisc:Theorem1}--\ref{theoremDisc:Theorem2} (although the assumptions of Theorem~\ref{theoremDisc:Theorem3}(b) remain more restrictive).
The additional condition in \eqref{eqDisc:ZeroMatrixConditionTheorem3} emerges from the proof of Theorem~\ref{theoremDisc:Theorem3} by causing certain mixed partial derivatives in $\omega$ and $\tau$ to vanish.
Again, it appears to be hard to provide a concise interpretation of this condition.

Overall, the results in Theorem~\ref{theoremDisc:Theorem3} show that where the model mean depends only on a subset $\omega$ of its parameters, the saddlepoint-based estimates for $\omega$ closely align with those from the true likelihood, whereas the $\tau$ estimates have larger true and approximated discrepancies. All these discrepancies, however, become negligible as $n \to \infty$.

To illustrate these findings, we simulate from a simple model for which true and saddlepoint MLEs, and the corresponding true and estimated discrepancies, can be calculated exactly.

\begin{figure}[!b] 
\centering
\figuresize{.6}
\figurebox[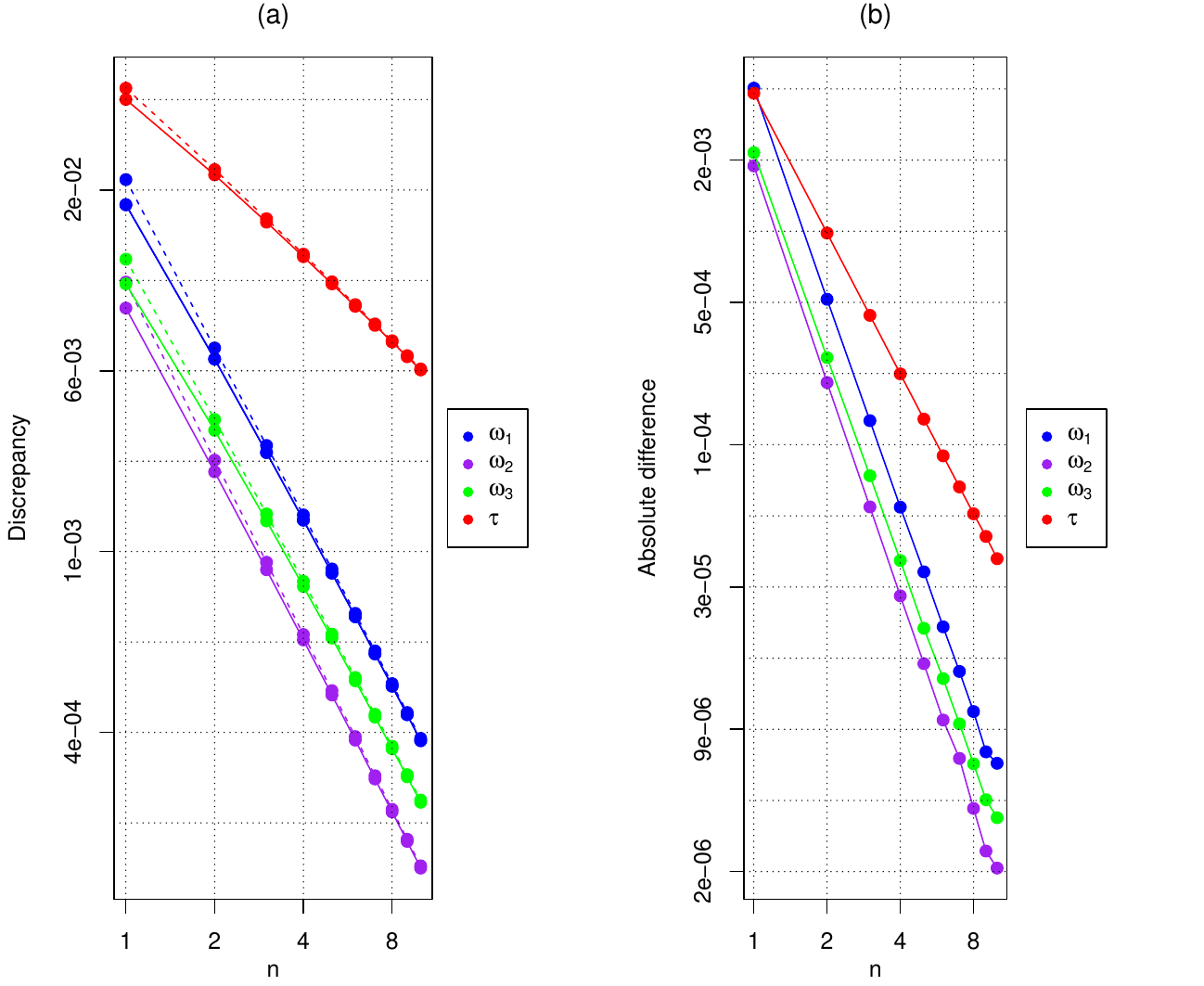]{20pc}{25pc}
\caption{
Log-log plots of discrepancies and their absolute differences for estimates of $\omega = (\omega_1, \omega_2, \omega_3) $ and $\tau$ for the model in Example~\ref{secDisc:MVGammaDiscSec}. 
The estimates are based on a deterministic sequence
$x_n = n u_0 + n^{1/2} z_0$,
where $u_0$ is the mean vector corresponding to $\omega_0=(1.5,3.6,5.8)$.
The vector $z_0$ of length $km = 15$ is assigned by simulating once from $z_0 \sim X_1 - E(X_1)$, then held fixed across all values of $n$.
Plot (a) compares true (solid lines) and approximated (dotted lines) discrepancies for $\omega$ and $\tau$ estimates, with discrepancies plotted as absolute values. 
Plot (b) shows the absolute differences in discrepancies between the true and approximated discrepancies.
Grid lines are plotted at equivalent spacing on horizontal and vertical axes.
}
\label{figDisc:MVGammaDiscPlot}
\end{figure}

\begin{example}[Discrepancies in $\omega$ and $\tau$ estimates as $n$ increases]
\label{secDisc:MVGammaDiscSec}
Let $U = U_{\omega, \tau}$ be a random vector with entries $U_{ij}$ indexed by $i = 1,\dots,k$ and $j = 1,\dots,m$, all independent. 
We consider $U_{ij}$ to be drawn from a gamma distribution with shape parameter $\omega_i \tau$ and rate parameter $\tau$, independently across different $i$ and i.i.d.\ across different $j$.
This parameterization ensures that the parameter $\tau$ influences the variance but not the mean, in accordance with the setup of Theorem~\ref{theoremDisc:Theorem3}.
For $t = (t_{11},t_{12},\dots,t_{km}) \in \mathbb{R}^{1\times km}$, the CGF of $U_{\omega,\tau}$ is 
\begin{equation}
\label{eqDisc:CGF_simple_MVgamma}
K_{U}(t; \omega, \tau) 
= 
\sum_{i=1}^k 
\sum_{j=1}^m 
-\omega_i \tau \log(1 - t_{ij}/\tau),
\end{equation}
where $t_{ij} < \tau$ and $\omega = (\omega_1, \dots, \omega_k)$.

Introducing the limiting framework as defined in \eqref{eqDisc:SAR}, consider observing a sum of $n$ i.i.d copies of $U_{\omega,\tau}$, represented by $X_{n} = \sum_{r=1}^n U^{(r)}$.
The $i,j$ entry of $X_n$ follows a $\mathrm{Gamma}(n\omega_i \tau, \tau)$ distribution. The matrix
$K_{U}''(0; \omega, \tau)$ is diagonal, and 
it can be shown that matrix $\Tilde{J}$ in assumption \eqref{eqDisc:ZeroMatrixConditionTheorem3} is zero for all $\theta = (\omega, \tau)^T$.
Therefore, we examine the theoretical findings in \eqref{eqDisc:TrueOmegaDiscSpecialCase}--\eqref{eqDisc:OmegaDiscDiffSpecialCase} for $\omega$ estimates, and \eqref{eqDisc:Theorem3DiscTrue}--\eqref{eqDisc:Theorem3DiscDiff} for $\tau$ estimates, by comparing the saddlepoint estimates based on \eqref{eqDisc:CGF_simple_MVgamma} with the true MLEs for a range of values of $n$.

Figure~\ref{figDisc:MVGammaDiscPlot} illustrates these findings. 
The plots are based on a
deterministic sequence $x_n = n u_0 + n^{1/2} z_0$, 
where $u_0$ is the vector of means of $U_{ij}$ for $i=1,\dots,k$ and $j=1,\dots,m$, and $z_0$ is chosen randomly, but is held fixed over all $n$ in the plots. We set $k = 3$ and $m = 5$.
In plot (a), the lines corresponding to true and approximated discrepancies remain visually close, and become closer to zero as $n$ increases. The discrepancies for $\omega$ estimates decrease more rapidly at a rate of approximately $n^{-2}$, reflected by the slope of about $-2$, whereas the
$\tau$-based discrepancies decay with a slope of approximately $-1$.
In plot (b), we examine the differences between the true and approximated discrepancies. Here, the $\tau$-based differences show a slope of approximately $-2$, whereas the $\omega$-based differences again have a steeper slope of approximately $-3$.
Both plots show overall agreement with the predictions of Theorem~\ref{theoremDisc:Theorem3}.
\end{example}

\appendix

\beginAppendixCounter

\section*{Appendix 1}
\customlabel{appendix1}{Appendix 1}
\subsection*{Proof of Theorems~\ref{theoremDisc:Theorem1}--\ref{theoremDisc:Theorem2}}

For $\theta \in \mathbb{R}^{p \times 1}$,
consider the following $p$-dimensional row-vector-valued expressions representing the gradients of rescaled log-likelihoods \eqref{eqDisc:AsyptoticTrueLoglik}, \eqref{eqDisc:AsymptoticsFirstOrderLogLikelihood}, and \eqref{eqDisc:AsymptoticsSecondOrderLogLikelihood}:
\begin{align}
      n^{-1} \nabla_{\theta} \log L_n(\theta; x_n) &= 
      \nabla_{\theta} \log \hat{L}_{U}^* (\theta, u_n) + n^{-1} \nabla_{\theta} \log P_n(\theta, u_n) ,
      \label{eqDisc:rescaledTrueLogLik}
    \\
    n^{-1} \nabla_{\theta} \log \hat{L}_n(\theta;x_n) &= 
    \nabla_{\theta}  \log \hat{L}_{U}^* (\theta, u_n) + n^{-1} \nabla_{\theta} \log \hat{P}_U(\theta, u_n) ,
    \label{eqDisc:rescaledFirstOrderLogLik}
    \\
    n^{-1} \nabla_{\theta} \log \hat{L}_n(\theta;x_n)_2 &= 
    \nabla_{\theta}  \log \hat{L}_{U}^* (\theta, u_n) + n^{-1} \nabla_{\theta} \log \hat{P}_U(\theta, u_n) +
    n^{-2} \nabla_{\theta} T_U(\theta, u_n).
    \label{eqDisc:rescaled_second_order_log_lik}
\end{align}
In \eqref{eqDisc:rescaledFirstOrderLogLik}--\eqref{eqDisc:rescaled_second_order_log_lik},
the gradients $n^{-1} \nabla_{\theta} \log \hat{L}_n(\theta;x_n)$ and $n^{-1} \nabla_{\theta} \log \hat{L}_n(\theta;x_n)_2$ depend on $n$ only as well-behaved functions of $u_n$ and $1/n$, and can be interpreted whether or not $n$ is an integer.
The term $ \nabla_{\theta} \log P_{n}(\theta, u_n)$
depends directly on $n$, which must be an integer.
However, 
\citet[Corollary~12]{Goodman2022} shows that $ \nabla_{\theta} \log P_{n}(\theta, u_n)$ is a well-behaved function of $1/n$. Specifically, there exists a continuously differentiable function $q(\theta, u, \rho)$ defined on a neighbourhood of $(\theta_0,u_0, 0)$ such that
\begin{equation}
\label{eqDisc:qFunctionDefinition}
    \nabla_\theta \log P_{n}(\theta, u)
    =
    \nabla_\theta \log \hat{P}_U(\theta, u)
    +  q(\theta,u,1/n) 
    ~
    \text{and}
    ~
    q(\theta,u,0) = 0
    .
\end{equation}
It follows that $n^{-1} \nabla_{\theta} \log L_n(\theta; x_n)$ 
can also be expressed as a well-behaved function of $u_n$ and $1/n$, even when $n$ is not an integer.

We prove the results in Theorem~\ref{theoremDisc:Theorem2}, which encompasses Theorem~\ref{theoremDisc:Theorem1} as a special case, by seeking the value of $\theta$ that solves $n^{-1} \nabla_{\theta} \log \hat{L}_n(\theta;x_n) = 0$ instead of $n^{-1} \nabla_{\theta} \log L_n(\theta; x_n) = 0$ in the limit $n \to \infty$.

\begin{proof}[\proofof{Theorem~\ref{theoremDisc:Theorem2}}]
\label{prf:Theorem2Proof}
Define the function
\begin{equation}
\label{eqDisc:Proof1_1}
\begin{aligned}
    F(\theta; u, \epsilon, \rho) 
    &= 
    \nabla_{\theta}^T  \log \hat{L}_{U}^* (\theta, u) 
    + \epsilon \nabla_{\theta}^T \log \hat{P}_U(\theta, u) 
    + 
    \epsilon
    \left(
        \rho
        \spacednablatranspose{\theta}{T_U(\theta, u)}
          + 
        h(\theta, u, \rho)^T
    \right)
    ,
    \\
    h(\theta, u, \rho) 
    &= 
    q(\theta, u, \rho) - \rho \nabla_{\theta} T_U(\theta, u).
\end{aligned} 
\end{equation}

Suppose we set $\epsilon = 1/n$ and $u = u_n$.
For $\rho=0$, $F(\theta; u_n, 1/n, 0)$ reduces to the rescaled gradient of the first-order saddlepoint log-likelihood in \eqref{eqDisc:rescaledFirstOrderLogLik}, while for $\rho = 1/n$, $F(\theta; u_n, 1/n, 1/n)$ reduces to the rescaled gradient of the true log-likelihood \eqref{eqDisc:rescaledTrueLogLik}.
From \citet[Theorem 2]{Goodman2022}, the true and saddlepoint log-likelihood functions have unique  maximizers provided that $n$ is sufficiently large and $(\theta,u_n)$ is restricted to a sufficiently small neighbourhood of $(\theta_0,u_0)$. It follows that  $\hat{\theta}_{\mathrm{spa}, n}$ and $\hat{\theta}_{\mathrm{true}, n}$, 
are solutions of $F(\theta; u_n, 1/n, 0) = 0$ and $F(\theta; u_n, 1/n, 1/n) = 0$, respectively.

\begin{lemma}
\label{lemmaDisc:Theorem1ContinousDiff}
Let $F(\theta; u, \epsilon, \rho)$ be defined as in \eqref{eqDisc:Proof1_1}. 
For $(\theta; u, \epsilon, \rho)$ in a sufficiently small neighbourhood of $(\theta_0; u_0, 0, 0)$, the following hold:

    \lemmasubitem{lemmaDisc:continousDiff_of_h_dh}
    $\nabla_{\theta} T_U(\theta, u)$ and $\nabla_{\theta} F(\theta; u, \epsilon, \rho)$ are continuous, and both $h(\theta, u, \rho)$ and its derivative $\frac{\partial h}{\partial \rho}(\theta, u, \rho)$ are continuously differentiable.

    \lemmasubitem{lemma1Disc:NonSinglarity}
    $F(\theta_0; u_0, 0, 0) = 0$,  $\nabla_{\theta} F(\theta_0; u_0, 0, 0)$ is non-singular, and $\nabla_{\theta} F(\theta; u, \epsilon, \rho)$ remains non-singular in the neighbourhood of $(\theta_0; u_0, 0, 0)$. Additionally, there exists a continuously differentiable function $G_{\theta}(u,\epsilon,\rho)$ satisfying 
        \begin{equation}
        \label{eqDisc:Proof1_ImplictFunction}
        F(G_{\theta}(u,\epsilon,\rho); u, \epsilon, \rho) = 0,
        \end{equation}
    and $\frac{\partial G_{\theta}}{\partial \rho} (u, \epsilon, \rho)$ is also continuously differentiable.
    Indeed, $T_U$, $h$ and $F$ are of class $C^3$, i.e., three times continuously differentiable. 
    
    \lemmasubitem{lemma1Disc:h_grad_h_continouslydiff}
    $\frac{\partial h}{\partial \rho}(\theta, u, 0) = 0$ for all $\theta$, $u$.
    
\end{lemma}
We defer the proof to 
\supsecref{SecDisc:Theorem1ApproximationFormulaProof}.

Since $F(\theta; u_n, 1/n, 1/n)$ and $F(\theta; u_n, 1/n, 0)$ reduce to the rescaled gradients of the true log-likelihood and 
first-order saddlepoint log-likelihoods respectively, we see that
\begin{equation}
\label{eqDisc:theta_spa_true_G_theta}
    G_{\theta}(u_n,1/n,1/n) = \hat{\theta}_{\mathrm{true},n}
    ~\text{and}~
    G_{\theta}(u_n,1/n,0) = \hat{\theta}_{\mathrm{spa},n}.
\end{equation}
Additionally, 
from 
\lemmaitemref{lemmaDisc:Theorem1ContinousDiff}{lemma1Disc:NonSinglarity}, 
$G_{\theta}(u,\epsilon,\rho)$ is the solution of $\theta$ in the neighbourhood of $\theta_0$, thus at the base point $(u_0,0,0)$,
$G_{\theta}(u_0,0,0) = \theta_0$ since $F(\theta_0; u_0, 0, 0) = 0$.

We seek to examine the discrepancy, denoted as $\delta_n$, in $\theta$ estimates derived from solving $F(\theta; u_n, 1/n, 0) = 0$ instead of $F(\theta; u_n, 1/n, 1/n) = 0$. By \eqref{eqDisc:theta_spa_true_G_theta}, this difference is 
\begin{equation}
\label{eqDisc:Proof1_trueDiscrepancy}
\delta_n =
G_{\theta}(u_n,1/n,1/n) - G_{\theta}(u_n,1/n,0)
.
\end{equation}

\begin{lemma}
\label{lemmaDisc:LemmaOnTheDiscrepancyFormular} 
The discrepancy approximation $\hat{\delta}_n$ as defined in \eqref{eqDisc:DiscrepancyFormula} is 
\begin{equation}
\label{eqDisc:Proof1_4}
\hat{\delta}_n = \left(\frac{1}{n} \right)\frac{\partial G_{\theta}}{\partial \rho} (u_n,1/n, 0).
\end{equation}
\end{lemma}

We will use \eqref{eqDisc:Proof1_trueDiscrepancy} and \eqref{eqDisc:Proof1_4} to verify our results. 
First, we will prove Lemma~\ref{lemmaDisc:LemmaOnTheDiscrepancyFormular} by examining $\epsilon \frac{\partial G_{\theta}}{\partial \rho} (u,\epsilon, 0)$, then continue to verify \eqref{Theorem1_results1}--\eqref{Theorem1_results3}.

By implicit partial differentiation of \eqref{eqDisc:Proof1_ImplictFunction},
\begin{equation}
\label{eqDisc:Theorem1Implicit_fully_expressed}
     \frac{\partial F}{\partial \rho} (G_{\theta}(u,\epsilon,\rho); u, \epsilon, \rho)
    +
    \nabla_{\theta} F(G_{\theta}(u,\epsilon,\rho); u, \epsilon, \rho)
    \frac{\partial G_{\theta}}{\partial \rho} (u,\epsilon,\rho) = 0.
\end{equation}
Rearranging  this, we get
\begin{equation}
\label{eqDisc:dG_theta_d_rho}
    \frac{\partial G_{\theta}}{\partial \rho} (u,\epsilon,\rho) =
     - (\nabla_{\theta} F(G_{\theta}(u,\epsilon,\rho); u, \epsilon, \rho))^{-1}
     \frac{\partial F}{\partial \rho} (G_{\theta}(u,\epsilon,\rho); u, \epsilon, \rho)
     . 
\end{equation}
Using \eqref{eqDisc:Proof1_1} and \lemmaitemref{lemmaDisc:Theorem1ContinousDiff}{lemma1Disc:h_grad_h_continouslydiff},
$
    \frac{\partial F}{\partial \rho} (\theta; u, \epsilon, 0)
    =
    \epsilon
    \spacednablatranspose{\theta}{T_U(\theta, u)}
    .    
$
Thus for $\rho=0$, \eqref{eqDisc:dG_theta_d_rho} simplifies to
\begin{equation}
    \frac{\partial G_{\theta}}{\partial \rho} (u,\epsilon,0) 
    =
    - \epsilon (\nabla_{\theta} F(G_{\theta}(u,\epsilon,0); u, \epsilon, 0))^{-1}
    \spacednablatranspose{\theta}{T_U(G_{\theta}(u,\epsilon,0), u)}
     \label{eqDisc:Proof1_DerivativeImplicitFn}
     .
\end{equation}
We can now determine $\epsilon \tfrac{\partial G_{\theta}}{\partial \rho} (u,\epsilon, 0)$ to be
\begin{equation}
    \epsilon \frac{\partial G_{\theta}}{\partial \rho} (u,\epsilon,0) 
    =
    - \epsilon^2 (\nabla_{\theta} F(G_{\theta}(u,\epsilon,0); u, \epsilon, 0))^{-1}
    \spacednablatranspose{\theta}{T_U(G_{\theta}(u,\epsilon,0), u)}
     \label{eqDisc:Proof1_DerivativeImplicitFn_epsilon}
     .
\end{equation}

\begin{proof}[\proofof{Lemma~\ref{lemmaDisc:LemmaOnTheDiscrepancyFormular}}]
    Setting $\epsilon = 1/n$ and $u = u_n$ in \eqref{eqDisc:Proof1_DerivativeImplicitFn_epsilon}, we have
    \begin{equation*}
    \begin{aligned}
        \left(
        \frac{1}{n}
        \right)
        \frac{\partial G_{\theta}}{\partial \rho} (u_n,1/n, 0)
        &= 
        - \left( 
              \frac{1}{n^2} 
          \right)
          \left(
                \nabla_{\theta} F(G_{\theta}(u_n,1/n,0); u_n, 1/n, 0)
          \right)^{-1} 
          \spacednablatranspose{\theta}{T_U(G_{\theta}(u_n,1/n,0), u_n)}
        \\
        &=
        -  \left(
                n \nabla_{\theta} F(\hat{\theta}_{\mathrm{spa},n}; u_n, 1/n, 0)
          \right)^{-1} 
          \left(\frac{1}{n}\right)
          \spacednablatranspose{\theta}{T_U(\hat{\theta}_{\mathrm{spa},n}, u_n)}
          . 
    \end{aligned}
    \end{equation*}
    Using \eqref{eqDisc:TX_TU}, 
    $\nabla_{\theta} T_U(\hat{\theta}_{\mathrm{spa},n}, u_n) /n = \nabla_{\theta} T_X(\hat{\theta}_{\mathrm{spa},n}, x_n)$, and by \eqref{eqDisc:Proof1_1} and \eqref{eqDisc:rescaledFirstOrderLogLik}, 
    $
    n F(\hat{\theta}_{\mathrm{spa},n}; u_n, 1/n, 0) = 
    \nabla_\theta^T
    \log \hat{L}_n(\hat{\theta}_{\mathrm{spa},n};x_n)
    $. Therefore, 
    \begin{equation}
        \left(
        \frac{1}{n}
        \right)
        \frac{\partial G_{\theta}}{\partial \rho} (u_n,1/n, 0)
        =
        -(\nabla_{\theta}^T \nabla_{\theta} \log \hat{L}_n(\hat{\theta}_{\mathrm{spa},n};x_n))^{-1}
        \spacednablatranspose{\theta}{T_X(\hat{\theta}_{\mathrm{spa},n}, x_n)} 
        =
        \hat{\delta}_n.
    \end{equation}
    This completes the proof of Lemma~\ref{lemmaDisc:LemmaOnTheDiscrepancyFormular}. 
\end{proof}

To verify \eqref{Theorem1_results2}, we proceed as follows. 
\begin{equation}
\label{eqDisc:delta_n_F_G_u_n}
\begin{aligned}
    \hat{\delta}_n
    &=
    \left(
    \frac{1}{n}
    \right)
    \frac{\partial G_{\theta}}{\partial \rho} (u_n,1/n, 0),
    \\
    &=
    -
    \left(
    \frac{1}{n^2}
    \right)
    \nabla_{\theta} F(G_{\theta}(u_n,1/n,0); u_n, 1/n, 0)^{-1}
    \spacednablatranspose{\theta}{T_U(G_{\theta}(u_n,1/n,0), u_n)} 
     .
\end{aligned}
\end{equation}
The assumptions imply that $u_n = u_0 + o_P(1)$, and since $G_\theta$ is continuously differentiable by \lemmaitemref{lemmaDisc:Theorem1ContinousDiff}{lemma1Disc:NonSinglarity},
\begin{equation}
    \hat{\theta}_{\mathrm{spa},n}
    =
    G_{\theta}(u_n,1/n,0) 
    = 
    G_{\theta}(u_0,0,0) + O_P(|u_n-u_0| + 1/n)  
    =
    \theta_0 + O_P(|u_n-u_0| + 1/n) 
    .
\end{equation}
Similarly, since 
$\nabla_{\theta} F$ and $\nabla_{\theta} T_U$ are continuously differentiable, \eqref{eqDisc:delta_n_F_G_u_n} 
can be expressed as
\[
    \hat{\delta}_n
    =
    -
    \left(
    \frac{1}{n^2}
    \right)
    \left(
    \nabla_{\theta} F(\theta_0; u_0, 0, 0)^{-1}
    \spacednablatranspose{\theta}{T_U(\theta_0, u_0)}
     +
    O_P(|u_n-u_0| + 1/n) 
     \right).
\]
Under the assumption $u_n=u_0+O_P(n^{-1/2})$, the error term $O_P(|u_n-u_0| + 1/n)$ reduces to $O_P(n^{-1/2})$.
Recognizing that $\nabla_\theta F(\theta_0; u_0, 0, 0) = \nabla_\theta^T \nabla_{\theta} \log \hat{L}_{U}^* (\theta_0, u_0)$, it follows that $\hat{\delta}_n = n^{-2} (k_0 + O_P(n^{-1/2}))$, as in \eqref{Theorem1_results2}.
Under the weaker assumption $u_n=u_0+o_P(1)$, the error term reduces to $o_P(1)$ instead.

Next, to assess the order of the difference between the true discrepancy and its approximation, we utilize Lagrange's form for the Taylor series remainder, see for instance \cite[Chapter~5]{WhittakerWatson1996}. Motivated by \eqref{eqDisc:Proof1_trueDiscrepancy} and \eqref{eqDisc:Proof1_4}, we write
\begin{align}
    G_{\theta}(u,\epsilon,\rho) - G_{\theta}(u,\epsilon,0) - \rho \cdot \frac{\partial G_{\theta}}{\partial \rho} (u,\epsilon, 0)
    =
        \tfrac{1}{2}\rho^2
        \cdot
        \frac{\partial^2 G_{\theta}}{\partial \rho^2} (u,\epsilon,\rho^{*})~
        \text{for some}~\rho^{*} \in \left[ 0,\rho \right].
        \label{eqDisc:LangrangeTaylorRemainderTheorem1}
\end{align} 
We examine the leading order behaviour of $\frac{\partial^2 G_{\theta}}{\partial \rho^2} (u,\epsilon,\rho)$ in the neighbourhood of $(u,0,0)$.
The next order derivative of \eqref{eqDisc:Theorem1Implicit_fully_expressed} yields
\begin{equation}
\label{eqDisc:SecondImpicitDerTheorem1}
\begin{aligned}
    &
    \frac{\partial^2 F}{\partial \rho^2}
    (G_{\theta}(u,\epsilon,\rho); u, \epsilon, \rho)
    +
    2
    \left(
        \bigl(\tfrac{\partial }{\partial \rho}
        \nabla_{\theta} F \bigr)
        (G_{\theta}(u,\epsilon,\rho); u, \epsilon, \rho)
        \frac{\partial G_{\theta}}{\partial \rho}(u, \epsilon, \rho) 
    \right)
    \\
    &\qquad\qquad
    +
    \sum_{i = 1}^p 
    \bigl(
    \tfrac{\partial}{\partial \theta_i} \nabla_{\theta} F
    \bigr) 
    (G_{\theta}(u,\epsilon,\rho); u, \epsilon, \rho)
    \frac{\partial G_{\theta}}{\partial \rho}(u, \epsilon, \rho)
    \frac{\partial G_{\theta}}{\partial \rho}(u, \epsilon, \rho)_i
    \\
    &\qquad\qquad\qquad
    +
    \nabla_{\theta} F
    (G_{\theta}(u,\epsilon,\rho); u, \epsilon, \rho)
    \frac{\partial^2 G_{\theta}}{\partial \rho^2}(u, \epsilon, \rho)
    = 0
\end{aligned}
\end{equation}
For all values of $(\theta, u)$ in the neighbourhood of $(\theta_0, u_0)$,
\begin{equation}
\label{eqDisc:d2Fdrho2}
     \frac{\partial^2 F}{\partial \rho^2}(\theta; u, \epsilon, \rho)
    =
    \epsilon
    \left(
    \frac{\partial^2 h}{\partial \rho^2}(\theta, u, \rho)
    \right)^T
    = O(\epsilon).
\end{equation}
The second and third terms involve the continuously differentiable function $\frac{\partial G_{\theta}}{\partial \rho}(u, \epsilon, \rho)$.
From \eqref{eqDisc:Proof1_DerivativeImplicitFn}, $\frac{\partial G_{\theta}}{\partial \rho}(u, \epsilon, 0) = O(\epsilon)$ for all $u$ in the neighbourhood of $u_0$. It follows that in the same neighbourhood,
\begin{equation}
\label{eqDisc:dGdrho_order_epsilon_rho}
    \frac{\partial G_{\theta}}{\partial \rho}(u, \epsilon, \rho) = O(\epsilon + \rho).
\end{equation}
Using \eqref{eqDisc:d2Fdrho2} and \eqref{eqDisc:dGdrho_order_epsilon_rho} in \eqref{eqDisc:SecondImpicitDerTheorem1}, and noting that $\nabla_{\theta} F(\theta; u, \epsilon, \rho)$ remains non-singular in the vicinity of $(\theta_0; u_0, 0, 0)$,  we conclude that 
\[
    \tfrac{\partial^2 G_{\theta}}{\partial \rho^2}(u, \epsilon, \rho) = O(\epsilon+\rho)~\text{for all values}~(u, \epsilon, \rho)
\]
in the neighbourhood of $(u_0, 0,0)$.
Therefore, we write 
\eqref{eqDisc:LangrangeTaylorRemainderTheorem1} as
\begin{align*}
    \tfrac{1}{2}\rho^2
    \cdot
    \tfrac{\partial^2 G_{\theta}}{\partial \rho^2} (u,\epsilon,\rho^{*})
    =
    O(\rho^2(\epsilon + \rho^*)).
\end{align*}
Evaluating this at $u = u_n$, $\epsilon = 1/n$ and $\rho = 1/n$ yields
$
    \delta_n - \hat{\delta}_n
    =
    O(n^{-3})
$, which verifies \eqref{Theorem1_results3}. 
Note that because both $\delta_n$ and $\hat{\delta}_n$ relate to the observed data value $u_n$, we do not obtain an error term arising from $|u_n-u_0|$.

Finally, \eqref{Theorem1_results1} follows by combining \eqref{Theorem1_results2} and \eqref{Theorem1_results3}.
\end{proof}
\begin{proof}[\proofof{Theorem~\ref{theoremDisc:Theorem1}}]
\label{prf:Theorem1Proof}
It suffices to show that the assumptions of Theorem~\ref{theoremDisc:Theorem2} apply.
The assumption that $M_U$ is defined in a neighbourhood of $t=0$ implies that $U$ has finite variance, so $u_n=u_0+O_P(n^{-1/2})$ follows from the central limit theorem.
Since Theorem~\ref{theoremDisc:Theorem1} sets $ u_0 = E(U_{\theta_0})$, the corresponding solution of the saddlepoint equation \eqref{eqDisc:Theorem2SaddlepointEqn} is $t_0 = 0$. 
Then, as we explain in \supsecref{labelDisc:ConditionNegDefin}, condition 
\eqref{eqDisc:gradientLeadingOrder} holds automatically, and condition \eqref{eqDisc:HessianLeadingOrder} reduces to condition \eqref{eqDisc:Theorem1NonSingularCondition}.
Consequently, Theorem~\ref{theoremDisc:Theorem1} follows from Theorem~\ref{theoremDisc:Theorem2}. 
\end{proof}

\vspace*{-6mm}
\bibliographystyle{apalike}
\bibliography{paper-ref}

\newpage
\section*{Supplementary material for ``What is the price of approximation? The saddlepoint approximation to a likelihood function''}

\beginSupplementCounter

\section{General notations and derivatives}
\label{secDisc:GenralNotatiosAndDer}

In addition to the notation from
Section~\ref*{secDisc:Notations}, this section collects some common gradients and Hessians relevant to the theoretical and computational components used in the subsequent derivations. 

As with $\nabla_\theta$, gradient vectors with respect to $t$ are denoted $\nabla_t$, and are vectors with the shape of $t^T$. For a scalar-valued function $f(\theta,t)$, where $\theta \in \mathbb{R}^{p \times 1}$ and $t \in \mathbb{R}^{1 \times d}$, we have $\spacednabla{t}{f(\theta,t)} \in \mathbb{R}^{d \times 1}$. Similarly,  
$\spacednablatranspose{t}{f(\theta,t)} \in \mathbb{R}^{1 \times d}$.
The function $\hat{t}(\theta; x)$ is a $d$-dimensional row-vector-valued function, where
$x \in \mathbb{R}^{d \times 1}$. 
To align with our conventions, we convert its values to column vectors, and express its gradient with respect to $\theta$ as $\nabla_{\theta} \hat{t}^T(\theta; x) \in \mathbb{R}^{d \times p}$. Additionally, we introduce the gradient operation $\nabla_x$, analogous to $\nabla_{\theta}$ and $\nabla_{t}$, resulting in $\nabla_x \hat{t}^T(\theta; x) \in \mathbb{R}^{d \times d}$.

Given that    
\begin{equation}
\label{eqDisc:SaddEqnAppendix}
    K'_{X}(\hat{t}(\theta; x);\theta) = x,
\end{equation}
implicit partial differentiation yields
\begin{align}
\notag
\nabla_\theta K'_{X}\left(\hat{t}(\theta; x);\theta\right) + 
K''_{X}\left(\hat{t}(\theta; x);\theta\right) \nabla_\theta \hat{t}^T(\theta; x) &= 0,
\\ \label{eqDisc:deltheta_t}
\nabla_\theta \hat{t}^T(\theta; x) 
&= 
- 
K''_{X}\left(\hat{t}(\theta; x);\theta\right)^{-1} \nabla_\theta K'_{X}\left(\hat{t}(\theta; x);\theta\right)
.
\end{align}
Using \eqref{eqDisc:SaddEqnAppendix}, we also have that
\begin{equation}
\label{eqDisc:dxdt}
    \begin{aligned}
        I_{d \times d} = \nabla_x x =  
        \nabla_x 
        \left(
            K'_{X}(\hat{t}(\theta; x);\theta)
        \right)
        &=
        K''_{X}(\hat{t}(\theta; x);\theta)
        \nabla_x \hat{t}^T(\theta; x),
        \\
        \nabla_x \hat{t}^T(\theta; x) 
        &= K''_{X}(\hat{t}(\theta; x);\theta)^{-1}.
    \end{aligned}
\end{equation}
Using \eqref{eqDisc:deltheta_t},
it follows that the first and second-order gradients of the first two terms of the saddlepoint log-likelihood \eqref{eqDisc:SaddlepointLogLikelihood} has the following expressions.
\begin{equation}
\label{eqDisc:Grad_of_the_log_L_U_star}
    \begin{aligned}
        \nabla_{\theta} \log \hat{L}_{X}^* (\theta, x) 
        &= \nabla_\theta 
        \left(
            K_{X}(\hat{t}(\theta;x);\theta) - \hat{t}(\theta;x) x
        \right)
        \\
        &=
        K'_{X}(\hat{t}(\theta; x);\theta)^T \nabla_\theta \hat{t}^T(\theta; x) 
        +
        \nabla_\theta K_{X}(\hat{t}(\theta; x);\theta) - 
        x^T \nabla_\theta \hat{t}^T(\theta; x)
        \\
        &=
        \nabla_\theta K_{X}\left(\hat{t}(\theta; x);\theta\right).
    \end{aligned}
\end{equation}
\begin{equation}\label{eqDisc:Second-orderFirst2Terms}
    \begin{aligned}
    &\nabla_{\theta}^T \nabla_{\theta} \log \hat{L}_{X}^* (\theta, x)
        \\
        &\quad
        =
        \nabla_\theta
        \left(
            \nabla_\theta K_{X}\left(\hat{t}(\theta; x);\theta\right)
        \right)^T
        \\
        &\quad
        =
        \nabla_{\theta}^T \nabla_\theta K_{X}\left(\hat{t}(\theta; x);\theta\right) + 
        \left(\nabla_\theta K'_{X}\left(\hat{t}(\theta; x);\theta\right) \right)^T \nabla_\theta \hat{t}^T(\theta; x)
        \\
        &\quad
        =
        \nabla_{\theta}^T \nabla_\theta K_{X}\left(\hat{t}(\theta; x);\theta\right) 
       - \left(\nabla_\theta K'_{X}\left(\hat{t}(\theta; x);\theta\right) \right)^T
       K''_{X}\left(\hat{t}(\theta; x);\theta\right)^{-1} \nabla_\theta K'_{X}\left(\hat{t}(\theta; x);\theta\right).
    \end{aligned}
\end{equation}
Note that $K_X(0;\theta)=0$ for all $\theta$, and therefore 
\begin{equation}
\label{eqDisc:grad_at_0_of_K}
    \nabla_{\theta} K_X(0;\theta) = 0
    ~\text{and}~
    \nabla_{\theta}^T \nabla_{\theta} K_X(0;\theta) = 0.
\end{equation}

Finally, for an invertible matrix-valued  
function $R(s)$ with scalar argument $s$, the
derivative of its inverse is expressed as
\[ 
\frac{\partial }{\partial s} R\left(s\right)^{-1} 
= 
-R(s)^{-1} R'(s) R(s)^{-1}. 
\]
Analysing this at a scalar level for individual matrix elements, we obtain
\begin{equation}\label{eqDisc:GradInverseMat}
    \frac{\partial }{\partial s} R\left(s\right)_{ij}^{-1} = 
    -\sum_{k,l} R(s)^{-1}_{ik} 
    R'_{kl}(s)
    R(s)^{-1}_{lj},
\end{equation}
where $R(s)_{ij}^{-1}$ means the $i,j$ entry of the inverse of the matrix $R(s)$.

\section{Assumption \texorpdfstring{\eqref{eqDisc:HessianLeadingOrder}}{\ref*{eqDisc:HessianLeadingOrder}} and proof of Lemma \texorpdfstring{\ref{lemmaDisc:Theorem1ContinousDiff}}{\ref*{lemmaDisc:Theorem1ContinousDiff}}}

\subsection{Gradient and assumption \texorpdfstring{\eqref{eqDisc:HessianLeadingOrder}}{\ref*{eqDisc:HessianLeadingOrder}} in a well-specified model}
\label{labelDisc:ConditionNegDefin}

We begin by showing that assumption
\eqref{eqDisc:Theorem1NonSingularCondition} from Theorem~\ref{theoremDisc:Theorem1} is a special case of
the negative definiteness of the Hessian 
in assumption \eqref{eqDisc:HessianLeadingOrder} from Theorem~\ref{theoremDisc:Theorem2}.
Theorem~\ref{theoremDisc:Theorem1} uses
$u_0 = K_U'(0;\theta_0) = E(U_{\theta_0})$, so that
$t_0 = \hat{t}_U(\theta_0;u_0) = 0$.
By \eqref{eqDisc:Grad_of_the_log_L_U_star}--\eqref{eqDisc:grad_at_0_of_K}, the gradient and Hessian in 
\eqref{eqDisc:gradientLeadingOrder}--\eqref{eqDisc:HessianLeadingOrder} reduce to 
\[
\nabla_\theta \log \hat L_U^*(\theta_0,u_0)
= \nabla_\theta K_U(0;\theta_0)
= 0,
\]
\begin{equation}
    \label{eqDisc:WellSpecifiedNegDefiniteCondition}
    \nabla_{\theta}^T \nabla_{\theta} \log \hat{L}_{U}^* (\theta_0, u_0)= - (\nabla_{\theta} K'_{U}(0;\theta_0))^T K''_{U}(0;\theta_0)^{-1}
    \nabla_{\theta} K'_{U}(0;\theta_0).
\end{equation}
The right-hand side of \eqref{eqDisc:WellSpecifiedNegDefiniteCondition} is negative definite 
if and only if the condition \eqref{eqDisc:Theorem1NonSingularCondition} holds,
because $K''_{U}(0;\theta_0)$ is positive definite by assumption \eqref{eqDisc:NonSingularCovAndMean}.
Therefore, \eqref{eqDisc:HessianLeadingOrder} reduces to \eqref{eqDisc:Theorem1NonSingularCondition} when $u_0=E(U_{\theta_0})$, as claimed.

\subsection{Proof of Lemma~\ref{lemmaDisc:Theorem1ContinousDiff}}
\label{SecDisc:Theorem1ApproximationFormulaProof}
\begin{proof}[\proofof{Lemma~\ref{lemmaDisc:Theorem1ContinousDiff}}]
Recall 
\begin{align}
F(\theta; u, \epsilon, \rho) 
    = 
    \nabla_{\theta}^T  \log \hat{L}_{U}^* (\theta, u) + \epsilon \nabla_{\theta}^T \log \hat{P}_U(\theta, u) 
    + 
    \epsilon
    \left(
        \rho \spacednablatranspose{\theta}{T_U(\theta, u)}
        + h(\theta, u, \rho)^T
    \right),
\end{align}
where $h(\theta, u, \rho)$ is defined in terms of a continuously differentiable $q$ as follows:
\[
h(\theta, u, \rho) = 
q(\theta, u, \rho) 
- 
\rho \nabla_{\theta} T_U(\theta, u),
~\text{with $q(\theta, u, 0) = 0$}.
\]
We will first show that $F$ is continuously differentiable with respect to its parameters in the neighborhood $(\theta;u,\epsilon,\rho) = (\theta_0; u_0, 0,0)$.

The dependence on $u$ arises through the saddlepoint equation solution, $\hat{t}_U(\theta;u)$, whose derivatives with respect to $\theta$ and $u$ are continuously differentiable functions, with forms analogous to \eqref{eqDisc:deltheta_t} and \eqref{eqDisc:dxdt}, respectively.

The functions $\log \hat{L}_{U}^* (\theta, u)$ and $\log \hat{P}_U(\theta, u)$ are continuously differentiable (see \eqref{eqDisc:Second-orderFirst2Terms} and \citep[Appendix~B.2]{Goodman2022}).
By the definition in \eqref{eqDisc:funcT}, $T_U(\theta, u)$ is a composition of continuously differentiable functions of the CGF, and thus, is itself continuously differentiable.
Additionally, $h(\theta, u, \rho)$ is continuously differentiable due to the continuous differentiability of $q$ from \citet[Corollary~12]{Goodman2022}. Therefore, $\nabla_\theta F(\theta; u, \epsilon, \rho)$ is a composition of continuous functions.  
Furthermore, since $q$ is twice continuously differentiable \citep[Appendix~H]{Goodman2022}, it follows that  
\[
    \tfrac{\partial h}{\partial \rho}(\theta, u, \rho)
    ~\text{and}~
    \tfrac{\partial F}{\partial \rho}(\theta, u, \rho)
    ~\text{are also continuously differentiable.}
\]
This completes the proof for part (\ref{lemmaDisc:continousDiff_of_h_dh}).

For part (\ref{lemma1Disc:NonSinglarity}), 
we refer to assumptions \eqref{eqDisc:gradientLeadingOrder} and \eqref{eqDisc:HessianLeadingOrder}.
First, we observe that assumptions \eqref{eqDisc:gradientLeadingOrder} and \eqref{eqDisc:HessianLeadingOrder} correspond directly to $F(\theta_0; u_0, 0, 0) = 0$ and the non-singularity of $\nabla_{\theta} F(\theta_0; u_0, 0, 0)$, respectively. 
And since $F$ is continuously differentiable in the neighborhood of $(\theta_0; u_0, 0, 0)$, then by the Implicit Function Theorem, a continuously differentiable function $G_{\theta}(u,\epsilon, \rho)$ exists in the neighborhood of $(u_0, 0, 0)$ such that 
\[
    F(G_{\theta}(u, \epsilon, \rho); u, \epsilon, \rho) = 0.
\] 
Further, the continuous differentiability of $\tfrac{\partial G_{\theta}}{\partial \rho} (u, \epsilon, \rho)$ follows from that of $\tfrac{\partial F}{\partial \rho}(\theta, u, \rho)$.

More generally, the argument of \citet[Appendix H]{Goodman2022} can be repeated to show that if the bounds \eqref{eqDisc:RegularityCondition2} hold then the function $q$ is $C^3$ instead of merely continuously differentiable. 
Similarly, $K_U$ and the derivatives of $K_U$ that appear in $\hat{t}_U$, $T_U$,
$\log \hat{L}_{U}^*$, and $\log \hat{P}_U$ are also $C^3$.
Hence $F$, and therefore also $G_{\theta}$, are also $C^3$ instead of merely continuously differentiable. We remark that we will make only limited use of the stronger $C^3$ property, and most of the proofs use only continuous differentiability.

For part (\ref{lemma1Disc:h_grad_h_continouslydiff}), we use a function $r(\theta,u,\rho)$ of the form introduced in Proposition 10 of \citet{Goodman2022} and shown to be $C^2$ in Appendix~H of that work. The function is defined such that
\begin{equation}
\label{eqDisc:r_def}
    \log P_{n}(\theta, u) - \log \hat{P}_U(\theta, u) = r(\theta, u, 1/n), ~ \text{and} ~ r(\theta, u, 0) = 0.
\end{equation}
Comparing this with \eqref{eqDisc:AsyptoticTrueLoglik}, we obtain
\begin{equation}
\label{eqDisc:r_expansion}
    r(\theta, u, 1/n) = \tfrac{1}{n}T_U(\theta, u) + O(1/n^2).
\end{equation}
Since $r(\theta,u,\rho)$ is continuously differentiable, the derivative $\frac{\partial r}{\partial \rho} (\theta, u, \rho)$ exists, thus 
\begin{equation}
\label{eqDisc:r_derivative}
    \lim_{n \to \infty} \frac{r(\theta, u, 1/n) - r(\theta, u, 0)}{1/n}
    =
    T_U(\theta, u)
    ,
\end{equation}
implies that 
$
    \frac{\partial r}{\partial \rho} (\theta, u, 0) = T_U(\theta, u).
$

Now, taking the gradient of \eqref{eqDisc:r_def} with respect to $\theta$, we have
\begin{equation}
\label{eqDisc:grad_theta_and_q}
    \nabla_\theta \log P_{n}(\theta, u) - \nabla_\theta \log \hat{P}_U(\theta, u) = 
    \spacednabla{\theta}{r}(\theta, u, 1/n)
\end{equation}
The left-hand side of \eqref{eqDisc:grad_theta_and_q} defines $q(\theta,u,1/n)$ (see \eqref{eqDisc:qFunctionDefinition}), therefore,
\begin{equation}
\label{eq:grad_r_equals_q}
\spacednabla{\theta}{r}(\theta, u, 1/n) = q(\theta, u, 1/n).
\end{equation}
Since $q$ is continuously differentiable, at $\rho = 0$,
\begin{equation}
\begin{aligned}
    \frac{\partial q}{\partial \rho} (\theta, u, 0)
    &=
    \lim_{n \to \infty} \frac{q(\theta, u, 1/n) - q(\theta, u, 0)}{1/n},
    \\
    &=
    \lim_{n \to \infty} \frac{\spacednabla{\theta}{r}(\theta, u, 1/n) - \spacednabla{\theta}{r}(\theta, u, 0)}{1/n},
    \\
    &=
    \frac{\partial}{\partial \rho} 
    \left(
    \spacednabla{\theta}{r}(\theta, u, \rho) 
    \right)
    \Big|_{\rho=0}.
\end{aligned}
\end{equation}
As $r$ is $C^2$, we can interchange the order of differentiation
\begin{equation}
\label{eqDisc:dq_d_rho_equals_grad_T}
    \frac{\partial q}{\partial \rho} (\theta, u, 0) 
    = 
    \frac{\partial}{\partial \rho} 
    \left(
    \spacednabla{\theta}{r}(\theta, u, \rho) 
    \right)
    \Big|_{\rho=0}
    =
    \nabla_\theta
    \left(
    \frac{\partial r}{\partial \rho} (\theta, u, 0)
    \right)
    =
    \nabla_\theta T_U(\theta, u).
\end{equation}

To complete the proof, we differentiate $h(\theta,u,\rho)$ and evaluate it at $\rho = 0$. 
\begin{align}
    \frac{\partial h}{\partial \rho}(\theta, u, 0) = \frac{\partial q}{\partial \rho} (\theta, u, 0) - \nabla_{\theta} T_U(\theta, u),
\end{align}
and using \eqref{eqDisc:dq_d_rho_equals_grad_T}, $\frac{\partial h}{\partial \rho}(\theta, u, 0) = 0$.
\end{proof}

\section{Proof of Theorem~\ref{theoremDisc:Theorem3}}
\label{ProofTheorem3}
We aim at verifying the discrepancy results \eqref{eqDisc:Theorem3DiscTrue}-\eqref{eqDisc:OmegaDiscDiffSpecialCase}. We validate the assumption of non-degeneracy for the local MLEs $v_0$ and $\tau_0$, then, analogously to the proof of Theorem~\ref{theoremDisc:Theorem1}, apply the Implicit Function Theorem and obtain the discrepancy results.

Motivated by the transformations $x_n = n u_n$ and $u_n = u_0 + n^{-1/2}z_0$ in \eqref{eqDisc:scalingOnObservationsTheorem3}, we write 
\begin{equation}
\label{eqDisc:ObservationXTheorem3}
x = \epsilon^{-2} u,
~~
u = u_0 + \epsilon z_0,
~\text{and}~
\theta = 
    \begin{pmatrix}
    \omega \\ 
    \tau 
    \end{pmatrix}
    =
    \begin{pmatrix}
    \omega_0 + \epsilon v \\ 
    \tau 
    \end{pmatrix}
.
\end{equation}
where $\epsilon$ is a scaling parameter.
Throughout this proof we assume that the quantities $x$, $\epsilon$, $u$, $u_0$, $z_0$, $\theta$, $v$, and $\tau$ are related as in \eqref{eqDisc:ObservationXTheorem3}.
For this proof, in contrast to the proof of Theorem~\ref{theoremDisc:Theorem2}, we now use $\epsilon=n^{-1/2}$ instead of $n^{-1}$ while we continue to use $\rho=n^{-1}$.
Recall that in the proof of Theorem~\ref{theoremDisc:Theorem2},
we used likelihood functions of a single parameter vector, $\theta$, and achieved well-behaved functions in the limit $n \to \infty$ by applying a scaling of $n^{-1}$.
However, the current results introduces two components of the parameter space, $\omega$ and $\tau$.
Notably, the scaling for these components differs because $\omega$ depends on the unknown parameter $v$, while $\tau$ does not.

Consider the following column vector-valued functions defined for $\epsilon \neq 0$, of dimensions $p_1$ and $p_2$, respectively.
\begin{equation}
\label{eqDisc:F_omega}
    \begin{aligned}
        F_{\omega}(v, \tau; \epsilon, \rho) 
        = 
        \epsilon^{-1} \nabla_{\omega}^T 
        & \log \hat{L}_{U}^* (\theta, u) + 
        \epsilon \nabla_{\omega}^T \log \hat{P}_U(\theta,u) +
        \epsilon
        \left(
        \rho \nabla_{\omega}^T T_U(\theta, u) + g_{\omega}(\theta,u,\rho)^T
        \right),
        \\
        & g_{\omega}(\theta,u,\rho) =  q_{\omega}(\theta,u,\rho) - \rho \nabla_{\omega} T_U (\theta, u).
    \end{aligned}
\end{equation}

\begin{equation}
\label{eqDisc:F_tau}
    \begin{aligned}
        F_{\tau}(v, \tau; \epsilon, \rho) = 
        \epsilon^{-2} \nabla_{\tau}^T 
        & \log \hat{L}_{U}^* (\theta, u) + 
        \nabla_{\tau}^T \log \hat{P}_U(\theta,u) +
        \rho \nabla_{\tau}^T T_U(\theta, u) + g_{\tau}(\theta,u,\rho)^T,
        \\
        & g_{\tau}(\theta,u,\rho) =  q_{\tau}(\theta,u,\rho) - \rho \nabla_{\tau} T_U (\theta, u).
    \end{aligned}
\end{equation}
The terms $q_\omega$ and $q_\tau$ mean the continuously differentiable vector-valued function $q(\theta, u, \rho)$ defined in \eqref{eqDisc:qFunctionDefinition}, split into two blocks of dimension $p_1$ and $p_2$, respectively.

Suppose we set $\epsilon = 1/n^{1/2}$ and $\rho = 1/n$. $F_{\omega}$ and $F_{\tau}$ 
relate to the gradients of the true log-likelihood 
$\log L_n(\theta; x_n)$ as follows:
\begin{equation}
\begin{aligned}
    F_{\omega}(v, \tau; 1/n^{1/2}, 1/n) &= \tfrac{1}{n^{1/2} } \nabla_{\omega}^T \log L_n(\theta; x_n) \quad \text{(rescaled gradient)},
    \\
    F_{\tau}(v, \tau; 1/n^{1/2}, 1/n) &= \nabla_{\tau}^T \log L_n(\theta; x_n) \quad \text{(unscaled gradient)}.
\end{aligned}
\end{equation}
On the other hand, for $\rho = 0$, these functions relate to the first-order saddlepoint log-likelihoods:
\begin{equation}
\label{eqDisc:F_omega_F_tau_log_xn}
\begin{aligned}
    F_{\omega}(v, \tau; 1/n^{1/2}, 0) &= \tfrac{1}{n^{1/2} } \nabla_{\omega}^T \log \hat{L}_n(\theta;x_n) \quad \text{(rescaled gradient)},
    \\
    F_{\tau}(v, \tau; 1/n^{1/2}, 0) &= \nabla_{\tau}^T \log \hat{L}_n(\theta;x_n) \quad \text{(unscaled gradient)}.
\end{aligned}
\end{equation}
Therefore, to examine the discrepancies in $\omega$ and $\tau$ estimates, we focus on identifying the critical points for which $F_{\omega} = 0$ and $F_{\tau} = 0$ when $(\epsilon, \rho) = (1/n^{1/2}, 0)$ instead of $(1/n^{1/2}, 1/n)$ as $n \to \infty$.
\begin{proof}[\proofof{Theorem~\ref{theoremDisc:Theorem3}}]
\label{prf:Theorem3Proof}
Define a continuously differentiable function $F$ by
\begin{equation}\label{eqDisc:F_combined_omega_tau}
    F
    \left(
    \smallmat{v \\ \tau } 
    ; \epsilon, \rho
    \right) 
    = 
    \begin{pmatrix}
        F_{\omega}(v, \tau; \epsilon, \rho) \\
        F_{\tau}(v, \tau; \epsilon, \rho)
    \end{pmatrix}
    \in \mathbb{R}^{(p_1 + p_2) \times 1}
    .
\end{equation}

\begin{lemma}
\label{lemmaDisc:Theorem3ContinousDiff}
The function $F$ as defined in \eqref{eqDisc:F_combined_omega_tau}, can be extended to a neighborhood of 
$
    \left(
    \smallmat{v_0 \\ \tau_0 } 
    ; 0, 0
    \right)
$ in such a way that the following hold:

    \lemmasubitem{lemmaItemNonSingularTheorem3ContinousDiff} 
    The function $F$ is continuously differentiable,
    $
    F \left( \smallmat{v_0 \\ \tau_0 } ; 0, 0 \right) = 0$, and
    $
    \nabla_{v, \tau} F \left( \smallmat{v_0 \\ \tau_0 } ; 0, 0 \right)
    $ is non-singular.

    \lemmasubitem{lemmaItem_g_omega_tau_differentiable}
    $(\theta_0, u_0, 0)$, 
    $g_{\omega}(\theta,u,\rho)$,
    $\tfrac{\partial g_{\omega}}{\partial \rho}(\theta,u,\rho)$,
    $g_{\tau}(\theta,u,\rho)$,
    and
    $\frac{\partial g_{\tau}}{\partial \rho}(\theta,u,\rho)$ are continuously differentiable. Further, $\frac{\partial g_{\omega}}{\partial \rho}(\theta,u,0) = 0$ and $\frac{\partial g_{\tau}}{\partial \rho}(\theta,u,0) = 0$.

    \lemmasubitem{lemmaItem_dG_drho_continuous_diff_Theorem3}
    There exist a continuously differentiable function
    $
    G(\epsilon, \rho)
    =
    \begin{pmatrix}
        G_v(\epsilon, \rho) \\ 
        G_{\tau}(\epsilon, \rho) 
    \end{pmatrix} 
    \in \mathbb{R}^{(p_1 + p_2) \times 1}
    $ 
    defined in the neighborhood of $(\epsilon, \rho) = (0, 0)$ such that 
    \begin{equation}\label{eqDisc:F_combined_omega_tau_at_0}
        F(G(\epsilon, \rho); \epsilon, \rho) = 
        \begin{pmatrix}
            F_{\omega}(G_v(\epsilon, \rho), G_{\tau}(\epsilon, \rho); \epsilon, \rho) \\
            F_{\tau}(G_v(\epsilon, \rho), G_{\tau}(\epsilon, \rho); \epsilon, \rho)
        \end{pmatrix}
        = 0,
    \end{equation}
    and further $\frac{\partial G}{\partial \rho}(\epsilon, \rho)$ is continuously differentiable in the same neighborhood.
    Indeed, $F$ and $G$ are of class $C^3$, i.e., three times continuously differentiable.

    \lemmasubitem{lemmaItem_grad_tau_F_omega}
    Abbreviate $Q = K''_U\left(0; \theta_0\right)^{-1}$ and
    let $W$ be the $d \times p_2$ matrix whose $i^{\mathrm{th}}$ column is
    \[
    \frac{\partial K''_U}{\partial \tau_{i}}
        \left(
        0; 
        \theta_0
        \right)
        \left(
        Q - Q B 
        (B^{T} Q B)^{-1}
        B^T Q
        \right)z_0
    \]
    for all $i = 1,\dots,p_2$. 
    Then
    $
        \nabla_{\tau} F_{\omega}
        \left(
        v_0, \tau_0; 0, 0
        \right) 
        =
        - 
        B^{T}
        Q
        W.
    $
    Additionally, 
    $\nabla_{v} F_{\omega}(v_0, \tau_0; 0, 0) 
    = 
    -
    B^{T}
    K''_U
    \left(0; 
        \theta_0
    \right)^{-1}
    B
    $.
\end{lemma}
We defer the proof to \supsecref{appDisc:DerivativesFomegaFtau}.

From \lemmaitemref{lemmaDisc:Theorem3ContinousDiff}{lemmaItem_dG_drho_continuous_diff_Theorem3}, 
$G(\epsilon, \rho)$ is the unique solution of \eqref{eqDisc:F_combined_omega_tau_at_0} 
in the neighborhood of
$
\smallmat{v_0 \\ \tau_0 }
$. 
Thus, 
$G(0, 0) = 
\smallmat{
    G_v(0,0) \\ G_{\tau}(0,0)
}
=
\smallmat{v_0 \\ \tau_0 }$ since $F \left( \smallmat{v_0 \\ \tau_0 } ; 0, 0 \right) = 0$.

Recall that the gradients of the true likelihoods \eqref{eqDisc:F_omega}-\eqref{eqDisc:F_tau} correspond to those of the first-order saddlepoint log-likelihoods at $\rho = 0$, whereas for $\rho = 1/n$, they correspond the true log-likelihoods.
From \citet[Theorem 6]{Goodman2022}, the true and saddlepoint log-likelihood functions have unique maximizers provided that $n$ is sufficiently large and $\theta$ is restricted to a sufficiently small neighborhood of $\theta_0$. Recall also that $\omega$ and $v$ are related as in \eqref{eqDisc:ObservationXTheorem3}.
It follows that,
\begin{align}
    &\hat{\theta}_{\mathrm{true},n}
    =
    \begin{pmatrix}
        \hat{\omega}_{\mathrm{true},n} \\
        \hat{\tau}_{\mathrm{true},n}
    \end{pmatrix}
    =
    \begin{pmatrix}
         \omega_0 + (1/n^{1/2}) G_{v}(1/n^{1/2}, 1/n) \\
         G_{\tau}(1/n^{1/2}, 1/n)
    \end{pmatrix}
    =
    G_\theta(1/n^{1/2}, 1/n),
    \\
    &\hat{\theta}_{\mathrm{spa},n}
    =
    \begin{pmatrix}
        \hat{\omega}_{\mathrm{spa},n} \\
        \hat{\tau}_{\mathrm{spa},n}
    \end{pmatrix}
    =
    \begin{pmatrix}
         \omega_0 + (1/n^{1/2}) G_{v}(1/n^{1/2},0) \\
         G_{\tau}(1/n^{1/2},0)
    \end{pmatrix}
    =
    G_\theta(1/n^{1/2}, 0)
    ,
\end{align}
where in accordance with \eqref{eqDisc:ObservationXTheorem3}, we introduce the abbreviation 
\begin{equation}
     G_{\theta}(\epsilon, \rho)
    =
    \begin{pmatrix}
        \omega_0 + \epsilon G_{v}(\epsilon,\rho) \\
        G_{\tau}(\epsilon,\rho)
    \end{pmatrix}.
\end{equation}
Therefore,
\begin{equation}
\label{eqDisc:TrueDiscTheorem3}
    \begin{pmatrix}
        \delta_{n,\omega} \\
        \delta_{n,\tau}
    \end{pmatrix}
    = 
    \begin{pmatrix}
        (1/n^{1/2}) 
        \left(
        G_{v}(1/n^{1/2},1/n) - G_{v}(1/n^{1/2},0)
        \right) 
        \\
        G_{\tau}(1/n^{1/2},1/n) - G_{\tau}(1/n^{1/2},0)
    \end{pmatrix}
    .
\end{equation}

\begin{lemma}
\label{lemmaDisc:Theorem3approximated_delta}
    The discrepancy approximation as defined in \eqref{eqDisc:DiscrepancyFormula} is given by
    \begin{equation}
    \label{eqDisc:ApprtDiscTheorem3}
        \begin{pmatrix}
            \hat{\delta}_{n,\omega} \\
            \hat{\delta}_{n,\tau}
        \end{pmatrix}
        =
        \begin{pmatrix}
            (1/n^{3/2})  
            \frac{\partial G_{v}}{\partial \rho} (1/n^{1/2},0) 
            \\
            (1/n) 
            \frac{\partial G_{\tau}}{\partial \rho} (1/n^{1/2},0)
        \end{pmatrix}
        .
    \end{equation} 
\end{lemma}
We defer the proof to \supsecref{AppDisc:LemmaTheorem3ApprxFormula}

We will use \eqref{eqDisc:TrueDiscTheorem3} and \eqref{eqDisc:ApprtDiscTheorem3} to verify our results for this proof.

Motivated by \eqref{eqDisc:ApprtDiscTheorem3}, we will examine 
$
\smallmat{
        \epsilon \rho \frac{\partial G_{v}}{\partial \rho} (\epsilon,0) 
        \\
        \rho \frac{\partial G_{\tau}}{\partial \rho} (\epsilon,0)
}
$. 
To analyze $\frac{\partial G}{\partial \rho}(\epsilon, \rho) = 
\smallmat{
    \frac{\partial G_{v}}{\partial \rho} (\epsilon,\rho) 
    \\
    \frac{\partial G_{\tau}}{\partial \rho} (\epsilon,\rho)
}
$,  
we apply partial implicit differentiation to  \eqref{eqDisc:F_combined_omega_tau_at_0}, obtaining the following identity: 
\begin{align}
    \frac{\partial F}{\partial \rho}
    \left( G(\epsilon, \rho); \epsilon, \rho \right)
    +
    \nabla_{v,\tau}
    F(G(\epsilon, \rho); \epsilon, \rho)
    \frac{\partial G}{\partial \rho}(\epsilon, \rho)
    = 0.
    \label{eqDisc:FirstOrderImplicitPartial}
\end{align}
Here, $\nabla_{v,\tau}F$ denotes a $2 \times 2$ block matrix of the form
$
\smallmat{
        \nabla_{v} F_{\omega} & 
        \nabla_{\tau} F_{\omega}
        \\
        \nabla_{v} F_{\tau} & 
        \nabla_{\tau} F_{\tau}
    }
$.
Because $G(\epsilon, \rho)$ is continuously differentiable, then in a restricted neighborhood of $(\epsilon, \rho) = (0,0)$, $\frac{\partial G}{\partial \rho}(\epsilon, \rho)$ is bounded. This implies that
$\frac{\partial G}{\partial \rho}(\epsilon, \rho) = 
\smallmat{
    \frac{\partial G_{v}}{\partial \rho} (\epsilon,\rho) 
    \\
    \frac{\partial G_{\tau}}{\partial \rho} (\epsilon,\rho)
}
=
O(1)
$
within this neighborhood.

Now we shift the focus to identifying the leading term. 
Rearranging \eqref{eqDisc:FirstOrderImplicitPartial}, we have 
\begin{align}
\hspace{-0.5em}
    \begin{pmatrix}
        \frac{\partial G_{v}}{\partial \rho} (\epsilon,\rho) 
        \\
        \frac{\partial G_{\tau}}{\partial \rho} (\epsilon,\rho)
    \end{pmatrix}
    &=
    -(\nabla_{v, \tau} F(G(\epsilon, \rho); \epsilon, \rho))^{-1}
    \begin{pmatrix}
        \frac{\partial F_{\omega}}{\partial \rho}( G_{v}(\epsilon,\rho), G_{\tau}(\epsilon,\rho); \epsilon, \rho)
        \\
        \frac{\partial F_{\tau}}{\partial \rho}(G_{v}(\epsilon,\rho), G_{\tau}(\epsilon,\rho); \epsilon, \rho)
    \end{pmatrix}
    \notag
    \\
    &\hspace{-1em}=
    -(\nabla_{v, \tau} F(G(\epsilon, \rho); \epsilon, \rho))^{-1}
    \begin{pmatrix}
        \epsilon 
        \spacednablatranspose{\omega}{T_U
        \left(
        G_{\theta}(\epsilon, \rho)
        , u
        \right)}         
        + \epsilon \frac{\partial g_{\omega}}{\partial \rho}(G_{\theta}(\epsilon, \rho),u,\rho)^T
        \\
        \spacednablatranspose{\tau}{T_U
        \left(
        G_{\theta}(\epsilon, \rho), 
        u
        \right)}
        + 
        \frac{\partial g_{\tau}}{\partial \rho} (G_{\theta}(\epsilon, \rho),u,\rho)^T
    \end{pmatrix}
    .
    \label{eqDisc:grad_rho_G}
\end{align}
By \lemmaitemref{lemmaDisc:Theorem3ContinousDiff}{lemmaItem_g_omega_tau_differentiable}, the
derivatives $\frac{\partial g_{\omega}}{\partial \rho}$, $\frac{\partial g_{\tau}}{\partial \rho}$
vanishes when $\rho = 0$. Therefore,
\begin{equation}
    \begin{aligned}
        \begin{pmatrix}
        \frac{\partial G_{v}}{\partial \rho} (\epsilon,0) 
        \\
        \frac{\partial G_{\tau}}{\partial \rho} (\epsilon,0)
        \end{pmatrix}
        =
        -(\nabla_{v, \tau} F(G(\epsilon, 0); \epsilon, 0))^{-1}
        \begin{pmatrix}
            \epsilon 
            \spacednablatranspose{\omega}{T_U
            \left(
            G_{\theta}(\epsilon, 0)
            , u
            \right)}
            \\
            \spacednablatranspose{\tau}{T_U
            \left(
            G_{\theta}(\epsilon, 0), 
            u
            \right)}
        \end{pmatrix}
    \end{aligned}
\end{equation}
We can set $\epsilon=0$ and define
\begin{equation}
    \begin{aligned}
         \begin{pmatrix}
        k_{0,v}\\
        k_{0,\tau}
    \end{pmatrix}
    =
    \begin{pmatrix}
        \frac{\partial G_{v}}{\partial \rho} (0,0) 
        \\
        \frac{\partial G_{\tau}}{\partial \rho} (0,0)
    \end{pmatrix}
    =
    -
    \left(\nabla_{v, \tau} 
    F
    \left(
    \begin{pmatrix}
        v_0 \\ \tau_0
    \end{pmatrix};  0, 0
    \right)\right)^{-1}
    \begin{pmatrix}
        0_{p_1 \times 1} \\
        \spacednablatranspose{\tau}{T_U
        \left(
        \theta_0, 
        u_0 
        \right)}
    \end{pmatrix}
    .
    \end{aligned}
\end{equation}
In particular, 
\[
\frac{\partial G_{v}}{\partial \rho} (1/n^{1/2},0) = k_{0,v} + o(1)
~
\text{and}
~
 \frac{\partial G_{\tau}}{\partial \rho} (1/n^{1/2},0) = k_{0,\tau} + o(1).
\]
Together with \eqref{eqDisc:ApprtDiscTheorem3}, this verifies \eqref{eqDisc:Theorem3DiscApprt}.

Next, we examine the difference between the true discrepancy \eqref{eqDisc:TrueDiscTheorem3} and its approximation \eqref{eqDisc:ApprtDiscTheorem3}. 
By Lagrange’s form for the Taylor series remainder, 
for some $\rho^{*} \in \left[ 0,\rho \right],$
we have the following differences. 
\begin{align}
    &
    \begin{pmatrix}
        \epsilon 
        \left(
        G_{v}(\epsilon,\rho) - 
        G_{v}(\epsilon,0) -
        \rho \frac{\partial G_{v}}{\partial \rho} \left(\epsilon,0\right)
        \right)
        \\
        G_{\tau}(\epsilon,\rho) - 
        G_{\tau}(\epsilon,0) - 
        \rho \frac{\partial G_{\tau}}{\partial \rho} \left(\epsilon,0\right)
    \end{pmatrix}
    =
    \begin{pmatrix}
        \tfrac{1}{2}\epsilon\rho^2 \frac{\partial^2 G_{v}}{\partial \rho^2}(\epsilon, \rho^*)
        \\
        \tfrac{1}{2}\rho^2 \frac{\partial^2 G_{\tau}}{\partial \rho^2} \left(\epsilon, \rho^*\right)
    \end{pmatrix}
    .
    \label{eqDisc:DiscDifferenceTheorem3}
\end{align}
Setting $\epsilon = 1/n^{1/2}$ and $\rho = 1/n$ and using 
\lemmaitemref{lemmaDisc:Theorem3ContinousDiff}{lemmaItem_dG_drho_continuous_diff_Theorem3}, it follows that
\[
\begin{pmatrix}
        \delta_{n,\omega} \\
        \delta_{n,\tau}
    \end{pmatrix}
    -
    \begin{pmatrix}
        \hat{\delta}_{n,\omega} \\
        \hat{\delta}_{n,\tau}
    \end{pmatrix}
    =
    \begin{pmatrix}
        O(1/n^{5/2})\\
        O(1/n^{2})
    \end{pmatrix}
    ~\text{as}~n \to \infty,
\]
which verifies  \eqref{eqDisc:Theorem3DiscDiff}. The result \eqref{eqDisc:Theorem3DiscTrue} follows from \eqref{eqDisc:Theorem3DiscApprt} and \eqref{eqDisc:Theorem3DiscDiff}. This completes the proof for Theorem~\ref{theoremDisc:Theorem3}(a).

For Theorem~\ref{theoremDisc:Theorem3}(b),
assume \eqref{eqDisc:ZeroMatrixConditionTheorem3}.
We will 
verify the results in \eqref{eqDisc:ConstantNuDiagonalMatrixTheorem3}-\eqref{eqDisc:OmegaDiscDiffSpecialCase}.
From \eqref{eqDisc:ApprtDiscTheorem3}, we have
\begin{align}
\label{eqDisc:delta_hat_theorem3b}
    \hat{\delta}_{n,\omega} 
    =
    \frac{1}{n^{3/2}}
    \frac{\partial G_v}{\partial \rho}(1/n^{1/2},0). 
\end{align}
Motivated by this, we will examine the quantity $\epsilon^3\frac{\partial G_v}{\partial \rho}(\epsilon,0)$.

The factor 
$
    \nabla_{v,\tau}
    F(G(\epsilon, \rho); \epsilon, \rho)
$ in \eqref{eqDisc:FirstOrderImplicitPartial} is a block matrix of the form:
\begin{align}
    \begin{pmatrix}
        \nabla_{v} F_{\omega}\left(G_v(\epsilon, \rho), G_{\tau}(\epsilon, \rho); \epsilon, \rho\right) & 
        \nabla_{\tau} F_{\omega}\left(G_v(\epsilon, \rho), G_{\tau}(\epsilon, \rho); \epsilon, \rho\right)
        \\
        \nabla_{v} F_{\tau}\left(G_v(\epsilon, \rho), G_{\tau}(\epsilon, \rho); \epsilon, \rho \right) & 
        \nabla_{\tau} F_{\tau}\left(G_v(\epsilon, \rho), G_{\tau}(\epsilon, \rho); \epsilon, \rho \right)
    \end{pmatrix}.
    \label{eqDisc:BlockMatrixZeroMatrixCondition}
\end{align}
Evaluating \eqref{eqDisc:FirstOrderImplicitPartial} at $(\epsilon,\rho) = (0,0)$ and to find $\frac{\partial G_v}{\partial \rho}(0,0)$, we have 
\begin{align}
    \begin{pmatrix}
        0_{p_1 \times 1} \\
        \spacednablatranspose{\tau}{T_U
        \left(
        \theta_0, 
        u_0 
        \right)}
    \end{pmatrix}
    +
    \begin{pmatrix}
        \nabla_{v} F_{\omega}\left(v_0, \tau_0; 0, 0\right) & 
        \nabla_{\tau} F_{\omega}\left(v_0, \tau_0; 0, 0\right)
        \\
        \nabla_{v} F_{\tau}\left(v_0, \tau_0; 0, 0 \right) & 
        \nabla_{\tau} F_{\tau}\left(v_0, \tau_0; 0, 0 \right)
    \end{pmatrix}
    \begin{pmatrix}
        \frac{\partial G_v}{\partial \rho}(0,0)
        \\
        \frac{\partial G_{\tau}}{\partial \rho}(0,0)
    \end{pmatrix}
    = 0.
    \label{eqDisc:FirstOrderImplicitDerivAtBasepoint}
\end{align}
By \lemmaitemref{lemmaDisc:Theorem3ContinousDiff}{lemmaItem_grad_tau_F_omega} and the assumption \eqref{eqDisc:ZeroMatrixConditionTheorem3},
\begin{align}
    \nabla_{\tau} F_{\omega}
    \left(
    v_0, \tau_0; 0, 0
    \right) 
    =
    - B^{T}
    Q
    W = 0.
    \label{eqDisc:eqDisc:Condition2ndSectionTheorem3}
\end{align}
Additionally, $\nabla_{\tau} F_{\omega}$ and $\nabla_{v} F_{\tau}$ are transposes.
Therefore,
under condition \eqref{eqDisc:ZeroMatrixConditionTheorem3}, we can rewrite \eqref{eqDisc:FirstOrderImplicitDerivAtBasepoint} as
\begin{align}
    \begin{pmatrix}
        0_{p_1 \times 1} \\
        \nabla_{\tau} 
        T_U
        \left(
        \theta_0, 
        u_0 
        \right)^T 
    \end{pmatrix}
    +
    \begin{pmatrix}
        \nabla_{v} F_{\omega}\left(v_0, \tau_0; 0, 0\right) & 
        0_{p_1 \times p_2}
        \\
        0_{p_2 \times p_1} & 
        \nabla_{\tau} F_{\tau}\left(v_0, \tau_0; 0, 0 \right)
    \end{pmatrix}
    \begin{pmatrix}
        \frac{\partial G_v}{\partial \rho}(0,0)
        \\
        \frac{\partial G_{\tau}}{\partial \rho}(0,0)
    \end{pmatrix}
    = 0.
    \label{eqDisc:FirstOrderImplicitDerivAtBasepoint_Diagonal}    
\end{align}
A direct rearrangement of \eqref{eqDisc:FirstOrderImplicitDerivAtBasepoint_Diagonal} indicates that
\begin{equation}
\label{eqDisc:dGv_drho_is_0}
    \frac{\partial G_v}{\partial \rho}(0,0) = 0.
\end{equation}
Since $\frac{\partial G_{v}}{\partial \rho} (\epsilon,0)$ is continuously differentiable by 
\lemmaitemref{lemmaDisc:Theorem3ContinousDiff}{lemmaItem_dG_drho_continuous_diff_Theorem3}
and vanishes at $\epsilon = 0$ by \eqref{eqDisc:dGv_drho_is_0}, we have
\begin{equation}
\label{eqDisc:delta_omega_at_epsilon}
    \epsilon^3\frac{\partial G_v}{\partial \rho}(\epsilon,0)
    =
    \epsilon^4
    \left(
        \frac{\partial^2 G_v}{\partial \rho \partial \epsilon}(0,0) + o(1)
    \right).
\end{equation}
Rewrite
\eqref{eqDisc:FirstOrderImplicitPartial} as
\begin{equation}
\begin{aligned}
        \nabla_{v,\tau}
        F(G(\epsilon, \rho); \epsilon, \rho)
        \frac{\partial G}{\partial \rho}(\epsilon, \rho)
        =
        -
        \begin{pmatrix}
            \epsilon 
            \spacednablatranspose{\omega}{T_U
            \left(
            G_{\theta}(\epsilon, \rho)
            , u
            \right)}         
            + \frac{\partial g_{\omega}}{\partial \rho}(G_{\theta}(\epsilon, \rho),u,\rho)^T
            \\
            \spacednablatranspose{\tau}{T_U
            \left(
            G_{\theta}(\epsilon, \rho), 
            u
            \right)}
            + 
            \frac{\partial g_{\tau}}{\partial \rho} (G_{\theta}(\epsilon, \rho),u,\rho)^T
        \end{pmatrix}
\end{aligned}
\end{equation}
Setting $\rho = 0$ and using
\lemmaitemref{lemmaDisc:Theorem3ContinousDiff}{lemmaItem_g_omega_tau_differentiable},
this reduces to 
\begin{equation}
\label{eqDisc:at_rho_0_implict_partial_derivative}
\begin{aligned}
        \nabla_{v,\tau}
        F(G(\epsilon, 0); \epsilon, 0)
        \frac{\partial G}{\partial \rho}(\epsilon, 0)
        =
        -
        \begin{pmatrix}
            \epsilon 
            \spacednablatranspose{\omega}{T_U
            \left(
            G_{\theta}(\epsilon, 0)
            , u
            \right)}         
            \\
            \spacednablatranspose{\tau}{T_U
            \left(
            G_{\theta}(\epsilon, 0), 
            u
            \right)}
        \end{pmatrix}.
\end{aligned}
\end{equation}
To obtain $\frac{\partial^2 G_v}{\partial \rho \partial \epsilon}(0,0)$, 
differentiate \eqref{eqDisc:at_rho_0_implict_partial_derivative} and set $\epsilon=0$.
From the LHS we obtain
\begin{equation}
\label{eqDisc:LHS_at_rho_0_implict_partial_derivative}
\resizebox{0.9\hsize}{!}{$
    \begin{aligned}
        &\left. \frac{\partial}{\partial \epsilon} \right|_{\epsilon=0}
        \left(
            \nabla_{v,\tau}
            F(G(\epsilon, 0); \epsilon, 0)
            \frac{\partial G}{\partial \rho}(\epsilon, 0)
        \right)
        \\
        &\qquad
        =
        \left. \frac{\partial}{\partial \epsilon} \right|_{\epsilon=0}
        \left( 
        \Big. 
        \nabla_{v,\tau}
        F(G(\epsilon, 0); \epsilon, 0) 
        \right) 
        \begin{pmatrix}
            \frac{\partial G_v}{\partial \rho}(0,0)
            \\
            \frac{\partial G_{\tau}}{\partial \rho}(0,0)
        \end{pmatrix}
        +
        \nabla_{v,\tau}
        F(G(0, 0); 0, 0)
        \begin{pmatrix}
            \frac{\partial^2 G_v}{\partial \rho \partial \epsilon}(0,0)
            \\
            \frac{\partial^2 G_{\tau}}{\partial \rho \partial \epsilon}(0,0)
        \end{pmatrix}.
    \end{aligned}
$}
\end{equation}
From the RHS,
\begin{equation}
\label{eqDisc:RHS_at_rho_0_implict_partial_derivative}
    \begin{aligned}
        -
        \left. \frac{\partial}{\partial \epsilon} \right|_{\epsilon=0}
        \begin{pmatrix}
            \epsilon 
            \spacednablatranspose{\omega}{T_U
            \left(
            G_{\theta}(\epsilon, 0)
            , u
            \right)}         
            \\
            \spacednablatranspose{\tau}{T_U
            \left(
            G_{\theta}(\epsilon, 0), 
            u
            \right)}
        \end{pmatrix}
        =
        -
        \begin{pmatrix}
            \spacednablatranspose{\omega}{T_U
            \left(
            G_{\theta}(0, 0)
            , u_0
            \right)}         
            \\
            \left. \frac{\partial}{\partial \epsilon} \right|_{\epsilon=0}
            \left(
                \spacednablatranspose{\tau}{T_U
                \left(
                G_{\theta}(\epsilon, 0), 
                u
                \right)}
            \right)
        \end{pmatrix}.
    \end{aligned}
\end{equation}
From \eqref{eqDisc:dGv_drho_is_0}, $\frac{\partial G_v}{\partial \rho}(0,0) = 0$, and from \eqref{eqDisc:FirstOrderImplicitDerivAtBasepoint_Diagonal}, $\nabla_{v,\tau} F(G(0, 0); 0, 0)$ is a block diagonal matrix.
Then, by comparing the upper block entries of \eqref{eqDisc:LHS_at_rho_0_implict_partial_derivative} and \eqref{eqDisc:RHS_at_rho_0_implict_partial_derivative} and recalling that $G_v(0,0)=v_0$, $G_\tau(0,0)=\tau_0$, and $G_\theta(0,0)=\theta_0$, we see that
\begin{equation}
    \begin{aligned}
        \nabla_v F_{\omega}(v_0, \tau_0; 0, 0)
        \frac{\partial^2 G_v}{\partial \rho \partial \epsilon}(0,0) 
        =
        -
        \spacednablatranspose{\omega}{T_U
            \left(
            \theta_0
            , u_0
            \right)}.
    \end{aligned}
\end{equation}
By \lemmaitemref{lemmaDisc:Theorem3ContinousDiff}{lemmaItem_grad_tau_F_omega},
$\nabla_{v} F_{\omega}(v_0, \tau_0; 0, 0) = - B^{T}
    K''_U
    \left(0; 
        \theta_0
    \right)^{-1}
    B$, and therefore,
\begin{equation}
    \frac{\partial^2 G_v}{\partial \rho \partial \epsilon}(0,0)
    =
    \left(
    B^{T}
    K''_U
    \left(0; 
        \theta_0
    \right)^{-1}
    B
    \right)^{-1} 
    \spacednablatranspose{\omega}{T_U
    \left(
     \theta_0, u_0
    \right)}
    =
    k_{0,\omega}^*
    .
\end{equation}
Therefore, combining \eqref{eqDisc:delta_hat_theorem3b}-\eqref{eqDisc:delta_omega_at_epsilon} gives $\hat{\delta}_{n,\omega} = \frac{1}{n^2}(k_{0, \omega}^* + o(1))$, which verifies \eqref{eqDisc:ApprxOmegaDiscSpecialCase}.

Next, we assess the consequence of condition \eqref{eqDisc:ZeroMatrixConditionTheorem3} on the difference $\delta_{n,\omega} - \hat{\delta}_{n,\omega}$, following from \eqref{eqDisc:DiscDifferenceTheorem3}. For
$\rho^* \in [0,\rho]$,
the $\omega$-based difference is expressed as
\begin{align*}
    \epsilon 
    \left(
    G_{v}(\epsilon,\rho) - G_{v}(\epsilon, 0)
    -
    \rho  
    \frac{\partial G_{v}}{\partial \rho} \left(\epsilon, 0\right)
    \right)
    =
    \tfrac{1}{2}\epsilon \rho^2
    \frac{\partial^2 G_{v}}{\partial \rho^2} \left(\epsilon, \rho^*\right).
\end{align*}
\begin{lemma}
\label{lemmaDisc:SecondDeriv}
    For the continuously differentiable function $G_v(\epsilon, \rho)$, under condition \eqref{eqDisc:ZeroMatrixConditionTheorem3},
    $\frac{\partial^2 G_{v}}{\partial \rho^2} \left(\epsilon, \rho\right) = O(\epsilon + \rho)$.
\end{lemma}
We defer the proof to \supsecref{lemmaDisc:SecondDerivTheorem3}.

Applying Lemma~\ref{lemmaDisc:SecondDeriv} yields
$
    \epsilon \rho^2
    \frac{\partial^2 G_{v}}{\partial \rho^2} \left(\epsilon, \rho^*\right)
    =
    O(\rho^2(\epsilon + \rho^*))
$.
Setting $\rho = 1/n$ and $\epsilon = 1/n^{1/2}$, we have
\[
    \delta_{n,\omega} - \hat{\delta}_{n,\omega} 
    =
    O(1/n^3)~\text{as}~n \to \infty.
\]
This verifies
\eqref{eqDisc:OmegaDiscDiffSpecialCase}, and 
\eqref{eqDisc:TrueOmegaDiscSpecialCase}
follows from \eqref{eqDisc:ApprxOmegaDiscSpecialCase} and \eqref{eqDisc:OmegaDiscDiffSpecialCase}.
\end{proof}

\section{Additional proofs}

\subsection{Proof of Lemma \ref{lemmaDisc:Theorem3ContinousDiff}}
\label{appDisc:DerivativesFomegaFtau}

Using the abbreviation 
\[
    Q = K''_U\left(0; 
        \smallmat{
        \omega_0 \\ 
        \tau_0 } 
        \right)^{-1}~\text{and}~
    B = \nabla_{\omega}
            K'_U\left(0; 
                \smallmat{
                \omega_0 \\ 
                \tau_0 }
            \right),
\]
we begin by remarking that the condition for $(v_0,\tau_0)$ to be a non-degenerate local MLE for $Z_{v,\tau}=z_0$ can be expressed explicitly as the conditions that 
\begin{align}
    v_0 = 
    &\left(
        B^{T}
        Q
        B
    \right)^{-1}
    B^{T} Q z_0,
    ~\text{and}~
    \label{eqDisc:H3_appendixH_Goodman2022}
    \\
    (Qz_0 - Q B v_0)^T
    \tfrac{\partial K''_U}{\partial \tau_k} 
    \left(0; 
        \smallmat{
        \omega_0 \\ 
        \tau_0 } 
    \right)
    &(Qz_0 - Q B v_0)
    =
    \tfrac{\partial }{\partial \tau_k} 
    \log \det \left( K''_{U}
    \left(0; 
        \smallmat{
        \omega_0 \\ 
        \tau_0 }
    \right)
    \right),
    \label{eqDisc:H4_appendixH_Goodman2022}
\end{align}
see \citet[Appendix H]{Goodman2022}.

\begin{proof}[\proofof{Lemma~\ref{lemmaDisc:Theorem3ContinousDiff}}]

Part (\ref{lemmaItemNonSingularTheorem3ContinousDiff}) supposes that, the function $F$
can be extended to be a continuously differentiable with respect to its parameters in the neighborhood of $\epsilon = 0$. 
Both $F_{\omega}$ and $F_{\tau}$ as defined in \eqref{eqDisc:F_omega} and \eqref{eqDisc:F_tau}, contain terms that depend on the solution of the saddlepoint equation, 
\[
\hat{t}_U
\left(
    \smallmat{
    \omega_0 + \epsilon v \\ 
    \tau }
    ;
    u_0 + \epsilon z_0
\right).
\]
Therefore, we will first address continuous differentiability of this within the same neighborhood, an essential step, before examining $F_{\omega}$ and $F_{\tau}$. 
Throughout this proof, we write $\theta_0 = 
\smallmat{\omega_0 \\ \tau_0}$.

For $\epsilon \neq 0$, define the rescaled solution of the saddlepoint equation as
\begin{equation}
\label{eqDisc:sigma_hat_epsilon_neq_0}
     \hat{\sigma}(v,\tau,\epsilon) = 
    \epsilon^{-1} \hat{t}_U
    \left(
        \smallmat{
        \omega_0 + \epsilon v \\ 
        \tau }
        ;
        u_0 + \epsilon z_0
    \right)
    .
\end{equation} 
To satisfy the saddlepoint equation, we have
\begin{equation}\label{eqDisc:SaddlepointEqnRescaled}
    \epsilon^{-1}(K'_U(\epsilon \hat{\sigma}(v,\tau, \epsilon); \smallmat{
        \omega_0 + \epsilon v \\ 
        \tau }) -
        u_0
    - \epsilon z_0) = 0.
\end{equation}
For $\epsilon \neq 0$, and $\sigma \in \mathbb{R}^{1 \times d}$, define the function
\begin{equation}
\label{eqDisc:SaddEqnFuncS}
    S(\sigma,v,\tau,\epsilon) = \epsilon^{-1}(K'_U(\epsilon \sigma; \smallmat{
        \omega_0 + \epsilon v \\ 
        \tau }) -
        u_0
    - \epsilon z_0).
\end{equation}
We will show that this function can be extended to be a continuously differentiable function of all its arguments in the neighborhood of $\epsilon = 0$. 
We have 
$u_0 =  
K'_U
\left(0;
    \theta_0
\right) $. 
By assumption \eqref{eqDisc:PartiallyIdentifiableAssumption}, 
\[
    K_U' 
    \left(
        0; 
        \theta_0
    \right)
    =
    K_U' 
    \left(
        0; 
        \smallmat{
        \omega_0 \\ 
        \tau } 
    \right).
\] 
Therefore, we can express \eqref{eqDisc:SaddEqnFuncS} as
\begin{align}
    S(\sigma,v,\tau,\epsilon)
    &=
    \epsilon^{-1}(K'_U(\epsilon \sigma; \smallmat{
        \omega_0 + \epsilon v \\ 
        \tau }) 
    - K'_U(0; \smallmat{
        \omega_0 + \epsilon v \\ 
        \tau }) 
    + K'_U(0; \smallmat{
        \omega_0 + \epsilon v \\ 
        \tau }) - 
    K'_U
    \left(
    0;  
    \smallmat{
        \omega_0 \\ 
        \tau }
    \right) 
    - \epsilon z_0)
    \label{eqDisc:SaddEqnExpressionIndependtTau},
    \\
    &=
    \int_0^1
    \left( 
        K''_U(y\epsilon \sigma; \smallmat{
        \omega_0 + \epsilon v \\ 
        \tau }) \sigma^T
        +
        \nabla_{\omega} K_U' 
        \left(
            0; 
             \smallmat{
            \omega_0 + y\epsilon v \\ 
            \tau 
            } 
        \right)
        v
    \right)
    \, dy
    -
    z_0.
    \label{eqDisc:SaddlepointEqnIntegral}
\end{align}
This integral which defines $S(\sigma,v,\tau,\epsilon)$, remains valid even when $\epsilon = 0$. Additionally, applying 
\citeauthor{Goodman2022}'s \citeyearpar{Goodman2022} Lemma~14  to this representation we can show that $S(\sigma,v,\tau,\epsilon)$ is continuously differentiable in the neighborhood of $\epsilon = 0$. At $\epsilon = 0$,
\begin{align}
    &S(\sigma,v,\tau,0)
    =
    K''_U(0; \smallmat{
        \omega_0 \\ 
        \tau }) \sigma^T
    +
    \nabla_{\omega} K_U' 
    \left(
        0; 
         \smallmat{
        \omega_0\\ 
        \tau 
        } 
    \right)
    v
    - z_0,
    \label{eqDisc:S_at_epsilon_0}
    ~\text{and}~
    \\
    &\nabla_{\sigma} S(\sigma,v,\tau,0)
    = 
    K''_U(0; \smallmat{
        \omega_0 \\ 
        \tau })
    ~\text{is non-singular by assumption 
    (\ref*{eqDisc:NonSingularCovAndMean}, Section~\ref*{secDisc:Notations})
    }.
\end{align}
By the Implicit Function Theorem, there exist a continuously differentiable function $\hat{\sigma}(v,\tau,\epsilon)$ that satisfies
\begin{equation}
    S(\hat{\sigma}(v,\tau,\epsilon), v, \tau, \epsilon) = 0.
\end{equation}
Since $S(\hat{\sigma}(v,\tau,\epsilon), v, \tau, \epsilon)$ corresponds to the left-hand side of \eqref{eqDisc:SaddlepointEqnRescaled}, for $\epsilon \neq 0$, $\hat{\sigma}(v,\tau,\epsilon)$ is given by \eqref{eqDisc:sigma_hat_epsilon_neq_0}. For $\epsilon = 0$,  using \eqref{eqDisc:S_at_epsilon_0}, we have that $S(\hat{\sigma}(v,\tau,0), v, \tau, 0) = 0$ if  
\begin{equation}\label{eqDisc:ConditionForSaddlepointEqn}
    \hat{\sigma}^T(v,\tau,0) = 
    K''_U
    \left(0; 
          \smallmat{
          \omega_0 \\ 
          \tau }
    \right)^{-1} 
    \left(
        z_0 - 
        \nabla_{\omega} K_U' 
        \left(
            0; 
            \smallmat{
            \omega_0 \\ 
            \tau }
        \right)
        v
    \right)
    .
\end{equation}
The gradient of \eqref{eqDisc:ConditionForSaddlepointEqn} with respect to $v$ is
\begin{align}
    \spacednabla{v}{\hat{\sigma}}^T(v, \tau, 0)
    =
    - 
    K''_U
    \left(0; 
            \smallmat{
            \omega_0 \\ 
            \tau } 
    \right)^{-1} 
    \nabla_{\omega} K_U' 
    \left(
        0; 
        \smallmat{
        \omega_0 \\ 
        \tau }
    \right)
    .
    \label{eqDisc:grad_v_sigma}
\end{align}
We also see that since 
$K'_U(0; \smallmat{
        \omega_0 \\ 
        \tau })$ 
does not depend on $\tau$ (by assumption \eqref{eqDisc:PartiallyIdentifiableAssumption}), the gradient with respect to $\tau$ only involves the first term of \eqref{eqDisc:S_at_epsilon_0}. We have
\begin{align}
    \nabla_{\tau}
    \left(
        K''_U(0; \smallmat{
        \omega_0 \\ 
        \tau })
        \hat{\sigma}(v, \tau, 0)^T
    \right)
    = 0
    ,
\end{align}
and by implicit partial differentiation, the $k^{\mathrm{th}}$ column of $d \times p_2$ matrix $\spacednabla{\tau}{\hat{\sigma}}^T (v, \tau, 0)$ is
\begin{equation}
\label{eqDisc:grad_tau_sigma}
    -
     K''_U
    \left(0; 
        \smallmat{
        \omega_0 \\ 
        \tau }
    \right)^{-1}
    \frac{\partial K''_U}{\partial \tau_k}
    \left(
    0; 
    \smallmat{
    \omega_0 \\ 
    \tau }
    \right)
    \hat{\sigma}(v, \tau, 0)^T
    .
\end{equation}

Next, we consider the functions $F_{\omega}$ and $F_{\tau}$. The functions as defined in \eqref{eqDisc:F_omega} and \eqref{eqDisc:F_tau} are only valid for $\epsilon \neq 0$. We use a similar argument to extend the definition of 
$
    F
    = 
    \smallmat{
    F_{\omega} \\
    F_{\tau}}
$
to be continuously differentiable and valid for $\epsilon = 0$.

We first examine $F_{\omega}$.
Recall from \eqref{eqDisc:grad_at_0_of_K} that $\nabla_{\omega} K_U(0; \theta) = 0$ for all $\theta$, so
the first term of \eqref{eqDisc:F_omega} can be expressed using \eqref{eqDisc:Grad_of_the_log_L_U_star} as
\begin{equation}
\label{eqDisc:FirstTermF_omega}
\begin{aligned}
    \epsilon^{-1} 
    \nabla_{\omega}^T 
    \log \hat{L}_{U}^* 
    &(\smallmat{
        \omega_0 + \epsilon v \\ 
        \tau }, 
        u_0 + \epsilon z_0)
    \\
    &\qquad
    =
    \epsilon^{-1} 
    \left(
        \nabla_{\omega}^T K_U(\epsilon\hat{\sigma}(v,\tau, \epsilon); 
        \smallmat{
        \omega_0 + \epsilon v \\ 
        \tau }) 
        - 
        \nabla_{\omega}^T K_U(0; 
        \smallmat{
        \omega_0 + \epsilon v \\ 
        \tau }) 
    \right)
    \\
    &\qquad
    =
    \int_0^1 
    \left(
    \nabla_{\omega} K'_U(y\epsilon\hat{\sigma}; 
    \smallmat{
        \omega_0 + \epsilon v \\ 
        \tau })
    \right)^T
    \hat{\sigma}(v,\tau,\epsilon)^T \,dy
    .
\end{aligned}
\end{equation}
With this integral, we see that $F_{\omega}$ can be extended to be continuously differentiable in the neighborhood of $\epsilon = 0$. See \citet[Lemma~14]{Goodman2022}.
It follows that
\begin{equation}
\label{eqDisc:F_omega_at_epsilon_rho_equal_0}
    F_{\omega}(v, \tau; 0, \rho) 
=
\left(
\nabla_{\omega} 
K'_U\left(0; 
        \smallmat{
        \omega_0 \\ 
        \tau }
\right)
\right)^T
\hat{\sigma}(v, \tau, 0)^T
,
\end{equation}
and using \eqref{eqDisc:ConditionForSaddlepointEqn}, 
it follows that $F_{\omega}(v, \tau, 0, \rho) = 0$ if and only if
\begin{equation}
\label{eqDisc:v_at_epsilon_zero}
\begin{aligned}
v = 
&\left(
    \left(
        \nabla_{\omega}
        K'_U\left(0; 
            \smallmat{
            \omega_0 \\ 
            \tau }
        \right)
    \right)^{T}
    K''_U\left(0; 
    \smallmat{
    \omega_0 \\ 
    \tau } 
    \right)^{-1}
    \nabla_{\omega}
    K'_U\left(0; 
    \smallmat{
    \omega_0 \\ 
    \tau }
    \right)
\right)^{-1}
\\
&\qquad\qquad\qquad
\cdot 
\left(
    \nabla_{\omega}
    K'_U\left(0; 
        \smallmat{
        \omega_0 \\ 
        \tau } 
    \right)
\right)^{T}
    K''_U\left(0; 
        \smallmat{
        \omega_0 \\ 
        \tau } \right)^{-1}
    z_0
    .
\end{aligned}
\end{equation}

For $F_{\tau}$, we have that
using \eqref{eqDisc:Grad_of_the_log_L_U_star}, the first term of \eqref{eqDisc:F_tau} is expressible as
\begin{equation}
\label{eqDisc:FirstTermF_tau_sigma_hat}
    \epsilon^{-2} \nabla_{\tau}^T \log \hat{L}_{U}^* (\smallmat{
    \omega_0 + \epsilon v \\ 
    \tau }, u_0+\epsilon z_0) 
= \epsilon^{-2} \nabla_{\tau}^T K_U(\epsilon\hat{\sigma}(v,\tau,\epsilon); 
\smallmat{
        \omega_0 + \epsilon v \\ 
        \tau }
).
\end{equation}
From \eqref{eqDisc:PartiallyIdentifiableAssumption}, we see that $\nabla_{\tau} K'_U(0;\theta) = 0$, therefore, \eqref{eqDisc:FirstTermF_tau_sigma_hat} 
can be represented by the first-order Taylor remainder as follows.
\begin{align}
    &\epsilon^{-2} \nabla_{\tau}^T K_U(\epsilon\hat{\sigma}(v,\tau,\epsilon); 
        \smallmat{
        \omega_0 + \epsilon v \\ 
        \tau }
    )
    \notag
    \\
    &\quad
    =
    \epsilon^{-2}
    \left(
        \nabla_{\tau}^T K_U(\epsilon\hat{\sigma}(v,\tau, \epsilon); 
        \smallmat{
        \omega_0 + \epsilon v \\ 
        \tau }
        ) - 
        \nabla_{\tau}^T K_U(0; 
        \smallmat{
        \omega_0 + \epsilon v \\ 
        \tau }
        ) - 
        \epsilon
        \left(
        \nabla_{\tau} K'_U(0; \smallmat{
        \omega_0 + \epsilon v \\ 
        \tau }
        )
        \right)^T
        \hat{\sigma}(v,\tau, \epsilon)^T
    \right)
    \notag
    \\
    &\quad
    = \int_0^1 (1-y) \sum_{i,j = 1}^d \hat{\sigma}(v,\tau, \epsilon)_i~ \hat{\sigma}(v,\tau, \epsilon)_j \nabla_{\tau}^T \frac{\partial^2 K_U}{\partial t_i \partial t_j} (y\epsilon\hat{\sigma}(v,\tau,\epsilon); 
    \smallmat{
    \omega_0 + \epsilon v \\ 
    \tau }
    )  \,dy.
    \label{eqDisc:FirstTerm_grad_tau_F_tau}
\end{align}
This integral representation is also continuously differentiable in the neighbourhood of $\epsilon = 0$, and thus, $F_{\tau}$ can be extended to be valid in this neighbourhood.
Evaluating \eqref{eqDisc:FirstTerm_grad_tau_F_tau} at $\epsilon = 0$, yields a $p_2$-dimensional column vector whose $k^{\mathrm{th}}$ entry is
\begin{equation}
\label{eqDisc:kth_entry_first_term_grad_tau_F_tau}
    \tfrac{1}{2}\hat{\sigma}(v,\tau,0) \frac{\partial K''_U}{\partial \tau_k} 
    \left(0; 
        \smallmat{
        \omega_0 \\ 
        \tau }
    \right)
    \hat{\sigma}(v,\tau, \epsilon)^T
    .
\end{equation}
On the second term of \eqref{eqDisc:F_tau}, the $k^{\mathrm{th}}$ entry of  $\nabla_{\tau}^T \log \hat{P}_U(\theta,u)$ at $\epsilon = 0$ is
\begin{equation}
\label{eqDisc:kth_entry_second_term_grad_tau_F_tau}
    - \tfrac{1}{2} \frac{\partial }{\partial \tau_k} \log \det \left( K''_{U}
    \left(0; 
        \smallmat{
        \omega_0 \\ 
        \tau }
    \right)
    \right)
    .
\end{equation}
Combining \eqref{eqDisc:kth_entry_first_term_grad_tau_F_tau} and \eqref{eqDisc:kth_entry_second_term_grad_tau_F_tau}, the $k^{\mathrm{th}}$ entry of $F_{\tau}(v,\tau, 0, 0)$ is given by
\begin{equation}
\label{eqDisc:F_tau_kth_entry}
\begin{aligned}
    F_{\tau}(v,\tau; 0, \rho)_k
    &=
    \tfrac{1}{2} 
    \left(
    \hat{\sigma}(v,\tau,0)
    \tfrac{\partial K''_U}{\partial \tau_k} 
    \left(0; 
        \smallmat{
        \omega_0 \\ 
        \tau }
    \right)
    \hat{\sigma}(v,\tau,0)^T
    -
    \tfrac{\partial }{\partial \tau_k} \log \det \left( K''_{U}
    \left(0; 
        \smallmat{
        \omega_0 \\ 
        \tau }
    \right)
    \right)
    \right)
    \\
    &\qquad
    +
    \rho \nabla_{\tau}^T 
    T_U(\smallmat{
        \omega_0 \\ 
        \tau }, u_0) 
    + 
    g_{\tau}(\smallmat{
        \omega_0 \\ 
        \tau }, u_0, \rho)^T.
\end{aligned}
\end{equation}
If in addition we set $\rho = 0$,
it follows that $F_{\tau}(v,\tau, 0, 0)_k = 0$
under the condition
\begin{equation}
\label{eqDisc:condition_for_Fv_equal_zero}
    \hat{\sigma}(v,\tau,0)
    \tfrac{\partial K''_U}{\partial \tau_k} 
    \left(0; 
        \smallmat{
        \omega_0 \\ 
        \tau } 
    \right)
    \hat{\sigma}(v,\tau,0)^T
    =
    \tfrac{\partial }{\partial \tau_k} 
    \log \det \left( K''_{U}
    \left(0; 
        \smallmat{
        \omega_0 \\ 
        \tau }
    \right)
    \right)
    .
\end{equation}

Thus far, we have demonstrated that the gradient functions $F_{\omega}$ and $F_{\tau}$ can be expressed in forms that are continuously differentiable in the neighbourhood of $\epsilon = 0$, and they reduce to 0 under certain conditions. 
Specifically, the conditions 
\eqref{eqDisc:ConditionForSaddlepointEqn}, \eqref{eqDisc:v_at_epsilon_zero} and
\eqref{eqDisc:condition_for_Fv_equal_zero} 
can be shown to reduce to conditions \eqref{eqDisc:H3_appendixH_Goodman2022}--\eqref{eqDisc:H4_appendixH_Goodman2022}, which are necessary conditions for $(v_0,\tau_0)$ to be a non-degenerate local MLE for $Z_{v,\tau}=z_0$. Hence, under the hypotheses of Theorem~\ref{theoremDisc:Theorem3}, we have
\begin{equation}
    F
    \left(
    \smallmat{v_0 \\ \tau_0 }
    ; 0, 0
    \right) = 0
    .
\end{equation}

Next we examine the non-singularity of the gradient of $F$ at $(v_0, \tau_0)$, which has the block form
\begin{equation}
\label{eqDisc:HessianF_unabbreviated}
    \nabla_{v, \tau} F
    = 
    \smallmat{
    \nabla_v F_{\omega} & \nabla_{\tau} F_{\omega} \\
    \nabla_v F_{\tau} & \nabla_{\tau} F_{\tau}
    }
    .
\end{equation}
Using \eqref{eqDisc:F_omega_at_epsilon_rho_equal_0} and \eqref{eqDisc:grad_v_sigma},
\begin{equation}
\label{eqDisc:del_vF_omega}
    \nabla_v F_{\omega}(v, \tau; 0, \rho)
    = 
    -
    \left(
    \nabla_{\omega}
    K'_U
    \left(0; 
        \smallmat{
        \omega_0 \\ 
        \tau }
    \right)
    \right)^{T}
    K''_U
    \left(0; 
        \smallmat{
        \omega_0 \\ 
        \tau }
    \right)^{-1}
    \nabla_{\omega}
    K'_U
    \left(0; 
        \smallmat{
        \omega_0 \\ 
        \tau } 
    \right)
    .
\end{equation}
Next, using \eqref{eqDisc:F_omega_at_epsilon_rho_equal_0} and \eqref{eqDisc:grad_tau_sigma}, let 
$\tilde{W}$
be $d \times p_2$ matrix whose $j^{\mathrm{th}}$ column is
\begin{equation}
\label{eqDisc:MatrixTildeW}
    \tfrac{\partial K''_U}{\partial \tau_j}
    \left(
    0; 
    \smallmat{
    \omega_0 \\ 
    \tau }
    \right)
    \hat{\sigma}(v, \tau, 0)^T,
\end{equation}
then, 
\begin{equation}
\label{eqDisc:del_tau_F_omega}
    \nabla_{\tau} F_{\omega}(v, \tau; 0, \rho) = 
    -
    \left(
    \nabla_{\omega}
    K'_U
    \left(0; 
        \smallmat{
        \omega_0 \\ 
        \tau }
    \right)
    \right)^{T}
    K''_U
    \left(0; 
        \smallmat{
        \omega_0 \\ 
        \tau }
    \right)^{-1}
    \tilde{W}
    .
\end{equation}
For the gradients of $F_{\tau}$ we consider only the case $\rho = 0$, 
where the last terms in \eqref{eqDisc:F_tau_kth_entry} vanish. Of the remaining terms, only the first term depends on $v$, so that
\begin{equation}
\label{eqDisc:del_v_F_tau}
    \nabla_v F_{\tau}(v, \tau; 0, 0) = 
    - 
    \tilde{W}^T
    K''_U
    \left(0; 
        \smallmat{
        \omega_0 \\ 
        \tau }
    \right)^{-1}
    \nabla_{\omega}
    K'_U
    \left(0; 
        \smallmat{
        \omega_0 \\ 
        \tau }
    \right)
    .
\end{equation}
Finally, using \eqref{eqDisc:F_tau_kth_entry} and \eqref{eqDisc:grad_tau_sigma}, we see that
the $k, l$ entry of the $p_2 \times p_2$ matrix $\nabla_{\tau} F_{\tau}(v, \tau; 0, 0)$ becomes 
\begin{equation}
\begin{aligned}
\hat{\sigma}(v,\tau,0)
&\left(
    \tfrac{1}{2}
    \tfrac{\partial^2 K''_U}{\partial \tau_k \partial \tau_l} 
    \left(
     0; 
        \smallmat{
        \omega_0 \\ 
        \tau }
    \right)
    -
    \tfrac{\partial K''_U}{\partial \tau_k} 
    \left(
         0; 
        \smallmat{
        \omega_0 \\ 
        \tau }
    \right)
    K''_U
    \left(0; 
        \smallmat{
        \omega_0 \\ 
        \tau } 
    \right)^{-1}
    \tfrac{\partial K''_U}{\partial \tau_l} 
    \left(
         0; 
        \smallmat{
        \omega_0 \\ 
        \tau }
    \right)
\right)
\hat{\sigma}(v,\tau,0)^T
\\
&\qquad\qquad\qquad
- \tfrac{1}{2}
\tfrac{\partial^2 }{\partial \tau_k \partial \tau_l} \log \det \left( K''_{U}
    \left(0; 
        \smallmat{
        \omega_0 \\ 
        \tau }
    \right)
    \right)
    .
\end{aligned}
\end{equation}
Abbreviate the evaluation of \eqref{eqDisc:HessianF_unabbreviated} as
\begin{equation}
\label{eqDisc:HessianF_abbreviated}
    \nabla_{v, \tau} F
    \left(
    \smallmat{
    v_0 \\ 
    \tau_0
    }
    ; 0, 0
    \right)
    = 
    \smallmat{
    A & -B^T Q W \\
    -W^T Q B & D - W^T Q W
    }
    ,
\end{equation}
where $A$ correspond to the evaluation of \eqref{eqDisc:del_vF_omega} at $\tau_0$, $W$ denote 
the evaluation of the matrix defined by \eqref{eqDisc:MatrixTildeW} evaluated at $\tau_0$, 
$
B =
    \nabla_{\omega}
    K'_U
    \left(0; 
        \theta_0 
    \right)
$,
$
Q = 
K''_U
\left(0; 
    \theta_0
\right)^{-1}
$
and the $k,l$ entry of $p_2 \times p_2$ matrix $D$ is defined by
\[
D_{k,l} =
\tfrac{1}{2}
\left(
\tfrac{\partial^2 K''_U}{\partial \tau_k \partial \tau_l} 
\left(
 0; 
\theta_0
\right)
- 
\tfrac{\partial^2 }{\partial \tau_k \partial \tau_l} \log \det \left( K''_{U}
    \left(0; 
        \theta_0 
    \right)
    \right)
    .
\right)
\]
To establish the non-singularity of $\nabla_{v, \tau} F$, consider the following block-wise matrix transformations of \eqref{eqDisc:HessianF_abbreviated}. 
\begin{align}
    \smallmat{
        I & 0 \\
        W^T Q B A^{-1} & I
    }
    &
    \smallmat{
        A & -B^T Q W \\
        -W^T Q B & D - W^T Q W
    } 
    \smallmat{
        I & A^{-1} B^T Q W \\
        0 & I
    }
    =
    \smallmat{
        A & 0 \\
        0 & D - W^T (Q - QBA^{-1}B^T Q) W
    }
    .
\end{align}
We note that $A = -B^T Q B$ is negative definite 
by the same argument as in \eqref{eqDisc:WellSpecifiedNegDefiniteCondition}.
On the other hand, the matrix $D - W^T (Q - QBA^{-1}B^T Q) W$ can be shown to be negative definite when $(v_0,\tau_0)$ is a non-degenerate local MLE for $Z_{v,\tau}=z_0$, see \citet[Appendix H]{Goodman2022}.

Overall, the condition that $v_0$ and $\tau_0$ are non-degenerate local MLEs can be explicitly expressed as the requirement that  \eqref{eqDisc:ConditionForSaddlepointEqn}, \eqref{eqDisc:v_at_epsilon_zero}, and \eqref{eqDisc:condition_for_Fv_equal_zero} hold, and additionally, the matrix $D - W^T (Q - QBA^{-1}B^T Q) W$ is negative definite.
Under these conditions, 
\begin{equation}
\label{eqDisc:Equal_to_zero_And_Non_singularity}
    F \left( 
    \smallmat{v_0 \\ \tau_0} ; 0, 0 
    \right) = 0
    ~\text{and}~
    \nabla_{v, \tau} F \left( \smallmat{v_0 \\ \tau_0} ; 0, 0 \right)
    ~\text{is non-singular}
    .
\end{equation}
This completes the proof for part~(\ref{lemmaItemNonSingularTheorem3ContinousDiff}).

We next consider part~(\ref{lemmaItem_g_omega_tau_differentiable}). 
Note that 
$g_{\omega}(\theta,u,\rho)$ and $g_{\tau}(\theta,u,\rho)$ from \eqref{eqDisc:F_omega}--\eqref{eqDisc:F_tau}
have the same structure as $h(\theta, u, \rho)$ in \eqref{eqDisc:Proof1_1}. The proofs of the corresponding parts of
\lemmaitemref{lemmaDisc:Theorem1ContinousDiff}{lemmaDisc:continousDiff_of_h_dh} and 
\lemmaitemref{lemmaDisc:Theorem1ContinousDiff}{lemma1Disc:h_grad_h_continouslydiff}
use only the properties of $q$ and can be applied without change.

For part (\ref{lemmaItem_dG_drho_continuous_diff_Theorem3}), by 
\eqref{eqDisc:Equal_to_zero_And_Non_singularity} and the continuous differentiability of $F$ in all its variables, the Implicit Function Theorem guarantees the existence of a continuously differentiable function $G(\epsilon, \rho)$ in the neighborhood of $(\epsilon,\rho) = (0, 0)$, satisfying
\[
    F(G(\epsilon, \rho); \epsilon, \rho) = 0.
\] 
The continuous differentiability of 
$\frac{\partial g_{\omega}}{\partial \rho}$,
$\frac{\partial g_{\tau}}{\partial \rho}$
implies that 
$\frac{\partial F}{\partial \rho}$
is continuously differentiable, and it follows that $\frac{\partial G}{\partial \rho}$ is also continuously differentiable.
All of the above argument can be repeated when continuous differentiability is replaced by $C^3$, as in 
\lemmaitemref{lemmaDisc:Theorem1ContinousDiff}{lemma1Disc:NonSinglarity}.
This completes the proof of part (\ref{lemmaItem_dG_drho_continuous_diff_Theorem3}).

Part (\ref{lemmaItem_grad_tau_F_omega}) follows by evaluating  
\eqref{eqDisc:del_tau_F_omega} and \eqref{eqDisc:del_vF_omega} at $v_0$ and $\tau_0$.
\end{proof}

\subsection{Proof of Lemma \ref{lemmaDisc:Theorem3approximated_delta}}
\label{AppDisc:LemmaTheorem3ApprxFormula}

We will show that
$
    \smallmat{
        \hat{\delta}_{n,\omega} \\
        \hat{\delta}_{n,\tau}
    }
$
in \eqref{eqDisc:ApprtDiscTheorem3} is the discrepancy approximation defined in \eqref{eqDisc:DiscrepancyFormula}.
By implicit partial differentiation to \eqref{eqDisc:F_combined_omega_tau_at_0}.
\begin{align*}
    \tfrac{\partial F}{\partial \rho}
    \left( G(\epsilon, \rho); \epsilon, \rho \right)
    +
    \nabla_{v,\tau}
    F(G(\epsilon, \rho); \epsilon, \rho)
    \tfrac{\partial G}{\partial \rho}(\epsilon, \rho)
    = 0.
\end{align*}
Rearranging this and calculating 
$
\tfrac{\partial F}{\partial \rho}
    \left( \theta; \epsilon, \rho \right)
$, 
we get 
\begin{align}
\hspace{-0.5em}
    \begin{pmatrix}
        \frac{\partial G_{v}}{\partial \rho} (\epsilon,\rho) 
        \\
        \frac{\partial G_{\tau}}{\partial \rho} (\epsilon,\rho)
    \end{pmatrix}
    =
    -(\nabla_{v, \tau} F(G(\epsilon, \rho); \epsilon, \rho))^{-1}
    \begin{pmatrix}
        \epsilon 
        \spacednablatranspose{\omega}{T_U
        \left(
        G_{\theta}(\epsilon, \rho)
        , u
        \right)} 
        + \epsilon \frac{\partial g_{\omega}}{\partial \rho}(G_{\theta}(\epsilon, \rho),u,\rho)^T
        \\
        \spacednablatranspose{\tau}{T_U
        \left(
        G_{\theta}(\epsilon, \rho), 
        u
        \right)}
        + 
        \frac{\partial g_{\tau}}{\partial \rho} (G_{\theta}(\epsilon, \rho),u,\rho)^T
    \end{pmatrix}.
    \notag
\end{align}
Evaluate this at $(\epsilon, \rho) = ( 1/n^{1/2}, 0)$ and 
apply \lemmaitemref{lemmaDisc:Theorem3ContinousDiff}{lemmaItem_g_omega_tau_differentiable}.
Then 
we can express the right-hand side of \eqref{eqDisc:ApprtDiscTheorem3} as
\begin{align}
    \begin{pmatrix}
        \left( \frac{1}{n^{3/2}} \right)   
        \frac{\partial G_{v}}{\partial \rho} ( 1/n^{1/2}, 0) 
        \\
        \left( \frac{1}{n} \right)  
        \frac{\partial G_{\tau}}{\partial \rho} ( 1/n^{1/2}, 0)
    \end{pmatrix}
    =
    -
    &\begin{pmatrix}
       \frac{1}{n^{3/2}}I_{p_1 \times p_1} & 0_{p_1 \times p_2} \\
       0_{p_2 \times p_1} & \frac{1}{n}I_{p_2 \times p_2} 
    \end{pmatrix}
    (\nabla_{v, \tau} F(G(1/n^{1/2}, 0); 1/n^{1/2}, 0))^{-1}
    \notag
    \\
    &\quad\quad\quad\quad\quad\quad\quad\quad
    \cdot
    \begin{pmatrix}
        \left( \frac{1}{n^{1/2}}\right)
        \spacednablatranspose{\omega}{ 
        T_U
        \left(
        G_{\theta}(1/n^{1/2}, 0), u
        \right)} 
        \\
        \spacednablatranspose{\tau}{T_U
        \left(
        G_{\theta}(1/n^{1/2}, 0), u
        \right)}
    \end{pmatrix}
    \label{eqDisc:ProofLemmaTheorem3ApproximationFormulaProof}
    .
\end{align}
We show that this expression simplifies to \eqref{eqDisc:DiscrepancyFormula}.

\begin{proof}[\proofof{Lemma~\ref{lemmaDisc:Theorem3approximated_delta}}]

Consider the following expression for the inverse factor of \eqref{eqDisc:ProofLemmaTheorem3ApproximationFormulaProof}. 
\begin{equation}
\label{eqDisc:Theorem3ProofFinverse}
    (\nabla_{v, \tau} F(G(1/n^{1/2}, 0);  1/n^{1/2}, 0))^{-1}
     =
     \begin{pmatrix}
         Y_{v,\omega,n} &~ 
         Y_{\tau,\omega,n} 
         \\
         Y_{v,\tau,n}  &~ 
         Y_{\tau,\tau,n}
     \end{pmatrix}^{-1}
     ,
\end{equation}
where the entries of the block matrix are 
\begin{equation}
\label{eqDisc:Y_n_omega_tau_s}
    \begin{aligned}
        Y_{v,\omega,n} 
        &= 
        \nabla_{v} F_{\omega} \left(
        G_{v}(1/n^{1/2},0), 
        G_{\tau}(1/n^{1/2},0);
        1/n^{1/2},0
        \right),
        \\
        Y_{\tau,\omega,n} 
        &= 
        \nabla_{\tau} F_{\omega} \left(
        G_{v}(1/n^{1/2},0), 
        G_{\tau}(1/n^{1/2},0);
        1/n^{1/2},0
        \right),
        \\
        Y_{v,\tau,n} 
        &= 
        \nabla_{v} F_{\tau} \left(
        G_{v}(1/n^{1/2},0), 
        G_{\tau}(1/n^{1/2},0);
        1/n^{1/2},0
        \right),
        \\
        Y_{\tau,\tau,n} 
        &= 
        \nabla_{\tau} F_{\tau} \left(
        G_{v}(1/n^{1/2},0), 
        G_{\tau}(1/n^{1/2},0);
        1/n^{1/2},0
        \right).
    \end{aligned}
\end{equation}
Motivated by \eqref{eqDisc:Y_n_omega_tau_s}, and using \eqref{eqDisc:F_omega}--\eqref{eqDisc:F_tau}, we calculate the following.

\begin{equation}
\label{eqDisc:F_omega_del_v}
\resizebox{1\hsize}{!}{$
\begin{aligned}
    \nabla_{v} F_{\omega}
    (v,\tau; \epsilon, 0)
    &=
    \nabla_{v}
    \left(
        \epsilon^{-1} 
        \nabla_{\omega}^T 
        \log \hat{L}_{U}^* 
        \left(\smallmat{\omega_0 + \epsilon v \\ \tau}, 
        u_0 + \epsilon z_0 \right) + 
        \epsilon 
        \nabla_{\omega}^T \log \hat{P}_U
        \left(\smallmat{\omega_0 + \epsilon v \\ \tau},
        u_0 + \epsilon z_0 \right)
    \right)
    \\
    &=
    \nabla_{\omega}^T \nabla_{\omega}
    \log \hat{L}_{U}^* 
    \left(\smallmat{\omega_0 + \epsilon v \\ \tau}, 
    u_0 + \epsilon z_0\right)
    +
    \epsilon^2
    \nabla_{\omega}^T \nabla_{\omega}
    \log \hat{P}_U
    \left(\smallmat{\omega_0 + \epsilon v \\ \tau},
    u_0 + \epsilon z_0\right)
\end{aligned}
$}
\end{equation}

\begin{equation}
\label{eqDisc:F_omega_del_tau}
\resizebox{1\hsize}{!}{$
\begin{aligned}
    \nabla_{\tau} F_{\omega}
    (v,\tau; \epsilon, 0)
    &=
    \nabla_{\tau}
    \left(
        \epsilon^{-1} 
        \nabla_{\omega}^T 
        \log \hat{L}_{U}^* 
        \left(\smallmat{\omega_0 + \epsilon v \\ \tau}, 
        u_0 + \epsilon z_0\right) + 
        \epsilon 
        \nabla_{\omega}^T \log \hat{P}_U
        \left(\smallmat{\omega_0 + \epsilon v \\ \tau},
        u_0 + \epsilon z_0\right)
    \right)
    \\
    &=
    \epsilon^{-1}
    \nabla_{\tau}^T \nabla_{\omega}
    \log \hat{L}_{U}^* 
    \left(\smallmat{\omega_0 + \epsilon v \\ \tau}, 
    u_0 + \epsilon z_0\right)
    +
    \epsilon
    \nabla_{\tau}^T \nabla_{\omega}
    \log \hat{P}_U
    \left(\smallmat{\omega_0 + \epsilon v \\ \tau},
    u_0 + \epsilon z_0\right)
\end{aligned}
$}
\end{equation}

\begin{equation}
\label{eqDisc:F_tau_del_v}
\resizebox{1\hsize}{!}{$
\begin{aligned}
    \nabla_{v} F_{\tau}
    (v,\tau; \epsilon, 0)
    &=
    \nabla_{v}
    \left(
        \epsilon^{-2} 
        \nabla_{\tau}^T 
        \log \hat{L}_{U}^* 
        \left(\smallmat{\omega_0 + \epsilon v \\ \tau}, 
        u_0 + \epsilon z_0\right) + 
        \nabla_{\tau}^T \log \hat{P}_U
        \left(\smallmat{\omega_0 + \epsilon v \\ \tau},
        u_0 + \epsilon z_0\right)
    \right)
    \\
    &=
    \epsilon^{-1}
    \nabla_{\omega}^T \nabla_{\tau}
    \log \hat{L}_{U}^* 
    \left(\smallmat{\omega_0 + \epsilon v \\ \tau}, 
    u_0 + \epsilon z_0\right)
    +
    \epsilon
    \nabla_{\omega}^T \nabla_{\tau}
    \log \hat{P}_U
    \left(\smallmat{\omega_0 + \epsilon v \\ \tau},
    u_0 + \epsilon z_0\right)
\end{aligned}
$}
\end{equation}

\begin{equation}
\label{eqDisc:F_tau_del_tau}
\resizebox{1\hsize}{!}{$
\begin{aligned}
    \nabla_{\tau} F_{\tau}
    (v,\tau; \epsilon, 0)
    &=
    \nabla_{\tau}
    \left(
        \epsilon^{-2} 
        \nabla_{\tau}^T 
        \log \hat{L}_{U}^* 
        \left(\smallmat{\omega_0 + \epsilon v \\ \tau}, 
        u_0 + \epsilon z_0\right) + 
        \nabla_{\tau}^T \log \hat{P}_U
        \left(\smallmat{\omega_0 + \epsilon v \\ \tau},
        u_0 + \epsilon z_0\right)
    \right)
    \\
    &=
    \epsilon^{-2}
    \nabla_{\omega}^T \nabla_{\tau}
    \log \hat{L}_{U}^* 
    \left(\smallmat{\omega_0 + \epsilon v \\ \tau}, 
    u_0 + \epsilon z_0\right)
    +
    \nabla_{\omega}^T \nabla_{\tau}
    \log \hat{P}_U
    \left(\smallmat{\omega_0 + \epsilon v \\ \tau},
    u_0 + \epsilon z_0\right)
\end{aligned}
$}
\end{equation}
Using \eqref{eqDisc:F_omega_del_v}--\eqref{eqDisc:F_tau_del_tau}, and the gradient expression of \eqref{eqDisc:AsymptoticsFirstOrderLogLikelihood}, we see that
\begin{align}
    &Y_{v,\omega,n}
    =
    \nabla_{\omega}^T \nabla_{\omega} \log \hat{L}_{U}^* (\hat{\theta}_{\mathrm{spa},n}, u_n) + 
    \tfrac{1}{n} \nabla_{\omega}^T \nabla_{\omega} \log \hat{P}_U(\hat{\theta}_{\mathrm{spa},n},u_n)
    =
    \tfrac{1}{n} \nabla_{\omega}^T \nabla_{\omega} \log \hat{L}_n(\hat{\theta}_{\mathrm{spa},n}; x_n)
    ,
    \label{eqDisc:del_v_F_omega}
    \\
    &Y_{\tau,\omega,n} 
    =
    n^{1/2} \nabla_{\tau}^T \nabla_{\omega} \log \hat{L}_{U}^* (\hat{\theta}_{\mathrm{spa},n}, u_n) + 
    \tfrac{1}{n^{1/2}} \nabla_{\tau}^T \nabla_{\omega} \log \hat{P}_U(\hat{\theta}_{\mathrm{spa},n},u_n)
    =
    \tfrac{1}{n^{1/2}} \nabla_{\tau}^T \nabla_{\omega} \log \hat{L}_n(\hat{\theta}_{\mathrm{spa},n}; x_n)
    ,
    \\
    &Y_{v,\tau,n} 
    =
    n^{1/2} \nabla_{\omega}^T \nabla_{\tau} \log \hat{L}_{U}^* (\hat{\theta}_{\mathrm{spa},n}, u_n) + 
    \tfrac{1}{n^{1/2}} \nabla_{\omega}^T \nabla_{\tau} \log \hat{P}_U(\hat{\theta}_{\mathrm{spa},n},u)
    =
    \tfrac{1}{n^{1/2}} \nabla_{\omega}^T \nabla_{\tau} \log \hat{L}_n(\hat{\theta}_{\mathrm{spa},n}; x_n)
    ,
    \\
    &Y_{\tau,\tau,n} 
    =
    n \nabla_{\tau}^T \nabla_{\tau} \log \hat{L}_{U}^* (\hat{\theta}_{\mathrm{spa},n}, u_n) + 
    \nabla_{\tau}^T \nabla_{\tau} \log \hat{P}_U(\hat{\theta}_{\mathrm{spa},n},u_n)
    =
    \nabla_{\tau}^T \nabla_{\tau} \log \hat{L}_n(\hat{\theta}_{\mathrm{spa},n}; x_n).
\end{align}
Therefore,
\begin{equation}
\resizebox{1.0\hsize}{!}{$
 \label{eqDisc:F_as_log_x}
    \begin{aligned}
         (\nabla_{v, \tau} F(G(1/n^{1/2}, 0); 1/n^{1/2}, 0))^{-1}
     &=
     \begin{pmatrix}
         \tfrac{1}{n} \nabla_{\omega}^T \nabla_{\omega} \log \hat{L}_n(\hat{\theta}_{\mathrm{spa},n}; x_n) & 
         \tfrac{1}{n^{1/2}} \nabla_{\tau}^T \nabla_{\omega} \log \hat{L}_n(\hat{\theta}_{\mathrm{spa},n}; x_n) 
         \\
         \tfrac{1}{n^{1/2}} \nabla_{\omega}^T \nabla_{\tau} \log \hat{L}_n(\hat{\theta}_{\mathrm{spa},n}; x_n) & 
         \nabla_{\tau}^T \nabla_{\tau} \log \hat{L}_n(\hat{\theta}_{\mathrm{spa},n}; x_n)
     \end{pmatrix}^{-1}
     \\
     &
     =
     \smallmat{
         n^{1/2}I_{p_1 \times p_1} & 0 \\
         0 & I_{p_2 \times p_2}
    }
     \left(
        \nabla_{\theta}^T \nabla_{\theta} \log \hat{L}_n(\hat{\theta}_{\mathrm{spa},n}, x_n)
     \right)^{-1}
     \smallmat{
         n^{1/2}I_{p_1 \times p_1} & 0 \\
         0 & I_{p_2 \times p_2}
     }.
    \end{aligned}
$}
\end{equation}
Next, we transform the 
last factor of \eqref{eqDisc:ProofLemmaTheorem3ApproximationFormulaProof} using \eqref{eqDisc:TX_TU} as:
\begin{align}
    \begin{pmatrix}
        \left( \frac{1}{n^{1/2}}\right)
        \spacednablatranspose{\omega}{T_U
        \left(
        G_{\theta}(1/n^{1/2}, 0), 
        u_n
        \right)} 
        \\
        \spacednablatranspose{\tau}{T_U
        \left(
        G_{\theta}(1/n^{1/2}, 0), 
        u_n
        \right)}
    \end{pmatrix}
    &=
    \begin{pmatrix}
        n^{1/2}I_{p_1 \times p_1} & 0 \\
        0 & n I_{p_2 \times p_2}
    \end{pmatrix}
    \begin{pmatrix}
        \spacednablatranspose{\omega}{T_X(\hat{\theta}_{\mathrm{spa},n}, x_n)}
        \\
        \spacednablatranspose{\tau}{T_X(\hat{\theta}_{\mathrm{spa},n}, x_n)}
    \end{pmatrix}
    \notag
    \\
    &=
    \begin{pmatrix}
        n^{1/2} I_{p_1 \times p_1} & 0 \\
        0 & nI_{p_2 \times p_2}
    \end{pmatrix}
    \spacednablatranspose{\theta}{T_X(\hat{\theta}_{\mathrm{spa},n}, x_n)} 
    .
    \label{eqDisc:Theorem3T_UasT_X}
\end{align}
Combining \eqref{eqDisc:ProofLemmaTheorem3ApproximationFormulaProof}, \eqref{eqDisc:F_as_log_x} and \eqref{eqDisc:Theorem3T_UasT_X}, 
\begin{equation*}
    \begin{pmatrix}
        \left( \frac{1}{n^{3/2}} \right)   
        \frac{\partial G_{v}}{\partial \rho} ( 1/n^{1/2}, 0) 
        \\
        \left( \frac{1}{n} \right)  
        \frac{\partial G_{\tau}}{\partial \rho} ( 1/n^{1/2}, 0)
    \end{pmatrix}
    =
     -\left(
        \nabla_{\theta}^T \nabla_{\theta} \log \hat{L}(\hat{\theta}_{\mathrm{spa},n}, x_n)
     \right)^{-1}
     \spacednablatranspose{\theta}{T_X(\hat{\theta}_{\mathrm{spa},n}, x_n)}
     , 
\end{equation*}
which coincides with the discrepancy approximation as defined in \eqref{eqDisc:DiscrepancyFormula}.
\end{proof}

\subsection{Proof of Lemma \ref{lemmaDisc:SecondDeriv}}
\label{lemmaDisc:SecondDerivTheorem3}

\begin{proof}[\proofof{Lemma~\ref{lemmaDisc:SecondDeriv}}]
We begin by simplifying \eqref{eqDisc:FirstOrderImplicitPartial} for the case when $\epsilon=0$.
From
\eqref{eqDisc:S_at_epsilon_0} and 
\eqref{eqDisc:F_omega_at_epsilon_rho_equal_0},
we see that $v=G_v(0,\rho)$ and $\sigma=\hat{\sigma}(G_v(0,\rho), G_\tau(0,\rho), 0)$ satisfy the system of linear equations
\begin{equation}
\begin{aligned}
  \tilde{A} \sigma^T + \tilde{B} v &= z_0
  \\
  \tilde{B}^T \sigma^T \phantom{+\tilde{B}v} &= 0
\end{aligned}
\end{equation}
where we have used the abbreviations $\tilde{A}$ for the $d\times d$ matrix $\tilde{A} = K_U''\left(0;\smallmat{\omega_0\\ G_\tau(0,\rho)}\right)$ and $\tilde{B}$ for the $d\times p_1$ matrix $\tilde{B}=\grad_\omega K'_U\left(0;\smallmat{\omega_0\\ G_\tau(0,\rho)}\right)$.
After one block row-reduction step, $\sigma^T$ can be eliminated from the second block equation to find
  \begin{equation}
    \tilde{B}^T \tilde{A}^{-1} \tilde{B} v = \tilde{B}^T \tilde{A}^{-1} z_0
  \end{equation}
  and thus
  \begin{align}
    v &= (\tilde{B}^T \tilde{A}^{-1}\tilde{B})^{-1}\tilde{B}^T\tilde{A}^{-1}z_0
    \\
    \sigma^T &= \tilde{A}^{-1}z_0 - \tilde{A}^{-1} \tilde{B} (\tilde{B}^T \tilde{A}^{-1}\tilde{B})^{-1}\tilde{B}^T\tilde{A}^{-1}z_0
    .
  \end{align}
In other words, with the further abbreviations 
$\tilde{Q}=\tilde{A}^{-1}=K_U''\left(0;\smallmat{\omega_0\\ G_\tau(0,\rho)}\right)^{-1}$ 
and 
$\tilde{J}=\tilde{A}^{-1} - \tilde{A}^{-1} \tilde{B} (\tilde{B}^T \tilde{A}^{-1}\tilde{B})^{-1}\tilde{B}^T\tilde{A}^{-1}$ 
as in the assumptions of Theorem~\ref{theoremDisc:Theorem3}(b), we have the identities
  \begin{equation}
    \begin{aligned}
      G_v(0,\rho) &= (\tilde{B}^T \tilde{Q}\tilde{B})^{-1}\tilde{B}^T\tilde{Q} z_0
      ,
      \\
      \hat{\sigma}(G_v(0,\rho), G_\tau(0,\rho), 0) &= \tilde{J}z_0
      .
    \end{aligned}
  \end{equation}
Combining this with \eqref{eqDisc:del_tau_F_omega}, it follows that the $j^\text{th}$ column of $\grad_\tau F_\omega(G_v(0,\rho),G_\tau(0,\rho); 0, \rho)$ is
  \begin{equation}
    \frac{\partial F_\omega}{\partial\tau_j}(G_v(0,\rho),G_\tau(0,\rho); 0, \rho) = -\tilde{B}^T\tilde{Q} \tfrac{\partial K_U''}{\partial\tau_j}\left(0;\smallmat{\omega_0\\ G_\tau(0,\rho)}\right) \tilde{J} z_0
    .
  \end{equation}
But by the assumption \eqref{eqDisc:ZeroMatrixConditionTheorem3}, the matrix product on the right-hand side vanishes identically.
Thus, under the assumption \eqref{eqDisc:ZeroMatrixConditionTheorem3},
  \begin{equation}\label{eqDisc:gradtauFomegaVanishes}
    \grad_\tau F_\omega (G_v(0,\rho),G_\tau(0,\rho); 0, \rho) = 0
  \end{equation}
  for all $\rho$.
  
We now set $\epsilon=0$ in \eqref{eqDisc:FirstOrderImplicitPartial}.
  Note from \eqref{eqDisc:F_omega_at_epsilon_rho_equal_0} that $F_\omega(v,\tau;0,\rho)$ does not depend on $\rho$, so the $\frac{\partial F}{\partial\tau}$ term from \eqref{eqDisc:FirstOrderImplicitPartial} vanishes.
  Using \eqref{eqDisc:gradtauFomegaVanishes}, the $\grad_\tau F_\omega$ term vanishes also, and we obtain the identity
  \begin{equation}
    \grad_v F_\omega(G(0,\rho);0,\rho) \frac{\partial G_v}{\partial\rho}(0,\rho)= 0
    ,
  \end{equation}
  where we have kept only the top $p_1\times 1$ block.
  By \eqref{eqDisc:del_vF_omega} and the fact that $B=\grad_\omega K_U'(0;\theta_0)$ has rank $p_1$, the $p_1\times p_1$ matrix $\grad_v F_\omega(G(0,\rho);0,\rho)$ is non-singular for sufficiently small $\rho$ and it follows that $\frac{\partial G_v}{\partial\rho}(0,\rho)= 0$ for all sufficiently small $\rho$.
  
  In particular, $\frac{\partial^2 G_v}{\partial\rho^2}(0,0)= 0$.
  Since $\frac{\partial^2 G_v}{\partial\rho^2}$ is continuously differentiable by 
\lemmaitemref{lemmaDisc:Theorem3ContinousDiff}{lemmaItem_dG_drho_continuous_diff_Theorem3}
  it follows that $\frac{\partial^2 G_v}{\partial\rho^2}(\epsilon,\rho)=O(\epsilon+\rho)$, as claimed.
\end{proof}

\subsection{Exactness of the discrepancy approximation in exponential families}
\label{sec:supp:exactness-expfam}

In this section we show that the discrepancy approximation $\hat\delta$ from Definition~2
vanishes identically when the parametric model is an exponential family with sufficient statistic
$X$.
To match common notation, we replace $\theta$ by the natural parameter $\eta$.

\begin{prop}[Exactness of $\hat\delta$ for exponential families]
\label{prop:expfam-zero-delta-hat}
Suppose the observable statistic $X_\eta$ has density
\[
  f(x \mid \eta)
  \;=\;
  h(x)\,\exp\!\bigl(\eta x - \rho(\eta)\bigr),
\]
so that $\{f(\,\cdot\mid\eta)\}$ is an exponential family with natural parameter
$\eta$ and sufficient statistic $X_\eta$.
Let $\hat L(\eta \mid x)$ denote the saddlepoint likelihood constructed from the CGF of
$X_\eta$, and suppose that the true and saddlepoint MLEs exist. 
Then the true and approximated discrepancies satisfy
\[
  \delta = \hat\delta = 0.
\]
\end{prop}
In particular, in this exponential-family setting the diagnostic correctly returns that the saddlepoint MLE is exact.
We note that a sufficient condition guaranteeing the existence of both true and saddlepoint MLEs is for $h(x)$ to be non-zero and for the saddlepoint equation to have a solution.

\begin{proof}[of Proposition~\ref{prop:expfam-zero-delta-hat}]
The fact that $\delta = 0$ is already shown in \citet[Section 2.5.6]{Goodman2022}, so it suffices to show $\hat \delta =0$.
By construction of the exponential family, the CGF of $X_\eta$
can be written as
\[
  K_X(t;\eta)
  =
  \rho\bigl(\eta + t\bigr) - \rho\bigl(\eta\bigr).
\]
Note first that 
\[
    K_X'(t;\eta) =
    \rho'(\eta + t)
    .
\]
$K_X'$ and its higher-order $t$-derivatives depend only on $\eta + t$.
From this we see that the solution of the saddlepoint equation 
\[
    \rho'(\eta + \hat{t}(\eta; x)) = x
\]
will satisfy:
\[
    \eta + \hat{t}(\eta; x)
\]
is constant in $\eta$.
$T_X(\eta, x)$ is constructed out of higher-order $t$-derivatives of $K_X'$ evaluated at $t = \hat{t}(\eta;x)$.
It too will be constant in $\eta$.
Therefore $\grad_{\eta} T_X(\eta,x) = 0$, and thus $\hat{\delta} = 0$.
\end{proof}

\section{Evaluating \texorpdfstring{$g_{X,\mathrm{factored}}$}{g\_Xfactored} via higher-order directional-derivative operators}
\label{labelDisc:K3K4operators}

Throughout this section, we assume that 
\begin{equation}
  Q = BDB^T,
\end{equation}
we abbreviate
\begin{equation}
  K_{i_1\cdots i_m}
  =
  \frac{\partial^m K_X(t;\theta)}{\partial t_{i_1}\cdots \partial t_{i_m}}
  ,
\end{equation}
and we omit some function arguments.
Thus we may summarize \eqref{eqDisc:TXagain} and \eqref{eqDisc:gXfactored} by writing
\begin{equation}
\label{eq:g_decomp_patterns_mult}
g_X = g_{X,\mathrm{factored}} =
\frac18 K_4^{AABB}
-
\frac18 K_{33}^{AABBCC}
-
\frac1{12} K_{33}^{ABCABC}
,
\end{equation}
where
\begin{align}
\label{eq:K4_AABB_def_mult}
K_4^{AABB}
&=
\sum_{i_1,i_2,j_1,j_2=1}^d
Q_{i_1 i_2}Q_{j_1 j_2}
K_{i_1 i_2 j_1 j_2}
,
\\
\label{eq:K33_AABBCC_def_mult}
K_{33}^{AABBCC}
&=
\sum_{i_1,i_2,i_3=1}^d
\sum_{j_1,j_2,j_3=1}^d
Q_{i_1 i_2}Q_{i_3 j_1}Q_{j_2 j_3}
K_{i_1 i_2 i_3} K_{j_1 j_2 j_3}
,
\\
\label{eq:K33_ABCABC_def_mult}
K_{33}^{ABCABC}
&=
\sum_{i_1,i_2,i_3=1}^d
\sum_{j_1,j_2,j_3=1}^d
Q_{i_1 j_1}Q_{i_2 j_2}Q_{i_3 j_3}
K_{i_1 i_2 i_3} K_{j_1 j_2 j_3}
.
\end{align}
We remark that in practice, our code implements functions corresponding to the quantities $K_4^{AABB}$, $K_{33}^{AABBCC}$, $K_{33}^{ABCABC}$, with factored and unfactored variants as in \eqref{eqDisc:gXfactored}, rather than the functions $g_X$ and $g_{X,\mathrm{factored}}$.
The relations noted in \eqref{eqDisc:g_X-g_Y-linear_map}, \eqref{eqDisc:gX-gY-factored} apply similarly to $K_4^{AABB}$, $K_{33}^{AABBCC}$, $K_{33}^{ABCABC}$.

Define the following scalar-valued ``operator-style'' forms of the third- and fourth-order derivatives,
\begin{align}
K^{(3)}_{\mathrm{op},X}(t,\theta,u^{(1)},u^{(2)},u^{(3)})
&=
\sum_{i_1,i_2,i_3=1}^d
\frac{\partial^3 K_X(t;\theta)}
{\partial t_{i_1}\partial t_{i_2}\partial t_{i_3}}
u^{(1)}_{i_1}u^{(2)}_{i_2}u^{(3)}_{i_3}
,
\label{eqDisc:K3operator}
\\
K^{(4)}_{\mathrm{op},X}(t,\theta,u^{(1)},u^{(2)},u^{(3)},u^{(4)})
&=
\sum_{i_1,i_2,i_3,i_4=1}^d
\frac{\partial^4 K_X(t;\theta)}
{\partial t_{i_1}\partial t_{i_2}\partial t_{i_3}\partial t_{i_4}}
u^{(1)}_{i_1}u^{(2)}_{i_2}u^{(3)}_{i_3}u^{(4)}_{i_4}
,
\label{eqDisc:K4operator}
\end{align}
which act multilinearly on vectors $u^{(1)},u^{(2)},u^{(3)},u^{(4)}\in\R^d$.
By symmetry of mixed partial derivatives, the functions $K^{(3)}_{\mathrm{op},X}$ and $K^{(4)}_{\mathrm{op},X}$ are symmetric in their vector arguments.
To make use of these functions, write the matrix factorization $Q = BDB^T$ in vector form as
\begin{equation}
  Q = \sum_{a=1}^r D_{aa}\, B^{(a)} (B^{(a)})^T
  ,
  \quad\text{i.e.,}\quad
  Q_{ij} = \sum_{a=1}^r D_{aa}\, (B^{(a)})_i (B^{(a)})_j
  ,
\end{equation}
where $B^{(a)}$ denotes the $a^\text{th}$ column of $B$.
Substituting into \eqref{eq:K4_AABB_def_mult} gives
\begin{align}
K_4^{AABB} &= 
\sum_{a=1}^r \sum_{b=1}^r
D_{aa}D_{bb}
\sum_{i_1,i_2,j_1,j_2=1}^d
(B^{(a)})_{i_1}(B^{(a)})_{i_2}(B^{(b)})_{j_1}(B^{(b)})_{j_2}
K_{i_1 i_2 j_1 j_2}
\notag
\\
&=
\sum_{a=1}^r \sum_{b=1}^r
D_{aa}D_{bb}
K^{(4)}_{\mathrm{op},X}(t,\theta,B^{(a)},B^{(a)},B^{(b)},B^{(b)}).
\label{eqDisc:factored_first_term}
\end{align}
Similarly
\begin{align}
K_{33}^{AABBCC} &=
\sum_{a=1}^r \sum_{b=1}^r \sum_{c=1}^r
D_{aa}D_{bb}D_{cc}
K^{(3)}_{\mathrm{op},X}(t,\theta,B^{(a)},B^{(a)},B^{(c)})
K^{(3)}_{\mathrm{op},X}(t,\theta,B^{(c)},B^{(b)},B^{(b)})
\notag\\&
= \sum_{c=1}^r
D_{cc}
\left(
\sum_{a=1}^r
D_{aa}
K^{(3)}_{\mathrm{op},X}(t,\theta,B^{(a)},B^{(a)},B^{(c)})
\right)^2,
\label{eqDisc:factored_second_term}
\end{align}
where the last step uses the symmetry of $K^{(3)}_{\mathrm{op},X}$, and
\begin{equation}
\label{eqDisc:factored_third_term}
K_{33}^{ABCABC}
=
\sum_{a=1}^r \sum_{b=1}^r \sum_{c=1}^r
D_{aa}D_{bb}D_{cc}
\Bigl(
K^{(3)}_{\mathrm{op},X}(t,\theta,B^{(a)},B^{(b)},B^{(c)})
\Bigr)^2.
\end{equation}
Combining \eqref{eqDisc:factored_first_term}--\eqref{eqDisc:factored_third_term}, we obtain
\begin{equation}
\label{eqDisc:gXfactoredWithK3K4}
\begin{aligned}
g_{X,\mathrm{factored}}(t,\theta,B,D)
&=
\frac18
\sum_{a=1}^r \sum_{b=1}^r
D_{aa}D_{bb}
K^{(4)}_{\mathrm{op},X}(t,\theta,B^{(a)},B^{(a)},B^{(b)},B^{(b)})
\\
&\quad
-
\frac18
\sum_{c=1}^r
D_{cc}
\left(
\sum_{a=1}^r
D_{aa}
K^{(3)}_{\mathrm{op},X}(t,\theta,B^{(a)},B^{(a)},B^{(c)})
\right)^2
\\
&\quad
-
\frac1{12}
\sum_{a=1}^r \sum_{b=1}^r \sum_{c=1}^r
D_{aa}D_{bb}D_{cc}
\Bigl(
K^{(3)}_{\mathrm{op},X}(t,\theta,B^{(a)},B^{(b)},B^{(c)})
\Bigr)^2.
\end{aligned}
\end{equation}

Equation \eqref{eqDisc:gXfactoredWithK3K4} evaluates $g_{X,\mathrm{factored}}$ using at most $O(r^3)$ evaluations of the scalar-valued functions $K^{(3)}_{\mathrm{op},X}$ and $K^{(4)}_{\mathrm{op},X}$, which do not need to compute the full third- or fourth-order derivative arrays.
In turn, the function $g_X(t,\theta,Q)$ may be evaluated by computing an $\mathrm{LDL}^T$ decomposition of the matrix $Q$.
In particular, this allows us to compute $T_X$ with reasonable efficiency, using only minimal derivative information about the CGF $K_X$.

In specific cases, we can find more efficient implementations of the functions $g_X$, $g_{X,\mathrm{factored}}$.
The next section illustrates this for a multivariate distribution.

\section{Evaluating \texorpdfstring{$g_X$}{g\_X} and \texorpdfstring{$g_{X,\mathrm{factored}}$}{g\_Xfactored} in the multinomial case}
\label{supsec:multinomialAABBCC}

This section shows how the function $g_X(t,\theta,Q)$ in \eqref{eqDisc:TXagain} simplifies when $X$ is multinomial of dimension $d$.
In contrast to \supsecref{labelDisc:K3K4operators}, where simplification comes from a factorization of $Q$, the simplification here comes from the derivative structure of the multinomial CGF itself.
The resulting formulas depend on a small collection of scalar summaries of $v(t;\theta)$ and $Q$, and yield an $O(d^2)$ evaluation of $g_X(t,\theta,Q)$ without forming third- or fourth-order
derivative arrays.

Let $X=(X_1,\ldots,X_d)^T$ be multinomial with parameters $N\in\mathbb Z_+$ and
$\tilde p=(p_1,\ldots,p_d)^T\in[0,1]^d$, collectively $\theta=(N,\tilde p)^T$, with
$\sum_{i=1}^d p_i=1$. Its CGF is
\begin{equation*}
K_X(t;\theta)
=
N \log  \left(\sum_{i=1}^d p_i \exp(t_i)\right),
\qquad t\in\mathbb R^{1\times d}
.
\end{equation*}
We define
\begin{equation}
v_i(t;\theta)
=
\frac{p_i \exp(t_i)}{\sum_{m=1}^d p_m \exp(t_m)},
\qquad
v(t;\theta)=(v_1(t;\theta),\ldots,v_d(t;\theta))^T,
\qquad
\sum_{i=1}^d v_i(t;\theta)=1.
\end{equation}
Then 
\begin{equation}
  \frac{\partial K_X(t;\theta)}{\partial t_i} = N v_i(t;\theta)
  .
\end{equation}
To compute higher-order derivatives systematically, note that
\begin{equation}\label{eq:mult_v_derivative}
  \frac{\partial v_i(t;\theta)}{\partial t_j}
  =
  v_i(t;\theta)\bigl(\indicator{i=j}-v_j(t;\theta)\bigr)
  ;
\end{equation}
applying this identity repeatedly leads to
\begin{equation}
\label{eq:mult_second_derivative}
\frac{\partial^2 K_X(t;\theta)}
{\partial t_{i_1}\partial t_{i_2}}
=
N\Bigl(
v_{i_1}(t;\theta)\indicator{i_1=i_2}
-
v_{i_1}(t;\theta)v_{i_2}(t;\theta)
\Bigr)
,
\end{equation}
\begin{equation}
\label{eq:mult_third_derivative}
\begin{aligned}
\frac{\partial^3 K_X(t;\theta)}
{\partial t_{i_1}\partial t_{i_2} \partial t_{i_3} }
=
N\Bigl(
&v_{i_1}(t;\theta)\indicator{i_1=i_2=i_3}
-
v_{i_1}(t;\theta)v_{i_3}(t;\theta)\indicator{i_1=i_2}
-
v_{i_1}(t;\theta)v_{i_2}(t;\theta)\indicator{i_1=i_3}
\\
&-
v_{i_1}(t;\theta)v_{i_2}(t;\theta)\indicator{i_2=i_3}
+
2v_{i_1}(t;\theta)v_{i_2}(t;\theta)v_{i_3}(t;\theta)
\Bigr)
,
\end{aligned}
\end{equation}
\begin{align}
\frac{\partial^4 K_X(t;\theta)}
{\partial t_{i_1}\partial t_{i_2} \partial t_{i_4} \partial t_{i_4} }
&=
N\Bigl(
v_{i_1}(t;\theta)\indicator{i_1=i_2=i_3=i_4}
\notag\\&\qquad\quad
-v_{i_1}(t;\theta)v_{i_4}(t;\theta)\indicator{i_1=i_2=i_3}
-v_{i_1}(t;\theta)v_{i_3}(t;\theta)\indicator{i_1=i_2=i_4}
\notag\\&\qquad\quad
-v_{i_1}(t;\theta)v_{i_2}(t;\theta)\indicator{i_1=i_3=i_4}
-v_{i_1}(t;\theta)v_{i_2}(t;\theta)\indicator{i_2=i_3=i_4}
\notag\\&\qquad\quad
-v_{i_1}(t;\theta)v_{i_3}(t;\theta)\indicator{i_1=i_2}\indicator{i_3=i_4}
-v_{i_1}(t;\theta)v_{i_2}(t;\theta)\indicator{i_1=i_3}\indicator{i_2=i_4}
\notag\\&\qquad\qquad\quad
-v_{i_1}(t;\theta)v_{i_2}(t;\theta)\indicator{i_1=i_4}\indicator{i_2=i_3}
\notag\\&\qquad\quad
+2v_{i_1}(t;\theta)v_{i_3}(t;\theta)v_{i_4}(t;\theta)\indicator{i_1=i_2}
+2v_{i_1}(t;\theta)v_{i_2}(t;\theta)v_{i_4}(t;\theta)\indicator{i_1=i_3}
\notag\\&\qquad\quad
+2v_{i_1}(t;\theta)v_{i_2}(t;\theta)v_{i_3}(t;\theta)\indicator{i_1=i_4}
+2v_{i_1}(t;\theta)v_{i_2}(t;\theta)v_{i_4}(t;\theta)\indicator{i_2=i_3}
\notag\\&\qquad\quad
+2v_{i_1}(t;\theta)v_{i_2}(t;\theta)v_{i_3}(t;\theta)\indicator{i_2=i_4}
+2v_{i_1}(t;\theta)v_{i_2}(t;\theta)v_{i_3}(t;\theta)\indicator{i_3=i_4}
\notag\\&\qquad\quad
-6v_{i_1}(t;\theta)v_{i_2}(t;\theta)v_{i_3}(t;\theta)v_{i_4}(t;\theta)
\Bigr).
\label{eq:mult_fourth_derivative}
\end{align}

Define the vectors
\begin{equation}
\label{eq:mult_q_s}
q = Qv
\qquad\text{and}\qquad
w=(w_i)_{i=1}^d,
\;\;
w_i = v_i(Q_{ii}-2q_i),
\end{equation}
and the scalars
\begin{equation}\label{eq:mult_scalar_summaries}
\begin{gathered}
s = v^TQv = v^Tq,
\quad
c_{\mathrm{diag},1}
=
\sum_{i=1}^d v_i Q_{ii},
\quad
c_{\mathrm{diag},2}
=
\sum_{i=1}^d v_i Q_{ii}^2,
\\
c_{Q^2}
=
\sum_{i=1}^d \sum_{j=1}^d v_i v_j Q_{ij}^2,
\quad
c_{q^2}
=
\sum_{i=1}^d v_i q_i^2,
\quad
c_{\mathrm{diag},q}
=
\sum_{i=1}^d v_i Q_{ii} q_i.
\end{gathered}
\end{equation}
Then, recalling the abbreviations in \eqref{eq:K4_AABB_def_mult}--\eqref{eq:K33_ABCABC_def_mult}, we find that
\begin{align}
\label{eq:mult_K4_closed}
K_4^{AABB}
&=
N\Bigl(
-2c_{Q^2}
+ c_{\mathrm{diag},2}
- 4c_{\mathrm{diag},q}
- c_{\mathrm{diag},1}^2
+ 4s\,c_{\mathrm{diag},1}
+ 8c_{q^2}
- 6s^2
\Bigr)
,
\\
\label{eq:mult_K33_AABBCC_closed}
K_{33}^{AABBCC}
&=
N^2\Bigl(
w^TQw
+
2(c_{\mathrm{diag},1}-2s)(2c_{q^2}-c_{\mathrm{diag},q})
+
s(c_{\mathrm{diag},1}-2s)^2
\Bigr)
,
\\
\label{eq:mult_K33_ABCABC_closed}
K_{33}^{ABCABC}
&=
N^2
\sum_{i=1}^d \sum_{j=1}^d
v_i v_j \bigl(Q_{ij}-q_i-q_j+s\bigr)^3.
\end{align}

To indicate how \eqref{eq:mult_K4_closed}--\eqref{eq:mult_K33_ABCABC_closed} arise, we consider the contribution to \eqref{eq:K33_ABCABC_def_mult} obtained by taking the first term
$N v_{i_1}(t;\theta)\indicator{i_1=i_2=i_3}$ from
\eqref{eq:mult_third_derivative} for the factor $\frac{\partial^3 K_X(t;\theta)}{\partial t_{i_1}\partial t_{i_2} \partial t_{i_3}}$, and the analogous
first term $N v_{j_1}(t;\theta)\indicator{j_1=j_2=j_3}$ from the second copy.
This contribution is
\begin{equation*}
\begin{aligned}
&N^2
\sum_{i_1,i_2,i_3=1}^d
\sum_{j_1,j_2,j_3=1}^d
Q_{i_1 j_1}Q_{i_2 j_2}Q_{i_3 j_3}
v_{i_1}(t;\theta)v_{j_1}(t;\theta)
\indicator{i_1=i_2=i_3}\indicator{j_1=j_2=j_3}
\\
&\qquad=
N^2\sum_{i=1}^d\sum_{j=1}^d
v_i(t;\theta)v_j(t;\theta)Q_{ij}^3.
\end{aligned}
\end{equation*}
This is exactly the $Q_{ij}^3$ contribution in the expansion of
\(\bigl(Q_{ij}-q_i-q_j+s\bigr)^3\) in \eqref{eq:mult_K33_ABCABC_closed}.
The remaining pairings of terms from the two copies of \eqref{eq:mult_third_derivative}
yield the mixed contributions involving $q_i$, $q_j$, and $s$, and by analogous calculations, we obtain
the closed forms in \eqref{eq:mult_K4_closed} and
\eqref{eq:mult_K33_ABCABC_closed}.

Combining \eqref{eq:g_decomp_patterns_mult} and \eqref{eq:mult_K4_closed}--\eqref{eq:mult_K33_ABCABC_closed} yields
\begin{equation}
\label{eq:mult_gx_closed}
\begin{aligned}
g_X(t,\theta,Q)
&=
\frac{N}{8}\Bigl(
-2c_{Q^2}
+ c_{\mathrm{diag},2}
- 4c_{\mathrm{diag},q}
- c_{\mathrm{diag},1}^2
+ 4s\,c_{\mathrm{diag},1}
+ 8c_{q^2}
- 6s^2
\Bigr)
\\
&\quad
-
\frac{N^2}{8}\Bigl(
w^TQw
+
2(c_{\mathrm{diag},1}-2s)(2c_{q^2}-c_{\mathrm{diag},q})
+
s(c_{\mathrm{diag},1}-2s)^2
\Bigr)
\\
&\quad
-
\frac{N^2}{12}
\sum_{i=1}^d \sum_{j=1}^d
v_i v_j \bigl(Q_{ij}-q_i-q_j+s\bigr)^3.
\end{aligned}
\end{equation}
This evaluation uses only $v(t;\theta)$, the matrix--vector products $q=Qv$ and $Qw$,
and $O(d^2)$ scalar operations.
No third- or fourth-order derivative arrays are formed, and no explicit six-fold summations appear.
The overall cost is therefore $O(d^2)$.

For the factored form $g_{X,\mathrm{factored}}$, suppose $Q=BDB^T$ is given where $B$ is a $d\times r$ matrix and $D$ is an $r\times r$ diagonal matrix.
We will obtain a formula for $g_{X,\mathrm{factored}}(t,\theta,B,D)$ that can be computed in $O(r^3 d)$ steps, and in particular without forming all $d^2$ entries $Q_{ij}$, by obtaining expressions for the quantities in \eqref{eq:mult_q_s}--\eqref{eq:mult_scalar_summaries} and \eqref{eq:mult_K33_AABBCC_closed}--\eqref{eq:mult_K33_ABCABC_closed} in terms of $B$ and $D$.
This represents a gain if $d$ is sufficiently large compared to $r$.

Let $u,m$ be the vectors of lengths $r, d$, respectively, defined by
\begin{equation}
u=B^Tv,
\qquad
m_i=\sum_{k=1}^r D_{kk}B_{ik}^2,
\qquad i=1,\ldots,d
.
\end{equation}
We note that $m$ is precisely the diagonal vector of the matrix $Q$.
In terms of $u$ and $m$ we can now compute
\begin{equation}
q=BDu,
\qquad
s=u^T D u,
\qquad
w_i=v_i(m_i-2q_i),
\qquad
w^TQw=(B^T w)^T D (B^T w),
\end{equation}
in $O(rd)$ steps, and similarly for the scalars $c_{\mathrm{diag},1}$, $c_{\mathrm{diag},2}$, $c_{\mathrm{diag},q}$, and $c_{q^2}$ in terms of $m$ and $q$.
For $c_{Q^2}$, let $M$ be the $r\times r$ matrix $B^T V B$, where $V$ means the diagonal matrix with diagonal entries \(v_1,\ldots,v_d\); we note that $M$ can be computed in $O(r^2 d)$ steps.
Then
\begin{equation}
\label{eq:mult_factored_cQ2}
c_{Q^2}
=\operatorname{trace}(QVQV)
=
\operatorname{trace}(MDMD) = \sum_{k=1}^r \sum_{\ell=1}^r M_{k\ell}D_{\ell\ell}M_{\ell k}D_{kk}
.
\end{equation}
Finally, for \eqref{eq:mult_K33_ABCABC_closed}, let $\widetilde{Q}$ denote the $d\times d$ matrix with entries $Q_{ij}-q_i-q_j+s$.
Then $\widetilde{Q}$ has a factorization $\widetilde{Q}=\widetilde{B} D \widetilde{B}^T$ where
\begin{equation}
\widetilde B_{ik}=B_{ik}-u_k
\end{equation}
and it can be shown that \eqref{eq:mult_K33_ABCABC_closed} simplifies to
\begin{equation}
K_{33}^{ABCABC}
= N^2 
\sum_{k_1,k_2,k_3=1}^r
D_{k_1k_1}D_{k_2k_2}D_{k_3k_3}
\left(
\sum_{i=1}^d
v_i \widetilde B_{i k_1}\widetilde B_{i k_2}\widetilde B_{i k_3}
\right)^2
,
\end{equation}
leading to the cost $O(r^3 d)$ claimed above.

\end{document}